\documentclass{article} % For LaTeX2e
\usepackage{iclr2027_conference,times}

\iclrfinalcopy

\usepackage[utf8]{inputenc}
\usepackage[T1]{fontenc}
\usepackage{hyperref}
\usepackage{url}
\usepackage{xcolor}
\usepackage{booktabs}
\usepackage{amsfonts}
\usepackage{nicefrac}
\usepackage{microtype}
\usepackage{caption}

\usepackage{graphicx}
\usepackage{float}
\usepackage{comment}
\usepackage{enumitem}
\usepackage{subfigure}
\usepackage{tikz}
\usepackage{multirow}
\usepackage{placeins}

\usepackage{amsmath,amssymb,amsthm,mathtools,mathrsfs}
\usepackage{bbm}
\usepackage{bm}
\usepackage{dsfont}

\usepackage[capitalize,noabbrev]{cleveref}
\usepackage[textsize=tiny]{todonotes}
\usepackage[ruled,vlined]{algorithm2e}

\theoremstyle{plain}
\newtheorem{theorem}{Theorem}[section]
\newtheorem{proposition}[theorem]{Proposition}

\newtheorem{remark}[theorem]{Remark}

\theoremstyle{definition}
\newtheorem{definition}[theorem]{Definition}

\title{Watermarkable Multi-Draft Speculative \\ Sampling via Poisson Processes}

\author{
Yanxiao Liu$^{1,\star}$, Sicheng Wan$^{2,\star}$, Zhan Gao$^1$, and Deniz G\"und\"uz$^1$\\
$^1$Imperial College London,
$^2$University of Washington, $^\star$equal contribution\\
}

\begin{document}

\maketitle
\fancyhead{}

\begin{abstract}
Large language models (LLMs) have achieved state-of-the-art performance across a wide range of tasks, motivating two important aspects of deployment: inference efficiency and output provenance, which can be tackled by speculative sampling and watermarking, respectively.
However, recent works have shown that combining these two goals is highly nontrivial and can be potentially impossible. 
\textcolor{black}{In this work, we develop a novel multi-draft speculative sampling algorithm
based on Poisson processes that improves the frontier of this fundamental trade-off.}
The proposed algorithm has strong sampling efficiency on its own and, more interestingly, is naturally \emph{watermarkable}: we can embed an unbiased watermark without degrading speculative acceptance. 
Moreover, our algorithm is based on an \emph{exact} list-coupling-without-communication scheme, which yields a \emph{drafter invariance} property that benefits both sampling and watermarking. 
It is the first multi-draft, drafter-invariant speculative sampling scheme that maintains both watermark strength and sampling efficiency, and we experimentally verify its strong performance in both aspects.
\end{abstract}

\section{Introduction}
\label{sec::intro}

\textcolor{black}{Large language models (LLMs) have become a central component of modern
generative AI, achieving strong performance in dialogue, reasoning, code generation,
and open-ended content creation. Their probabilistic nature based on autoregressive sampling enables them to generate diverse responses~\citep{brown2020language, grattafiori2024llama, singh2025openai}.}
However, as LLM-generated text becomes increasingly fluent and difficult to distinguish from human-written content, concerns about accountability, provenance, misuse, and intellectual property have attracted increasing attention~\citep{wu2025survey}. 
For instance, ChatGPT-generated scientific abstracts are hard for human reviewers to distinguish from human-generated ones, yet they may contain fabricated data~\citep{shumailov2024ai}. 
Hence, watermarking has become an important, and potentially necessary, tool for tracing the provenance of LLM-generated text by embedding statistical signals into the sampling process while preserving output quality~\citep{kirchenbauer2023watermark, aaronson2023watermarking, kuditipudi2023robust}. 
\textcolor{black}{For instance, reflecting its growing importance, Anthropic recently announced that future Claude models will generate text with embedded watermarks~\citep{anthropic2026watermark}. 
}

Apart from provenance, the efficiency of autoregressive generation is another bottleneck in the system pipeline. 
\emph{Speculative sampling} (also known as decoding), has emerged as a promising approach to accelerate content generation without changing the output distribution of the target model~\citep{chen2023accelerating, leviathan2023fast}. 
However, speculative sampling is not independent of watermarking.
On the contrary, recent works~\citep{hu2024inevitable, he2026improving} show that combining the two is highly nontrivial: in~\citep{hu2024inevitable}, a ``no-go'' theorem was proved, showing that one cannot embed a watermark without degrading the performance of speculative sampling, thereby establishing a fundamental trade-off between watermark strength and speculative sampling efficiency. 
\citet{he2026improving} recently improves this trade-off by considering alternative watermark metrics.

{\color{black}
In this work, we ask the following question: can we design a watermarkable speculative sampling scheme that further improves the trade-off~\citep{hu2024inevitable, he2026improving} by increasing speculative sampling efficiency without compromising watermark detectability?
}

We answer this question in the \textbf{affirmative}: we design a novel multi-draft speculative sampling algorithm that achieves strong acceleration performance, can naturally embed unbiased watermark, and, more importantly, has watermark strength that does not degrade as the sampling efficiency increases.\footnote{\textcolor{black}{Our multi-draft scheme does not contradict the inevitable trade-off~\citep{hu2024inevitable}: it still applies for a fixed number of drafts.
We improve the trade-off along an additional dimension, namely, the number of drafts.}} 
Our multi-draft construction relies on Poisson processes~\citep{maddison2016poisson, li2021unified} and yields a \emph{list-coupling without communication} scheme, which endows our algorithm with a \emph{drafter-invariance} property~\citep{daliri2025coupling, rowan2025list} that becomes more desirable and essential in the context of watermarking. 
Our contributions are two-fold: 
\begin{itemize}
    \item \textbf{Multi-draft speculative sampling}: We design a novel multi-draft speculative sampling algorithm that generates and verifies multiple draft continuations. 
    Our scheme uses desirable properties of Poisson processes, which naturally yield a drafter invariance property~\citep{daliri2025coupling}. 
    \textcolor{black}{Compared with~\citep{rowan2025list}, which has been shown to outperform other existing methods, we theoretically prove that our acceptance rates are higher and empirically demonstrate better acceleration performance across multiple models and datasets. }

    \item \textbf{Efficiency-preserving watermarking}: Besides desirable sampling efficiency, our algorithm is naturally \emph{watermarkable} with %by
    the use of Poisson processes, which can embed a recoverable statistical signal while maintaining the distribution \emph{unbiased}.
    By building the watermark directly into the sampling procedure through keyed randomness, we maintain watermark strength while preserving sampling efficiency, \textcolor{black}{thereby improving the trade-off frontier identified in~\citep{hu2024inevitable} and achieving better overall performance.}
\end{itemize}

\section{Related Works}
\label{sec::related}

\paragraph{LLM Decoding and Speculative Sampling.}

Standard autoregressive decoding methods, including top-$k$ sampling~\citep{radford2019language, fan2018hierarchical}, nucleus sampling~\citep{holtzman2020curious}, and the permute-and-flip decoder~\citep{zhao2024permute}, generate tokens sequentially; and therefore, suffer from the inherent latency of one target-model call per token~\citep{stern2018blockwise}. 
Speculative sampling~\citep{chen2023accelerating, leviathan2023fast} addresses this bottleneck by using a smaller draft model to propose several future tokens, while the larger target model verifies the proposed continuations in parallel. 
The efficiency gain depends on how well the draft distribution matches the target distribution. 
Recent works further improve this paradigm through optimal-transport-based token selection~\citep{sun2023spectr}, tree-based inference and verification~\citep{miao2024specinfer}, and multi-draft architectures~\citep{khisti2024multi, rowan2025list}. 
It can also be used to accelerate diffusion models \citep{de2025accelerated, bullo2026accelerating}.

\paragraph{Watermarking.}

LLM watermarking aims to embed detectable statistical signals into generated text while preserving fluency and utility. 
Early logit-bias methods pseudorandomly select green-list tokens and detect the resulting bias by hypothesis testing~\citep{kirchenbauer2023watermark}, while later works improve robustness, accessibility, and probability balance~\citep{wu2023resilient, giboulot2024watermax, park2026watermod}. 
A parallel line studies unbiased watermarking, where the sampler maintains the distribution while correlating the output with secret randomness~\citep{kuditipudi2023robust, hu2023unbiased, christ2024undetectable}. 
They connect watermarking to information theory~\citep{moulin2003information} and statistics, including detection power, false-positive control, optimal tests, and couplings~\citep{li2025statistical, he2025empirical, tsur2025optimized}. 
Related works include robust or multi-bit attribution~\citep{zhao2023provable, boroujeny2024multi, qu2025provably}, and efficient identification of watermarked segments~\citep{zhao2025efficiently}. 
In recent works~\citep{hu2024inevitable, he2026improving} an inevitable trade-off between speculative sampling and watermarking has been discussed.

\paragraph{Poisson Functional Representation (PFR).}

PFR is a sampling scheme introduced in~\citep{li2018strong}, with related Poisson process constructions also discussed by~\citet{maddison2014sampling, maddison2016poisson}. 
It selects a sample from a target distribution through a Poisson process and a proposal distribution, and has been applied to information theory settings~\citep{li2021unified, khisti2024unequal, liu2024hiding, liu2025one, liu2025nonasymptotic} and machine learning~\citep{he2024accelerating, liu2024universal, zhou2026dual, flamich2026scalable}.
It can be extended to list-coupling~\citep{li2021unified} by the mapping theorem~\citep{last2017lectures} of Poisson process, and achieves the best known bound on coupling without communication~\citep{daliri2025coupling}.

\textbf{Notations}. 
% We use uppercase and lowercase letters to denote random variables and realizations, respectively. 
$x_{i:j}$ denotes the sequence $x_i, x_{i + 1}, \ldots, x_j$, and we let $x_{:j} := x_{1:j}$. 
% $[N]$ denotes $\{1,\ldots,N\}$. 
% We assume that all mentioned random variables lie in a Polish space with Borel $\sigma$-algebra and that all functions are measurable. 
% The Lebesgue measure over $\mathbb{R}$ is denoted as $\lambda$. 
$P\ll Q$ denotes that probability measure $P$ is absolutely continuous with respect to $Q$. 
% , and $\mathrm{d}P(\cdot)/ \mathrm{d}Q$ denotes the Radon-Nikodym derivative. 
The logarithm is on base $2$.

\section{Speculative Sampling}
\label{sec::spec}

We first briefly review the technical background  of speculative sampling~\citep{chen2023accelerating,leviathan2023fast,sun2023spectr, rowan2025list}. 
Given volcabulary $\mathcal V$ and a context $x_{:t}:=(x_1,x_2, \ldots,x_t)$, an autoregressive language model $\mathcal{M}_b$ generates the next token by sampling from the conditional distribution $\mathcal{M}_b(\cdot| x_{:t})$ under temperature sampling~\citep{ackley1985learning,ficler2017controlling}. 
\emph{Speculative sampling} accelerates this process by introducing a smaller and cheaper \emph{draft model} $\mathcal{M}_s$, which proposes several candidate tokens before the target model verifies them. 
Given a prefix $x_{:t}$, one iteration of the speculative sampling consists of the following three stages:
\begin{enumerate}
    \item  \textbf{Draft Stage.} The draft model $\mathcal{M}_s$ sequentially samples $L$ candidate tokens $\tilde{x}_{t+1},\ldots,\tilde{x}_{t+L}$. For each position $i=1,\ldots,L$, we record the draft distribution $\mathcal{M}_s(\cdot| x_{:t}, \tilde{x}_{t+1:t+i-1})$. 

    \item \textbf{Computation Stage.} Conditioned on the drafted context, the target model $\mathcal{M}_b$ evaluates the corresponding conditional distributions, %as follows, 
    potentially in parallel:
    \[
    \mathcal{M}_b(\cdot| x_{:t}),
    \mathcal{M}_b(\cdot| x_{:t}, \tilde{x}_{t+1}), \ldots,
    \mathcal{M}_b(\cdot| x_{:t}, \tilde{x}_{t+1:t+L}), 
    \]
    
    \item \textbf{Selection Stage.} By the distributions $\mathcal{M}_s$ and $\mathcal{M}_b$, we accept the %a 
    longest prefix of the drafted tokens until the first rejection, and a \emph{correction token} is sampled from a specific residual distribution. If all $L$ drafts are accepted, one more token from the target is sampled. 
\end{enumerate}

The acceptance-correction procedure ensures that the generated tokens follow \emph{exactly} the same distribution as direct autoregressive sampling, i.e., improving decoding efficiency does not change the target law. 
The exactness is achieved via a recursive token-level maximal coupling (Algorithm~\ref{alg::max_coupling}), where we use $P, Q$ to denote the target and draft distributions, respectively, and achieve
\[
\mathrm{Pr}(X = Y) = \sum\nolimits_{x\in\mathcal{V}} \min (P(x), Q(x)) = 1 - \Vert P - Q\Vert_\mathrm{TV}, 
\]
i.e., the closer $P$ and $Q$ are, the higher the probability that the draft token will be accepted. 
Together with the extra sampled token, there could be at most $L + 1$ tokens %been 
generated in each iteration, and therefore, the speedup is up to $(L+1)$ times if the decoding time of the draft model is negligible.

\textcolor{black}{
\citet{kobus2025speculative} proposed a new sampling scheme (Algorithm~\ref{alg:ESGD_token}) based on exponential random variables, which, together with the token probabilities, determine the scores assigned to the tokens. The token with the smallest score wins the race and follows the desired distribution.
}

\begin{figure}[H]
% \begin{figure}[t]
\centering
\begin{minipage}[t]{.545\textwidth}

\begin{algorithm}[H]
\SetAlgoNoEnd
\SetAlgoLined
\DontPrintSemicolon
\SetKwIF{If}{ElseIf}{Else}{if}{:}{else if}{else}{end}
\SetKwInOut{Input}{Input}\SetKwInOut{Output}{Output}
\textbf{Input:} Target and draft distribution $P,Q$; $X \sim Q$. \; 
Compute the residual $P_{\mathrm{res}}$: for any $x\in \mathcal{V}$,\vspace{-3pt}
\[
P_{\mathrm{res}}(x)=\frac{P(x)-\min\{P(x),Q(x)\}}{1-\sum_y \min\{P(y),Q(y)\}}.\vspace{-5pt}
\]
Sample $U \sim \mathrm{Unif}(0,1)$
\;\vspace{2pt}
\If{$U \leq \min\left(1, {P(X)}/{Q(X)}\right)$}{
\Return{$Y=X$} \tcp*[r]{Accept}
}
\Return{$Y \sim P_{\mathrm{res}}$} \tcp*[r]{Correction}
\caption{Token-level sampling~\citep{sun2023spectr, chen2023accelerating, leviathan2023fast}}
\label{alg::max_coupling}
\end{algorithm}

\end{minipage}%
\hfill
\begin{minipage}[t]{.45\textwidth}

\begin{algorithm}[H]
\SetAlgoNoEnd
\SetAlgoLined
\DontPrintSemicolon
\SetKwIF{If}{ElseIf}{Else}{if}{:}{else if}{else}{end}
\SetKwInOut{Input}{Input}\SetKwInOut{Output}{Output}
\textbf{Input:} Target and draft distribution $P, Q$. \;
\For{$i\in \mathcal{V}$}{
    $E_i\sim\mathrm{Exp}(1)$\;
}
$X \gets \arg\min_{i\in \mathcal{V}} E_i/Q(i)$\;
$Y \gets \arg\min_{i\in \mathcal{V}} E_i/P(i)$\;
\If{$X = Y$}{
\Return{$Y = X$} \tcp*[r]{Accept}
}
\Return{$Y$} \tcp*[r]{Correction}
\caption{Token-level sampling by exponential race~\citep{kobus2025speculative}}
\label{alg:ESGD_token}
\end{algorithm}
\end{minipage}
\end{figure}
% \subsection*{Speculative Sampling via Exponential Races}

\section{Watermarkable Speculative Sampling}\label{sec:SingleVersion}

\textcolor{black}{We first present a single-draft version for illustration.
Our use of Poisson processes naturally embeds the watermark into the sampling process via pseudorandomness.
The formal multi-draft scheme, which further accelerates autoregressive generation without compromising watermark strength, will be presented in Section~\ref{sec::multi_draft}.
Before proceeding, we first introduce two key components of our scheme.}

\subsection{Unbiased Watermark}
\label{subsec::gumbel_water}

A good watermark should be detectable from the generated text while preserving output quality. 
\textit{Unbiased} schemes use pseudorandomness to couple tokens with secret keys, so that the watermark can be detected through statistical tests without significantly altering the marginal distribution.

Suppose $\tilde{P}$ is a base distribution over $\mathcal{V}$, an unbiased watermark is a decoding rule such that, for each seed $\zeta$, it induces a watermarked distribution $\tilde{P}_\zeta$, which preserves $\tilde{P}$ after averaging over the seed, 
\[
\mathbf{E}_{\zeta} [\tilde{P}_\zeta(v)] = \tilde{P}(v), \qquad \forall v\in\mathcal{V}.
\]
The Gumbel watermark~\citep{aaronson2023watermarking}, which uses the Gumbel-max trick~\citep{gumbel1954statistical} and lets the seed $\zeta$ determine the Gumbel noise, is unbiased, and we leverage a similar idea for our scheme. 
Let $u_t(\cdot | x, y_{t-1}): \mathcal{V}\to \mathbb{R}$ be a real-valued function (a.k.a. \emph{logits}) that encodes the model's preferences on words and $\mathsf{T}$ be the temperature, the softmax sampling is equivalent to
%\begin{equation}
%    y_t \sim p(y) = {\exp\big(u(y|x,y_{1:t-1})/T\big)}
%    \Big/ \, \Big({\sum_{\tilde{y}} \exp\big(u(\tilde{y}|x,y_{1:t-1})/T\big)\Big)} \label{eq::softmax}
%\end{equation}
%where $x$ is the \emph{prompt} for the current task, $y_{1:t-1}$ is the \emph{prefix} that include all previously generated words up to time $t$, and 
%$u(\cdot | x, y_{1:t-1}): \mathcal{V}\to \mathbb{R}$ is a real-value function (a.k.a. \emph{logits}) that encodes the model's preferences on words from the vocabulary $\mathcal{V}$. 
%By the Gumbel-max trick, \eqref{eq::softmax} is equivalent to 
\[
y_t =  \mathrm{argmax}_{y\in\mathcal{V}}  \big({u_t(y)} \big/ \mathsf{T} \big) + G_t(y),\quad G_t(y)\sim\textrm{Gumbel}(0,1) \text{ i.i.d for each } t,y,
\]
where  the Gumbel noise can be generated by $\textrm{Gumbel}(0,1) \sim -\log(\log(1/\textrm{Unif}([0,1])))$. 

A random $r_t\sim  (\textrm{Unif}([0,1]))^{|\mathcal{V}|}$ can be represented by $G_t(y) = - \log (-  log (r_t(y)))$. 
We can replace $\textrm{Unif}([0,1])$ by a pseudo-random function $r_t(y) = F_{y_{t-m:t-1},\mathsf{k}}(y)$ with prefix length $m$, and employ
\begin{equation}
    y_t =  \mathrm{argmax}_{y\in\mathcal{V}}  \big({u_t(y)} \big/ \mathsf{T} \big) -\log(-\log(r_t(y))). \label{eq::decoding_pseudo}
\end{equation}
Conditioned on a secret key $\mathsf{k}$, $r_t(y)$ is a deterministic function, but over the distribution of $\mathsf{k}$, $r_t(y)$ is computationally indistinguishable from %sampled 
samples drawn from a truly i.i.d. uniform distribution. 
% Therefore,   the distribution of $y_t$ from~\eqref{eq::decoding_pseudo} is considered the same as the unwatermarked one. 

The detector has the key $\mathsf{k}$ and computes the watermark score $ \sum^n_{t = m + 1}-\log\big(1 - r_t(y_t) \big)$. 
If $y_{1:n}$ is unwatermarked, the expected score is $n - m$; otherwise, it is larger and the watermark is detected. 
% See Appendix~\ref{subapp::gumbel_water} for details.  

Our algorithm relies on Poisson processes and also incorporates the use of pseudo-random functions.

\subsection{Poisson Functional Representation (PFR)}
\label{subsec::PFR}

\textcolor{black}{We leverage a scheme called PFR~\citep{li2018strong}; related works are discussed in Section~\ref{sec::related}.}

\begin{definition}[PFR] 
\label{def::PFR}
Let $(T_{i})_{i}\sim \mathrm{PP}(1)$ be a Poisson process of rate $1$,
% (i.e., $T_{1},T_{2}-T_{1}\ldots \stackrel{\mathrm{iid}}{\sim} \mathrm{Exp}(1)$)
independent of $Z_{i}\stackrel{\mathrm{iid}}{\sim} \nu$ for $i=1,2,\ldots$. 
$(Z_{i},T_{i})_{i}$
is a Poisson process of intensity measure $Q\times\lambda_{[0,\infty)}$, where $\lambda_{[0,\infty)}$ is the Lebesgue measure over $[0,\infty)$.
Fix distribution $P$ over $\mathcal{Z}$ such that $P\ll Q$. 
PFR selects the point
\begin{equation}
    Z=Z_{K}\quad \text{ where }\quad K:=\mathrm{argmin}_{i}\big( T_{i} \cdot \big({\mathrm{d}P}/{\mathrm{d}Q}(Z_{i})\big)^{-1}\big).
    \label{eq:PFR}
\end{equation}
% \citep{last2017lectures}
% Then $(Z_{i},\tilde{T}_{i})$ is also a Poisson process of intensity measure $P\times\lambda_{[0,\infty)}$. 
% PFR selects the point $Z=Z_{K}$ with the smallest associated $\tilde{T}_{K}$, 
% i.e., let $K:=\mathrm{argmin}_{i}\tilde{T}_{i}$ and $Z := Z_K$.
\end{definition}

\begin{remark}
    By drawing a sequence $(Z_i)_i$ from the \emph{reference} distribution $Q$ and a sequence of \emph{times} $(T_i)_i$, the PFR selects a sample following the \emph{target} distribution $P$. 
    The sample $Z_i$ with the smallest $T_i$ only follows $Q$; to obtain a sample from $P$ exactly, we inflate the time by factor $({\mathrm{d}P}/{\mathrm{d}Q}(Z_{i}))^{-1}$. 
    % Although PFR requires infinite shared randomness, it can be implemented easily in practice~\citep{theis2022algorithms}. 
\end{remark}

\subsection{Watermarkable Speculative Sampling via Poisson Processes}
\label{label::water_spec_single}

\paragraph{Core idea.} 
We unify  watermarking and speculative sampling through a keyed Poisson coupling. 
Every emitted token is target-distributed, while the draft model only controls how many such samples are emitted. 
The watermarking is performed at the \emph{coupling level}, rather than through a post-hoc reweighting of marginals. 
This allows our scheme to stay efficient when embedding watermark.
\footnote{\textcolor{black}{The Maintain Watermark Strength scheme \citep{hu2024inevitable} shares the same spirit; see Figure~\ref{fig::single_draft_qwen_cnn_new}.
Nevertheless, our scheme can be extended to the multi-draft version in Section~\ref{sec::multi_draft}, achieving further acceleration.}
}

% We integrate the ideas in Sections~\ref{subsec::gumbel_water} and \ref{subsec::PFR} to design a single-draft version of our algorithm. %first present a single-draft version of our algorithm that connects the ideas in Section~\ref{subsec::gumbel_water} and \ref{subsec::PFR}. 
% , and \textcolor{red}{we show our algorithm achieves the fundamental trade-off curve of watermark strength and sampling efficiency???}

Algorithm~\ref{alg:single_PFR_Spec_Water_simplified} is a sketch of our single-draft speculative sampling with watermark, see a detailed implementation in Algorithm~\ref{alg:fwss_pfr}. 
Fixing a Poisson process $\Pi_\mathsf{k}$ by key $\mathsf{k}$; with context $c$, we denote the labeled Poisson process and the sample selected by PFR by $G(\mathsf{k}, c)$ and PFR$(\bar{P}(\cdot | c),\Pi )$, respectively.

% PFR$(\bar{P}(\cdot | c),\Pi )$ denotes the with target distribution $\bar{P}(\cdot | c)$ and Poisson process $\Pi$. %, and we 

% \begin{algorithm}[H]
\begin{algorithm}[tbp]
\SetAlgoLined
\DontPrintSemicolon
\SetAlgoNoEnd
\SetKwIF{If}{ElseIf}{Else}{if}{:}{else if}{else}{end}
\SetKwInOut{Input}{Given}

\Input{lookahead $K$, output length $N$, 
% target model $P$, draft model $Q$,
initial prompt $w_{1:n}$, Poisson generator $G$, secret key $\mathsf{k}$}

\While{$n < N$}{
    $h \gets w_{1:n}$; $L \gets \min\{K,\,N-n\}$; $accepted \gets \textbf{true}$; \;
    % \tcp{Draft a speculative block from keyed shared Poisson sources}
    \For{$s=1$ \KwTo $L$}{
        $c_s \gets h \,\|\, \tilde w_{1:s-1} $; 
        \textcolor{black}{$\Pi_s \gets G(\mathsf{k}, c_s)$} 
        % \hspace{-10pt} 
        \tcp*[r]{\textcolor{black}{keyed Poisson for watermark}}
        $\tilde w_s \gets \textsc{PFR}(Q(\cdot | c_s), \Pi_s)$\;
    }

    \For(\tcp*[f]{evaluation in parallel}){$s=1$ \KwTo $L$}{
        $y_s \gets \textsc{PFR}(P(\cdot | c_s), \Pi_s)$; 
        $a_s \gets \mathds{1} \left\{y_s = \tilde w_s\right\}$; 
        % \tcp*[r]{watermark-bearing acceptance bit}
        $w_{n+1} \gets y_s$; $n \gets n+1$; \;
        \If{$a_s = 0$}{
            $accepted \gets \textbf{false}$ and \textbf{break} \tcp*[r]{later draft contexts are invalid}
        }
    }
    \If(\tcp*[f]{all draft tokens accepted}){$accepted$ \textbf{and} $n < N$}{
    $\Pi^\star \gets G(\mathsf{k}, w_{1:n})$ \tcp*[r]{keyed Poisson source}
    $w_{n+1} \gets \textsc{PFR}(P(\cdot | w_{1:n}), \Pi^\star)$; 
    $n \gets n+1$ \tcp*[r]{bonus}
    }
}\vspace{-1pt}
\caption{Watermarkable Speculative Sampling via Poisson Processes (Simplified)}
\label{alg:single_PFR_Spec_Water_simplified} 
\end{algorithm}

Note that for discrete LLM target distributions, as we prove in Appendix~\ref{subapp:PFR_discrete}, PFR reduces to the exponential race~\citep{kobus2025speculative}.
However, Poisson processes become useful, and even necessary, when we extend to the multi-draft version, where we generate multiple samples that all follow the target distribution, by the mapping theorem~\citep{last2017lectures} for Poisson processes.

{\color{black}

\citet{he2026improving} used a different watermark metric to improve the trade-off frontier. 
In comparison, our scheme is robust across different watermark measurements, as we consider various scores~\citep{aaronson2023watermarking, hu2024inevitable, li2025statistical, lattimore2026refined} in Appendix~\ref{app::water}, which makes our scheme directly comparable to existing watermarking schemes. 
We use Aaronson's score~\citep{aaronson2023watermarking} in the main section, which is introduced as follows.  

}

\vspace{-3pt}

If the output of key $\mathsf{k}$ is $w_{1:N}$ and $n_0$ is the initial prompt length, the detector evaluates the scored positions $\mathcal{T} \hspace{-1pt} := \hspace{-1pt}  \{n_0+1,\hspace{-1pt} \ldots\hspace{-1pt} ,N\}$. 
For $t\in\mathcal{T}$, let $c_t := w_{:t-1}$ be the prefix before position $t$, and the detector regenerates $\Pi_t\hspace{-1pt}  := \hspace{-1pt} G(\mathsf{k}, c_t)$. 
For each $v\hspace{-1pt} \in\hspace{-1pt} \mathcal{V}$, let $r_t(v)\hspace{-1pt} :=\hspace{-1pt}  \exp(-\mu(v)\tau_t(v))$ be the induced  uniform value where $\tau_t(v)\hspace{-1pt} \sim \hspace{-1pt} \mathrm{Exp}(\mu(v))$ is the first arrival time of $v$ in $\Pi_t$, and the detector computes \vspace{-3pt}
\begin{equation}
    S_{\mathrm{A}}
    :=
    \sum\nolimits_{t\in\mathcal{T}} -\log(1-r_t(w_t)),\label{eq::S_A} 
\end{equation}
whose expectation is $\mathbf{E}[S_{\mathrm{A}}] \geq |\mathcal{T}| + (\pi^2/6-1)\sum_{t\in \mathcal{T}} \mathbf{E}[H(P(\cdot|w_{:t-1}))]$. 
If the sequence is not generated with $\mathsf{k}$, the observed $w_t$ is independent of $\Pi_t$, $r_t(w_t)\sim \mathrm{Unif}[0,1]$ and $\mathbf{E}[S_{\mathrm{A}}]  = |\mathcal{T}|$.  

\vspace{-3pt}

\paragraph{At-a-Glance Experimental Validation.}
We use Figure~\ref{fig::single_draft_qwen_cnn_new} to illustrate how our algorithm improves the trade-off by increasing sampling efficiency without degrading watermark strength.

\begin{figure}[htpb]
    \centering
    \includegraphics[scale=0.35]{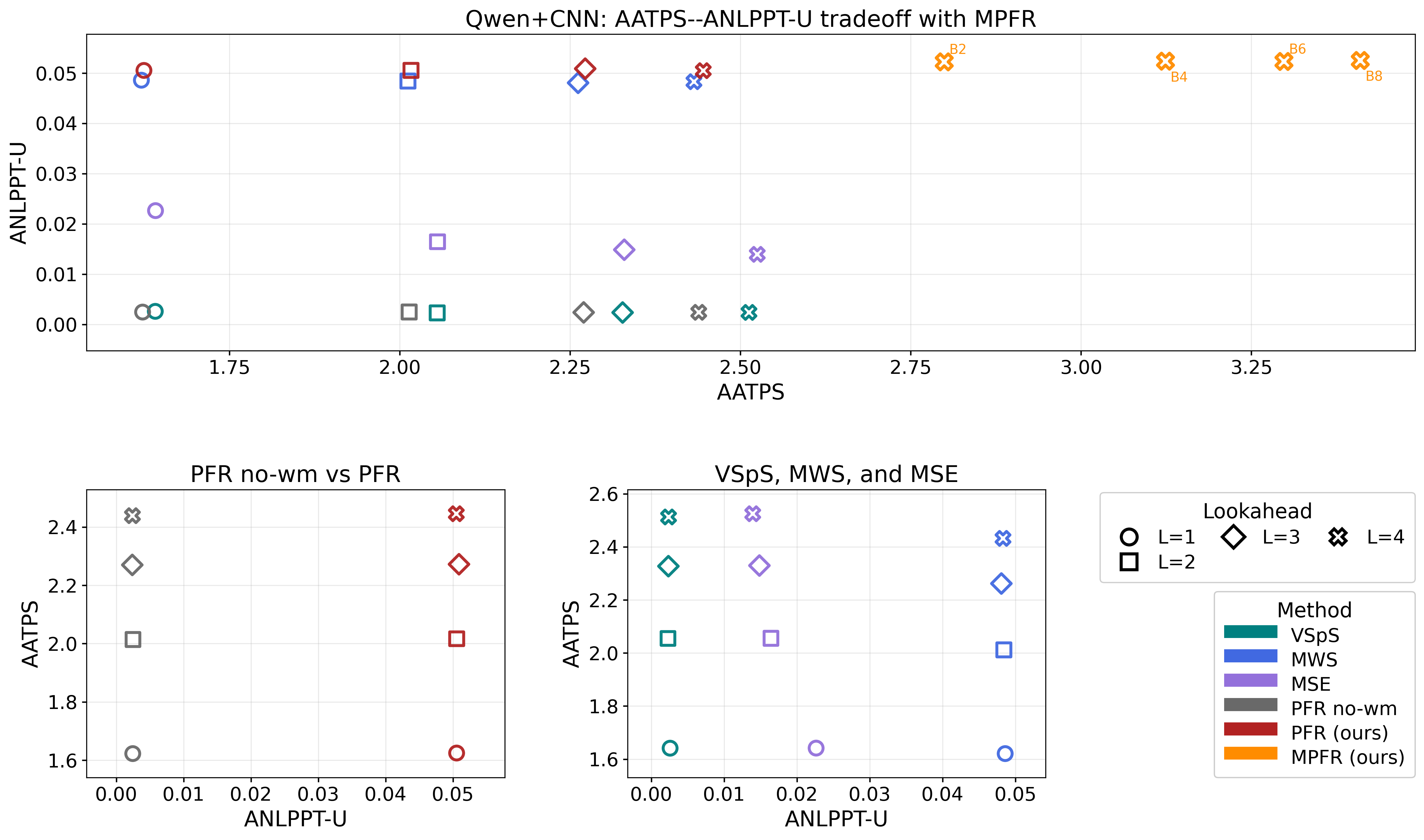}
    \caption{
    The trade-off between watermark and efficiency on existing methods~\citep{hu2024inevitable}: Vanilla Speculative Sampling (VSpS), Maintain Watermark Strength (MWS), and Maintain Sampling Efficiency (MSE), and our algorithms (PFR, and multi-draft MPFR that is presented in Section~\ref{sec::multi_draft}). 
    For both PFR (lower-left) and MPFR (upper), we embed watermark (measured by Average Negative Log P-value Per Token (ANLPPT)) without degrading the Average Accepted Tokens Per Step (AATPS). 
    This figure uses Qwen2.5-7B-Instruct and Qwen2.5-0.5B-Instruct as the target and draft, respectively, on dataset \textsc{cnn\_dailymail}. \textcolor{black}{Experiments on more models and datasets are shown in Appendix \ref{app::experiments}. }
    }
    \label{fig::single_draft_qwen_cnn_new}
\end{figure}

\newpage
\subsection{Drafter Invariance}
\label{sec::drafter_invar}

% {\color{red}
% One crucial property of Algorithm~\ref{alg:ESGD_token} is that the target and the draft models do not communicate, and \emph{coupling without communication} yields a \emph{drafter invariance} property~\citep{daliri2025coupling, rowan2025list}. 
% Our algorithm admits a stronger level of drafter invariance; it will be  discussed in Section~\ref{sec::drafter_invar}. 

% This property is particularly helpful for watermarking because it ties the watermark to the target-side keyed randomness rather than drafter-specific behavior; hence changing the drafter (e.g., via a model update) does not affect the watermark carried by the generated text. 

% emphasize the reason we consider drafter invariance is motivated by watermarking. It is not post-hoc (how?)
% }

Conventional speculative sampling suffers from a major disadvantage~\citep{daliri2025coupling}: when the small drafter model $\mathcal{M}_d$ changes, e.g., because of a model update, the tokens generated by the large model may also change due to the coupling (Algorithm~\ref{alg::max_coupling}). 
In autoregressive language models, this can be even more problematic, since users usually do not want the model output to change under a fixed random seed: \emph{stable} models allow the output to be reliably reproduced for use, testing, and debugging.

Based on the coupling without communication~\citep{bavarian2016optimality, li2018strong}, \emph{drafter invariant} sampling has recently been studied~\citep{daliri2025coupling}. It %, which 
refers to the case where, conditioned on the shared randomness $\mathcal R$ and the context $x_{:t}$, and potentially other things, the output-prefix law does not depend on the draft model. 
Let $Y_{1:\tau}, \widetilde{Y}_{1:\tau}$ denote the output sequences from different drafter models in block $\tau$, we define the following notion. 
\vspace{-3pt}

\begin{definition}
\label{def:st_tm_drafter_invariance_single}
The algorithm $\mathcal{A}$ is said to be \emph{stopping-time drafter invariant} if, for every stopping time value $\tau_0$ in the support of $\tau$, every $1\le j\le \tau_0$, and every output prefix $y_{1:j}$, we have\vspace{-3pt}
\begin{equation*}
\mathbf{P}\big\{
Y_{1:j}=y_{1:j}\,\big|\,\mathcal{R},\,x_{:t},\,\tau=\tau_0
\big\}
=
\mathbf{P}\big\{
\widetilde{Y}_{1:j}=y_{1:j}\,\big|\,\mathcal{R},\,x_{:t},\,\widetilde{\tau}=\tau_0
\big\}.
\end{equation*}
\end{definition}

\begin{proposition} 
\label{prop::single_dft_st_tm_inv} 
Our Algorithm~\ref{alg:single_PFR_Spec_Water_simplified} is stopping-time drafter invariant. 
\end{proposition}

The proof can be found in Appendix~\ref{subsubapp::proof_single_dft_st_tm_inv}. 
The stopping-time drafter invariance is not satisfied by most existing schemes, e.g., SpecTr~\citep{sun2023spectr}, SpecInfer~\citep{miao2024specinfer}, and also \citep[Algorithm 2]{rowan2025list}.

{\color{black}
Intuitively, a stronger notion of drafter invariance helps watermarking but degrades sampling efficiency. 
The original definition, referred to as the ``strong drafter invariance''~\citep{daliri2025coupling}, was shown to have degraded sampling efficiency~\citep{rowan2025list}, in where the authors hence proposed a ``conditional drafter invariance'' that also conditions on the realized draft sequence. 
We propose the stopping-time drafter invariance that is weaker than the former to admit a better sampling efficiency, but stronger than the latter and more suitable for watermarking. 
The conditioning on $\tau$ is intentional: it separates the role of the drafter in determining the \emph{length} of the current speculative block from its role in determining the \emph{content} of the emitted tokens. 
Thus, the drafter is allowed to affect the stopping time, and hence the sampling efficiency, but once this stopping time is fixed, the law of every emitted prefix is governed by the target-side randomness and does not depend on the full realized draft sequence. 
}

In watermarking, such invariance is important. 
A watermark is usually tied to pseudorandomness and is detected from the final output sequence, while speculative sampling introduces an additional ambiguity: an emitted token may come from the draft, the target correction, or the bonus, which affects the watermark statistic, and is one reason behind the no-go trade-off in~\citep{hu2024inevitable} and the use of acceptance-side information in~\citep{he2026improving}. 
However, a watermark detector typically observes only the final tokens and the secret key, not the realized draft sequence or acceptance-side information. 
Stopping-time drafter invariance is hence better aligned with watermarking: changing the drafter may affect the stopping time and the efficiency, but conditional on the block length, it does not introduce draft-specific dependence into the target-keyed token prefix. 
This makes the watermark directly tied to the target-side keyed randomness and leads to a more reproducible provenance signal.

{\color{black}
Moreover, \citep{he2026improving} proposed a new watermark metric that improves the trade-off between watermark strength and sampling efficiency~\citep{hu2024inevitable}; see, e.g., \citep[Figure~1]{he2026improving}. 
However, their analysis relies on the maximal coupling probability $\mathbf{P}(X=Y)=1-\Vert P-Q\Vert_{\mathrm{TV}}$ for coupling $X\sim P, Y\sim Q$~\citep{thorisson2000coupling, den2012probability}. 
In contrast, the best matching probability under no-communication coupling is fundamentally different:
\begin{equation}
    \mathbf{P}(X = Y) \leq (1 - \Vert P - Q\Vert_{\mathrm{TV}}) / (1 + \Vert P - Q\Vert_{\mathrm{TV}}).
    \label{eq::cpl_wo_commu_bd}
\end{equation}
Thus, enforcing drafter invariance may appear to fundamentally reduce sampling efficiency. 
We argue, however, that this loss is not inevitable: by using Poisson processes, we derive a multi-draft, drafter-invariant extension that greatly improves the token acceptance rate, as elaborated as follows. 
}

\section{Watermarkable Multi-Draft Speculative Sampling}
\label{sec::multi_draft}

% \textcolor{red}{Moreover, as proved in Appendix~\ref{subsapp::PFR_gumbel}, our scheme is conceptually equivalent to the Gumbel watermark~\citep{aaronson2023watermarking} in the single-draft case. }

We now extend our Algorithm~\ref{alg:single_PFR_Spec_Water_simplified} to multi-draft speculative sampling, by a multi-sample generalization of PFR built on the mapping theorem for Poisson processes~\citep{li2021unified}.

We use $\mathcal{M}_t$ and $\mathcal{M}_d^{(b)}$ to denote the \emph{target} and \emph{draft} language models with %, respectively, where 
$1\leq b\leq B$. They induce conditional distributions $\mathcal{M}_t(\cdot| x_{:t})$ and $\mathcal{M}_d^{(b)}(\cdot| x_{:t})$, which give the probability of the next token given the current context $x_{:t}$, usually denoted by $c$.
In multi-draft speculative sampling~\citep{miao2024specinfer, sun2023spectr, khisti2024multi, rowan2025list}, at each context $c$, we generate a list of $B$ proposal drafts, often i.i.d. in practice~\citep{rowan2025list, sun2023spectr}  from $\mathcal{M}_d^{(1)}, \ldots, \mathcal{M}_d^{(B)}$. These proposal drafts %, which 
induce a depth-$L$ speculative tree rooted at the current prefix, which is then verified %and verify them 
in parallel with $\mathcal{M}_t$.
At each step, %Among the $B$ candidate tokens, 
if at least one of the $B$ candidate tokens is accepted, %at each step, 
the first accepted token will be appended to the final output; otherwise, the verification procedure stops. 
The output is $Y_{1:\tau}$, where $\tau$ is the number of accepted tokens plus one.

% \subsection{Algorithm Design}

The PFR (Definition~\ref{def::PFR}) selects a sample $\tilde{Z}_P$ from $(Z_i)_i$ that has the smallest \emph{score} $T_i\cdot \big(\frac{\mathrm{d}P}{\mathrm{d}Q}(Z_i)\big)^{-1}$, where $(T_i)_i\sim \mathrm{PP}(1)$, and it follows the target distribution $P$. 
To get $B$ samples, multi-sample PFR (MPFR) selects samples with the smallest, second smallest, and so on up to the $B$-th smallest scores:
\begin{equation}
    \tilde{Z}_P(1), \tilde{Z}_P(2),\ldots, \tilde{Z}_P(B) \stackrel{\mathrm{iid}}{\sim} P, \label{eq::MPFR}
\end{equation}
where the exactness follows~\citep{li2021unified}; see Appendix~\ref{subapp::mapped_PP} for the formal definition and implementation.

For multi-draft speculative sampling, we use $B$ i.i.d. drafts from a model $\mathcal{M}_d$, and at each sampling step we consider $P := \mathcal{M}_t(\cdot| x_{:t})$ as the target distribution and $Q := \mathcal{M}_d(\cdot| x_{:t})$ as the proposal distribution. 
Note that we consider \emph{identical} proposals~\citep{sun2023spectr} in this section, but our framework can also be generalized to non-identically distributed proposals~\citep{rowan2025list}, which is deferred to Appendix~\ref{app::non_id_prop}.

Our scheme is %briefly 
sketched in Algorithm~\ref{alg:multi_PFR_spec_water_simplified}; see Algorithm~\ref{alg:multi_fwss_pfr_formal} for details. 
The core idea is as follows. 
To implement a multi-draft algorithm, Algorithm~\ref{alg:multi_PFR_spec_water_simplified} uses \emph{context-indexed} mapped Poisson processes.
Starting from the root prefix $h=w_{1:n}$, we build a depth-$\ell$ speculative tree, where each occupied context $c$ is assigned a keyed source $\Pi(c)$ and generates $\nu(c)$ drafts from $Q(\cdot |  c)$ via MPFR.
We then compute the target-side sample $\textsc{MPFR}(P,\Pi(c),1)$ for every $c$ and emit tokens by following the unique realized target path, continuing as long as the current target winner belongs to the draft set.

Importantly, the Poisson process is indexed by contexts rather than by drafter identities, and thus MPFR works in a way that makes the emitted tokens tightly tied to the prefix and the secret key, which helps watermark detection.
In comparison, \citep{rowan2025list} uses draft-indexed algorithms, keeps a set of currently viable draft indices, and chooses the next emitted token using only the active drafts.
This is a natural implementation, but the active set makes the next emitted token depend on the realized draft sequences, causes the output to retain hidden dependence on draft-side trajectories, and hence only satisfies their conditional drafter invariance, but not stopping-time drafter invariance (Definition~\ref{def:st_tm_drafter_invariance_single}).

\begin{algorithm}[tbp]
% \begin{algorithm}[H]
\SetAlgoLined
\DontPrintSemicolon
\SetAlgoNoEnd
\SetKwIF{If}{ElseIf}{Else}{if}{:}{else if}{else}{end}
\SetKwInOut{Input}{Given}

\Input{lookahead $L$, number of homogeneous drafts $B$, output length $N$, target model $P$,
draft model $Q$, initial prompt $w_{1:n}$, keyed source generator $G$,
and secret key $\mathsf{k}$}

\While{$n < N$}{
    $h \gets w_{1:n}$; $\ell \gets \min\{L,\,N-n\}$\tcp*[r]{root context of current block} 
    initialize root context $h$ with multiplicity $B$\;
    
    \For(\tcp*[f]{build context tree up to depth $\ell$}){$s=1$ \KwTo $\ell$}{
        \For{each currently occupied context $c$}{
            $\Pi(c) \gets G(\mathsf{k},c)$ \tcp*[r]{\textcolor{black}{keyed Poisson source for watermark}}
            
            generate $\nu(c)$ drafts by $Q(\cdot | c)$ using $\textsc{MPFR} (Q(\cdot | c),\Pi(c), \cdot)$ \tcp*[r]{Eq.\eqref{eq::MPFR}} 
            % \[
            % \textsc{MPFR}(Q(\cdot | c),\Pi(c),1),\ldots,
            % \textsc{MPFR}(Q(\cdot | c),\Pi(c),\nu(c))
            % \]
            merge equal draft tokens into child contexts and update their multiplicities\;
        }
    }

    compute target-side sample $y(c)=\textsc{MPFR}(P(\cdot | c),\Pi(c),1)$ for every $c$ \tcp*[r]{in parallel}

    $c^\star \gets h$; $accepted \gets \textbf{true}$\;
    \For{$s=1$ \KwTo $\ell$}{
        emit the target token $y(c^\star)$ and append it to the output\;
        \If(\tcp*[f]{$\mathcal{D}(c^\star)$ contains distinct draft tokens at $c$}){$y(c^\star)\notin \mathcal{D}(c^\star)$}{
            $accepted \gets \textbf{false}$ and \textbf{break}
            \tcp*[r]{first rejection ends the block}
        }
        \If{$n = N$}{
            \textbf{break}
        }
        $c^\star \gets c^\star \,\|\, y(c^\star)$
        \tcp*[r]{move to the next realized context}
    }

    \If{$accepted$ \textbf{and} $n < N$}{
        emit one bonus token by $\textsc{MPFR}(P(\cdot | c^\star),G(\mathsf{k},c^\star),1)$ 
    }
}
\caption{Watermarkable Multi-Draft Speculative Sampling (Simplified)}
\label{alg:multi_PFR_spec_water_simplified}
\end{algorithm}

% \subsection{Theoretic Guarantee}
\begin{theorem}
    \label{thm::PML_accept}
    The acceptance probability of at least one token at the current step satisfies 
    \begin{equation}
        \label{eq::PML_accept}
        \mathbf{P}(\mathrm{accept}) \geq 
        1-\sum_{i\in\mathcal{V}} P(i)\big({P(i)}\big/{(P(i) + Q(i))}\big)^B. 
    \end{equation}
\end{theorem}
The proof is in Appendix~\ref{app::theory}, where we also give a generalized version that accounts for the number of active drafts, and demonstrate our advantages over the existing schemes~\citep{rowan2025list}.

\begin{table*}[t]
\centering
\small
\color{black}
\caption{\textcolor{black}{Multi-draft results on Qwen2.5-7B-Instruct and
Llama-3.1-8B-Instruct on CNN/\textsc{DailyMail} and \textsc{ELI5}, $L=4$.
We report the average accepted tokens per step (AATPS) and token rate (TR).
The better result between \textsc{MPFR} and \textsc{GLS} for each setting
with the same $B$ is in bold.  Each cell is computed over the $1000$ prompts and $3$ different seeds. See more detailed data in Appendix~\ref{app:multi_draft_full}.}}
\label{tab:multi_full_comparison}
\begin{tabular}{l c cc cc cc cc}
\toprule
& &
\multicolumn{2}{c}{Qwen-CNN/\textsc{DM}}
& \multicolumn{2}{c}{Qwen-\textsc{ELI5}}
& \multicolumn{2}{c}{Llama-CNN/\textsc{DM}}
& \multicolumn{2}{c}{Llama-\textsc{ELI5}} \\
\cmidrule(lr){3-4}
\cmidrule(lr){5-6}
\cmidrule(lr){7-8}
\cmidrule(lr){9-10}
Decoder & $B$
& AATPS & TR
& AATPS & TR
& AATPS & TR
& AATPS & TR \\
\midrule
\textsc{MPFR} & 2
& $2.797$ & $\mathbf{25.11}$
& $\mathbf{2.563}$ & $\mathbf{23.57}$
& $\mathbf{3.459}$ & $\mathbf{37.44}$
& $\mathbf{3.103}$ & $\mathbf{34.29}$ \\

\textsc{MPFR} & 4
& $\mathbf{3.123}$ & $\mathbf{27.50}$
& $\mathbf{2.918}$ & $\mathbf{26.46}$
& $\mathbf{3.784}$ & $\mathbf{40.20}$
& $\mathbf{3.455}$ & $\mathbf{37.79}$ \\

\textsc{MPFR} & 6
& $\mathbf{3.294}$ & $\mathbf{28.50}$
& $\mathbf{3.109}$ & $\mathbf{27.81}$
& $\mathbf{3.935}$ & $\mathbf{41.71}$
& $\mathbf{3.635}$ & $\mathbf{39.59}$ \\

\textsc{MPFR} & 8
& $\mathbf{3.409}$ & $\mathbf{29.13}$
& $\mathbf{3.235}$ & $\mathbf{28.73}$
& $\mathbf{4.034}$ & $\mathbf{42.04}$
& $\mathbf{3.749}$ & $\mathbf{39.76}$ \\

\midrule

\textsc{GLS} & 2
& $\mathbf{2.803}$ & $24.78$
& $2.523$ & $22.73$
& $3.431$ & $36.79$
& $3.080$ & $33.72$ \\

\textsc{GLS} & 4
& $3.109$ & $27.12$
& $2.850$ & $25.60$
& $3.741$ & $39.31$
& $3.408$ & $37.11$ \\

\textsc{GLS} & 6
& $3.272$ & $27.84$
& $3.039$ & $27.01$
& $3.894$ & $40.51$
& $3.588$ & $39.12$ \\

\textsc{GLS} & 8
& $3.384$ & $28.62$
& $3.156$ & $28.09$
& $3.987$ & $39.87$
& $3.700$ & $39.37$ \\
\bottomrule
\end{tabular}
\end{table*}

\section{Experiments}
\label{sec::experiments}

We evaluate whether our algorithms improve the empirical frontier between sampling efficiency and watermark detectability, while preserving generation quality. 
The headline
configuration uses Qwen2.5-7B-Instruct as the target and Qwen2.5-0.5B-Instruct
as the drafter on \textsc{CNN/DailyMail}. 
Other model-dataset cells, full
per-$L$ and per-$B$ tables, and ablation details are deferred to
Appendix~\ref{app::experiments}.

\paragraph{Baselines.}
We compare against four families of methods, all run in a common harness with
the same prompts, decoding parameters, and watermark key.
\textbf{(i) Watermark-only reference.} \textsc{Basic-UWM} is autoregressive
Gumbel-style unbiased watermarking without speculation~\citep{aaronson2023watermarking};
it sets the watermark-strength ceiling.
\textbf{(ii) Speculation-only reference.} \textsc{VSpS} denotes vanilla
speculative sampling without watermark~\citep{leviathan2023fast}, and the
multi-draft list-coupling scheme \textsc{Invariant} of \citet{rowan2025list};
these set the efficiency ceiling.
\textbf{(iii) Trade-off frontier.} \textsc{MWS}~/\textsc{MSE} of
\citet{hu2024inevitable} are the two endpoints of the no-go trade-off, and
\textsc{MSE-Pseudo} is the pseudorandom-$r$ variant
of~\citet{he2026improving} that improves the trade-off via an alternative
watermark metric.
\textbf{(iv) Ours.} \textsc{PFR} (single-draft) and \textsc{MPFR} (multi-draft);
\textsc{PFR-NoWM} is \textsc{PFR} with the keyed source replaced by an
unkeyed Poisson source, isolating the cost of watermarking from the cost of
Poisson coupling.

\paragraph{Metrics.}
Following~\citep{hu2024inevitable, he2026improving} we separate three concerns.
\textbf{Efficiency:} AATPS; see other measures in appendix.
\textbf{Detectability:} ANLPPT, the average negative log $p$-value per token under
three score variants U / Li / PL~\citep{aaronson2023watermarking, li2025statistical, lattimore2026refined} (see Appendix~\ref{app::water});
and TPR@1\%FPR, calibrated on non-watermarked generations.
\textbf{Quality:} per-token log-perplexity (LPPL) under the target model, used as
an unbiasedness audit; we also report ROUGE-L when references are
available. 
% AATPS measures algorithmic acceptance, ANLPPT and TPR@1\%FPR measure
% detector-visible signal, and LPPL audits whether the target-model marginal is
% preserved.

\paragraph{Single-draft efficiency-watermark frontier.}
Figure~\ref{fig::single_draft_qwen_cnn_new}
plotted every method in the (AATPS, ANLPPT-U) plane sweeping
$L\in\{1,2,3,4\}$. 
The prior speculative watermarking baselines (\textsc{MWS}, \textsc{MSE}, \textsc{MSE-Pseudo}) exhibit the expected trade-off.
% : \textsc{MWS} preserves watermark strength but loses
% acceptance, \textsc{MSE} preserves acceptance but loses watermark strength, and
% \textsc{MSE-Pseudo} sits in between. 
% \textsc{PFR} lies in the upper-right of
% the plot at every $L$, dominating both Pareto endpoints simultaneously.
\textsc{PFR} dominates both Pareto endpoints at every $L$. 
Crucially, \textsc{PFR} and \textsc{PFR-NoWM} have nearly identical AATPS, confirming that the watermark is embedded at the coupling level. 
Other model/dataset results are tested in Appendix~\ref{app:single_draft_full}.

\paragraph{Multi-draft scaling.}
We compare \textsc{MPFR} with the \textsc{Invariant} decoder
by~\citet{rowan2025list} at lookahead $L{=}4$ and draft counts $B\in\{2,4,6,8\}$ on the headline cell. Across all values of $B$, MPFR achieves AATPS comparable to INVARIANT, with only small cell-dependent deviations, while delivering substantially larger ANLPPT-U. The watermark signal is itself stable across $B$ ($\Delta$ANLPPT $\leq 0.001$), and the per-token detection signal does not dilute as additional drafts are added. 
Full per-cell numbers across the four $(\textsc{target},\textsc{dataset})$ cells, all
three ANLPPT variants, and the LPPL audit column are in
Appendix~\ref{app:multi_draft_full}
(Tables~\ref{tab:full-qwen-cnn}-\ref{tab:full-vicuna-eli5}).

\paragraph{Detection at a fixed false-positive rate.}
ANLPPT measures average evidence per token; TPR@1\%FPR measures operational
detectability under a calibrated false-positive budget.
Figure~\ref{fig:drafter_invariance}(a) shows TPR@1\%FPR versus token budget
$T_{\mathrm{eval}}$ at the default drafter. \textsc{PFR} and \textsc{MPFR}
track the strong-watermark references \textsc{Basic-UWM} and \textsc{MWS},
while \textsc{MSE} and \textsc{MSE-Pseudo} are uniformly weaker.
\textsc{MSE-Pseudo} improves over \textsc{MSE} but still falls below  our methods, confirming that target-side keyed coupling preserves detector-visible evidence without degrading efficiency.

\begin{figure}[t]
\centering
\includegraphics[width=0.95\linewidth]{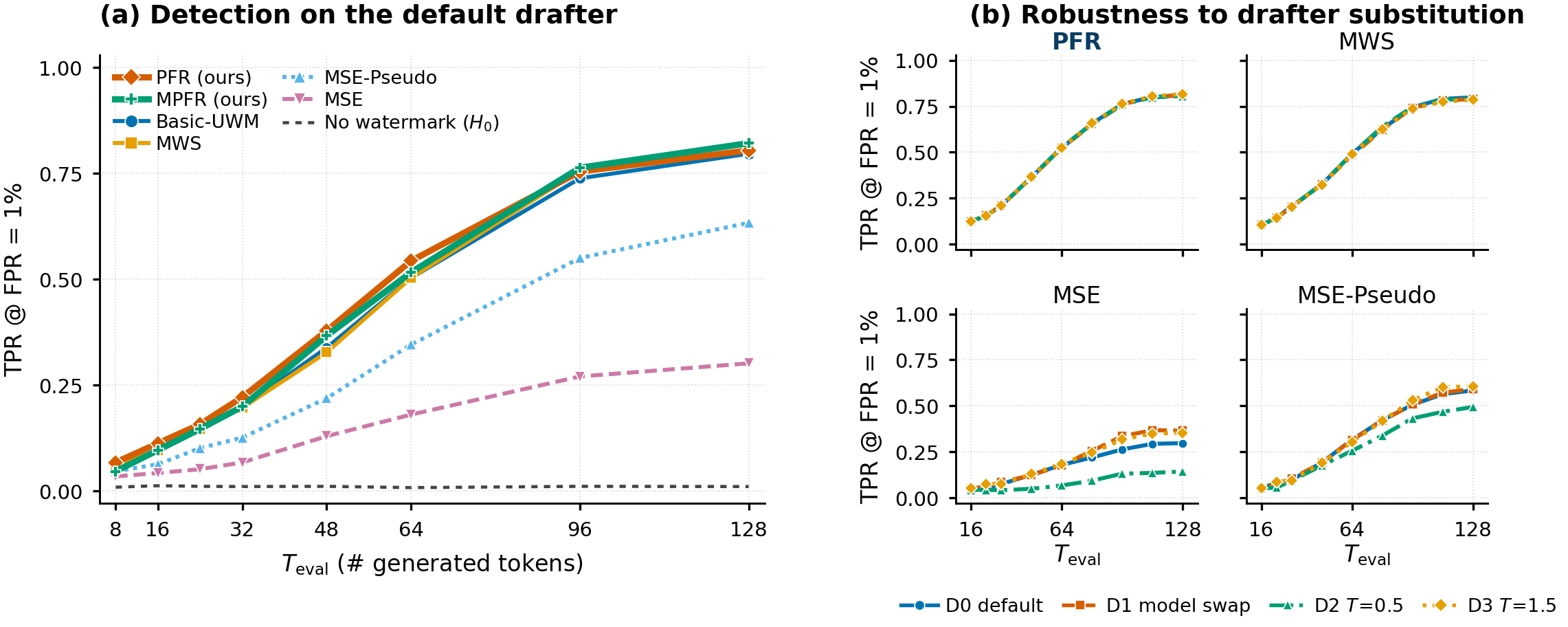}
\caption{\textbf{Detection and robustness to drafter substitution.}
\textbf{(a)} TPR@1\%FPR versus detection token budget $T_{\mathrm{eval}}$ at
the default drafter. \textsc{PFR} and \textsc{MPFR} match the strong-watermark
references \textsc{Basic-UWM} and \textsc{MWS}, and outperform
\textsc{MSE}/\textsc{MSE-Pseudo}.
\textbf{(b)} TPR@1\%FPR under drafter substitution: the target, watermark
key, and prompts are fixed and only the drafter is swapped. \textsc{PFR} drifts by $<\!1$ point, whereas \textsc{MSE} drifts
by $12$--$22$ points and \textsc{MSE-Pseudo} by $6$--$11$ points.}
\label{fig:drafter_invariance}
\end{figure}

\paragraph{Drafter-substitution ablation.}
We test whether the watermark signal is stable to the drafter changes.
We vary only the drafter along scale and temperature: $D_0$ (Qwen2.5-0.5B-Instruct, $T{=}1.0$, default), $D_1$ (Qwen2.5-1.5B-Instruct, $T{=}1.0$), $D_2$ (Qwen2.5-0.5B-Instruct,
$T{=}0.5$), $D_3$ (Qwen2.5-0.5B-Instruct, $T{=}1.5$).
Figure~\ref{fig:drafter_invariance}(b) shows that our method (\textsc{PFR}) has at most $1$
percentage point of TPR drift across $D_0$--$D_3$ at $T_{\mathrm{eval}}\in\{64,128\}$, whereas \textsc{MSE} drifts by $12$--$22$
points and \textsc{MSE-Pseudo} by $6$--$11$ points. 
The generation quality (LPPL
and ROUGE-L) is verified in Appendix~\ref{app:ablation}.

\subsection*{Concluding Remarks and Limitations}

We designed a novel drafter-invariant, multi-draft speculative sampling algorithm that is watermarkable, improves the trade-off between watermarking and sampling efficiency~\citep{hu2024inevitable}, and achieves a strong acceleration.
We use experiments to verify its strong performance in both aspects.
However, a fundamental characterization of the trade-off among target-draft communication, watermark strength, and speculative sampling efficiency remains unclear.
Moreover, the current implementation involves substantial coupling and watermark bookkeeping, which slightly slows down the algorithm and may be avoided in the future to further accelerate our algorithms.

\newpage

\bibliography{ref.bib}
\bibliographystyle{plainnat}

\newpage
\appendix

\section{More on PFR and Related Works}
\label{app::Poisson}

\subsection{PFR for Discrete Distributions}
\label{subapp:PFR_discrete}

It is known that the general PFR~\citep{li2018strong, li2021unified} reduces to the use of exponential random variables~\citep{li2018strong, liu2025one} in the case of simulating discrete distributions, which is similar to the Gumbel-max trick~\citep{gumbel1954statistical} used for speculative sampling~\citep{kobus2025speculative}.

In~\citep{li2024channel}, it is discussed how to extend the discrete case to the general case. However, the other direction, namely reducing the general case to the discrete case, has not been discussed in the literature, to the best of the authors' knowledge.
To make the paper self-contained, we provide formal statements and proofs for this direction.

The first define the exponential functional representation (EFR) as follows (also see \citep{li2018strong, liu2025one}). 
Let $\mathcal{Z}$ be a finite or countably infinite set, and let $P$ be a distribution on $\mathcal{Z}$. 
Let $(E_z)_{z\in\mathcal{Z}}$ be independent $\mathrm{Exp}(1)$ random variables. 
With the convention that $E_z/P(z)=+\infty$ whenever $P(z)=0$, define
\begin{equation}
    Z_{\mathrm{EFR}}
    :=
    \arg\min_{z\in\mathcal{Z}} \frac{E_z}{P(z)} .
    \label{eq:EFR}
\end{equation}
As proved in Appendix~\ref{subsapp::PFR_gumbel}, EFR is equivalent to the Gumbel-max trick.
In the following, we prove that PFR reduces to EFR in the discrete case, i.e., when the target and proposal distributions are discrete.

\begin{theorem}
\label{thm::PFR_reduces_to_EFR}
Let $\mathcal{Z}$ be a finite or countably infinite set. Let $Q$ and $P$ be distributions on $\mathcal{Z}$ such that $P\ll Q$. Write $q_z:=Q(z), p_z:=P(z), z\in\mathcal{Z}$. 
The PFR considers a Poisson process $(Z_i,T_i)_{i\geq 1}$ on $\mathcal{Z}\times[0,\infty)$ with intensity measure
$Q\times\lambda_{[0,\infty)}$, and $\tilde T_i = T_i\left(\frac{\mathrm{d}P}{\mathrm{d}Q}(Z_i)\right)^{-1}$. 
For each $z$ with $q_z>0$, define the first arrival time of symbol $z$ by
\begin{equation}
    S_z:=\inf\{T_i:Z_i=z\},
    \label{eq:first_arrival_symbol_z}
\end{equation}
and define
\begin{equation}
    E_z:=q_z S_z .
    \label{eq:Ez_from_PFR}
\end{equation}
For $z$ with $q_z=0$, define $E_z$ to be an arbitrary independent $\mathrm{Exp}(1)$ random variable, independent of everything else. Then the random variables $(E_z)_{z\in\mathcal{Z}}$ are independent $\mathrm{Exp}(1)$ random variables, and the PFR output satisfies
\begin{equation}
    Z_{\mathrm{PFR}}
    =
    \arg\min_i \tilde T_i
    =
    \arg\min_{z\in\mathcal{Z}}\frac{E_z}{p_z}
    =
    Z_{\mathrm{EFR}}
    \qquad \text{a.s.},
    \label{eq:PFR_equals_EFR}
\end{equation}
where we use the convention that $E_z/p_z=+\infty$ if $p_z=0$.
\end{theorem}

\begin{proof}
For $z\in\mathcal{Z}$ and $t\geq 0$, define
\[
    N_z(t):=\#\{i:T_i\leq t,\ Z_i=z\}.
\]
Since $(Z_i,T_i)_{i\geq 1}$ is a Poisson process with intensity measure
$Q\times\lambda_{[0,\infty)}$, for each $z$ with $q_z>0$, the process
$(N_z(t))_{t\geq 0}$ is a Poisson process with rate $q_z$. Moreover, for distinct symbols
$z_1,\ldots,z_m$, the processes
\[
    (N_{z_1}(t))_{t\geq 0},\ldots,(N_{z_m}(t))_{t\geq 0}
\]
are independent, because they count disjoint subsets of the underlying Poisson process.

Therefore, for $z$ with $q_z>0$, the first arrival time
\[
    S_z=\inf\{t:N_z(t)\geq 1\}
\]
satisfies
\[
    \mathbf{P}(S_z>s)
    =
    \mathbf{P}(N_z(s)=0)
    =
    e^{-q_z s},
    \qquad s\geq 0.
\]
Thus $S_z\sim\mathrm{Exp}(q_z)$. Hence
\[
    E_z=q_z S_z\sim \mathrm{Exp}(1).
\]
The independence of the Poisson counting processes for different symbols implies that the random variables
\[
    (E_z)_{z:q_z>0}
\]
are independent. For symbols with $q_z=0$, we have $p_z=0$ because $P\ll Q$, and these symbols can never be selected by PFR. Defining independent $\mathrm{Exp}(1)$ variables $E_z$ for those symbols therefore does not change the output. Hence the full family $(E_z)_{z\in\mathcal{Z}}$ is independent and identically distributed as $\mathrm{Exp}(1)$.

It remains to show that the PFR selection rule reduces to the EFR selection rule. Since $P\ll Q$, for every $z$ with $p_z>0$, we also have $q_z>0$, and
\[
    \frac{\mathrm{d}P}{\mathrm{d}Q}(z)
    =
    \frac{p_z}{q_z}.
\]
Therefore, for every index $i$ with $Z_i=z$ and $p_z>0$,
\[
    \tilde T_i
    =
    T_i\left(\frac{p_z}{q_z}\right)^{-1}
    =
    T_i\frac{q_z}{p_z}.
\]
Taking the minimum over all Poisson points with mark $z$, we get
\[
    \inf_{i:Z_i=z}\tilde T_i
    =
    \left(\inf_{i:Z_i=z} T_i\right)\frac{q_z}{p_z}
    =
    S_z\frac{q_z}{p_z}
    =
    \frac{E_z}{p_z}.
\]
If $p_z=0$, then $\mathrm{d}P/\mathrm{d}Q(z)=0$ whenever $q_z>0$, and hence
\[
    \tilde T_i
    =
    T_i\left(\frac{\mathrm{d}P}{\mathrm{d}Q}(z)\right)^{-1}
    =
    +\infty
\]
for every Poisson point with mark $z$. This agrees with the EFR convention
\[
    \frac{E_z}{p_z}=+\infty .
\]
If $q_z=0$, then also $p_z=0$, and no Poisson point has mark $z$ almost surely, so such a symbol is irrelevant in both constructions.

Thus,
\[
    \inf_i \tilde T_i
    =
    \inf_{z\in\mathcal{Z}}\inf_{i:Z_i=z}\tilde T_i
    =
    \inf_{z\in\mathcal{Z}}\frac{E_z}{p_z}.
\]
Moreover, the minimizer is unique almost surely. Indeed, for every $z$ with $p_z>0$,
\[
    \frac{E_z}{p_z}\sim \mathrm{Exp}(p_z),
\]
and these variables are independent. Since exponential random variables are continuous, ties have probability zero. Therefore, almost surely, 
\[
    Z_{\mathrm{PFR}}
    =
    Z_K
    =
    \arg\min_{z\in\mathcal{Z}}\frac{E_z}{p_z}
    =
    Z_{\mathrm{EFR}}.  
\]
\end{proof}

\subsection{Relation between PFR and Gumbel-Max Trick}
\label{subsapp::PFR_gumbel}

As discussed in~\citep{zhao2024permute}, the Gumbel Watermark~\citep{aaronson2023watermarking} is based on  
\begin{equation*}
    y_t = \mathrm{argmin}_{y\in\mathcal{V}} \frac{u_t(y)}{T} + G_t(y)
\end{equation*}
where $G(y)\sim \mathrm{Gumbel}(0,1)$ i.i.d. for each $t$ and $y$. 
Here, we show that it is equivalent to the PFR in the discrete case. 
Recall $\mathrm{Gumbel}(0,1) \sim -\log\big(\log(1/\mathrm{Unif}[0,1]) \big)$. 
We calculate
\begin{align*}
y_t
&= \arg\max_{y\in\mathcal{V}}\left(\frac{u_t(y)}{T} + G_t(y)\right)\\
&= \arg\max_{y\in\mathcal{V}}\left(\frac{u_t(y)}{T} - \log Z_t(y)\right)
\qquad\Bigl(Z_t(y) \stackrel{\mathrm{iid}}{\sim} \mathrm{Exp}(1),\; G_t(y)\stackrel{d}{=}-\log Z_t(y)\Bigr)\\
&= \arg\max_{y\in\mathcal{V}}\left(-\log \bigl(Z_t(y)e^{-u_t(y)/T}\bigr)\right)\\
&= \arg\min_{y\in\mathcal{V}} Z_t(y)e^{-u_t(y)/T}\\
&= \arg\min_{y\in\mathcal{V}} \frac{Z_t(y)}{\exp(u_t(y)/T)}\\
&= \arg\min_{y\in\mathcal{V}} \frac{Z_t(y)}{w_t(y)}
\qquad\Bigl(w_t(y):=\exp(u_t(y)/T)\Bigr)\\
&= \arg\min_{y\in\mathcal{V}} \frac{Z_t(y)}{P_t(y)},
\end{align*} 
where the last equality is because $P_t(y):=\tfrac{w_t(y)}{\sum_{y'\in\mathcal{V}} w_t(y')}$ and scaling by a positive constant doesn't change the argmin.

\subsection{Mapped Poisson Process}
\label{subapp::mapped_PP}

\subsubsection{Definitions}

\begin{definition}[Mapped Poisson Process~\citep{li2021unified}] \label{def::mapped_PP}
Let $(T_{i})_{i}$ be a Poisson process of rate $1$,
independent of $Z_{i}\stackrel{\mathrm{iid}}{\sim}Q$ for $i=1,2,\ldots$. 
Let $i_{P, 1},i_{P, 2}, \ldots\in\mathbb{N}$ be a sequence of distinct integers such that $\bigcup_{j=1}^\infty \big\{i_{P, j} \big\} = \big\{i:\frac{\mathrm{d}P}{\mathrm{d}Q}(Z_{i})>0 \big\}$ and $\big\{T_{i_{P, j}} \big(\frac{\mathrm{d}P}{\mathrm{d}Q}(Z_{i_{P, j}})\big)^{-1} \big\}_{j \in \mathbb{N}}$ is sorted in ascending order with arbitrary tie-breaking.\footnote{Note a tie occurs with probability $0$.}
For $j \in\mathbb{N}$ and suppose $P$ is the target distribution, we denote
\begin{equation}
    \left\{
    \tilde{Z}_P(j), \tilde{T}(j)
    \right\}_{j\in\mathbb{N}} 
    := \left\{ Z_{i_{P, j}}, T_{i_{P, j}}
    \big(\frac{\mathrm{d}P}{\mathrm{d}Q}(Z_{i_{P, j}})\big)^{-1} \right\}_{j\in\mathbb{N}}. 
\end{equation}
Give a positive integer $B$, the mapped Poisson Process selects the first $B$ samples in $\left\{ \tilde{Z}_P(j), \tilde{T}(j) \right\}_{j\in\mathbb{N}}$. 
\end{definition}

\begin{proposition}
\label{prop:ppr_output_dist}
By the mapping theorem~\citep{last2017lectures}, $\big\{ \tilde{Z}_P(j), \tilde{T}(j) \big\}_{j\in\mathbb{N}}$ is a Poisson process with intensity measure $P\times \lambda_{\mathbb{R}_{\geq 0}}$ where $\mathbb{R}_{\geq 0}$ to denote nonnegative real numbers, and therefore we have 
\begin{equation}
    \tilde{Z}_P (1), \tilde{Z}_P (2), \ldots \stackrel{\mathrm{iid}}{\sim} P.
\end{equation}
\end{proposition}

\subsubsection{Direct Implementations}
\label{subapp::direct_implementation_MPFR}
We can implement the MPFR just like~\citep{theis2022algorithms} for PFR, as follows: 

\begin{algorithm}[H]
\SetAlgoLined
\DontPrintSemicolon
\SetAlgoNoEnd
\SetKwIF{If}{ElseIf}{Else}{if}{:}{else if}{else}{end}
\SetKwInOut{Input}{Given}
\SetKwInOut{Output}{Return}

\Input{target distribution $P$, proposal distribution $Q$, likelihood ratio
$r=dP/dQ$, number of samples $B$, and a constant $\underline w>0$ such that
$1/r(u)\ge \underline w$ $Q$-a.s. on $\{r(u)>0\}$}
\Output{$B$ samples from $P$}

$H \gets \emptyset$ \tcp*[r]{$H$ stores the current $B$ smallest-score candidates}
$t \gets 0$; $i \gets 1$\;

\While{$|H|<B$ \textbf{or} $\max_{(s,j,z)\in H} s > t\underline w$}{
    Sample $E_i\sim \mathrm{Exp}(1)$ and set $t\gets t+E_i$ \tcp*[r]{Poisson process}
    
    Sample $Z_i\sim Q$\;
    
    \eIf{$r(Z_i)>0$}{
        $S_i \gets t/r(Z_i)$\;
    }{
        $S_i \gets +\infty$\;
    }

    \If{$|H|<B$}{
        Insert $(S_i,i,Z_i)$ into $H$\;
    }
    \ElseIf{$S_i < \max_{(s,j,z)\in H} s$}{
        Remove an element $(s^\star,j^\star,z^\star)\in H$ satisfying
        $s^\star=\max_{(s,j,z)\in H} s$\;
        
        Insert $(S_i,i,Z_i)$ into $H$\;
    }

    $i\gets i+1$\;
}
Sort the elements of $H$ as
\[
(S_{i_1},i_1,Z_{i_1}),\ldots,(S_{i_B},i_B,Z_{i_B})
\]
so that
\[
S_{i_1}\le S_{i_2}\le \cdots \le S_{i_B}.
\]

\Return{$(Z_{i_1},\ldots,Z_{i_B})$}

\caption{Finite implementation of the $B$-sample mapped PFR}
\label{alg:B_sample_finite_PFR}
\end{algorithm}

\subsubsection{Proof of Correctness}

We then prove that the implementation of Algorithm~\ref{alg:B_sample_finite_PFR} is correct.

\begin{theorem}
Let $P\ll Q$ be probability measures on a Polish space $\mathcal U$, and let
$r:=dP/dQ$. Let $Z_1,Z_2,\ldots\stackrel{iid}{\sim} Q$ and let
$0<T_1<T_2<\cdots$ be the arrival times of a unit-rate Poisson process,
independent of $(Z_i)_{i\ge1}$. Define
\[
S_i :=
\begin{cases}
T_i/r(Z_i), & r(Z_i)>0,\\
+\infty, & r(Z_i)=0.
\end{cases}
\]
Assume that there exists $\underline w>0$ such that
\[
1/r(u)\ge \underline w
\qquad Q\text{-a.s. on }\{r(u)>0\}.
\] 
Fix $B\in\mathbb N$. Consider the algorithm that scans the candidates in
the order $i=1,2,\ldots$, keeps the $B$ smallest observed scores, and stops
at the first time $n$ such that the current $B$-th smallest observed score
is at most $T_n\underline w$. Then the algorithm terminates almost surely
after finitely many iterations. Moreover, at termination, the retained
candidates are exactly the $B$ smallest-score candidates among the infinite
sequence. Consequently, the returned marks are distributed as
\[
(\widetilde U_P(1),\ldots,\widetilde U_P(B)),
\]
the first $B$ points of the mapped Poisson process with respect to $P$,
and hence are i.i.d. according to $P$.
\end{theorem}

\begin{proof}
Let
\[
\Pi_Q:=\{(Z_i,T_i):i\ge 1\}.
\]
Since $0<T_1<T_2<\cdots$ are the arrival times of a unit-rate Poisson
process and $Z_1,Z_2,\ldots\stackrel{iid}{\sim}Q$ are independent marks,
$\Pi_Q$ is a Poisson point process on $\mathcal U\times \mathbb R_{\ge 0}$
with intensity measure $Q\times \lambda$, where $\lambda$ denotes the
Lebesgue measure on $\mathbb R_{\ge 0}$.

Consider the measurable mapping
\[
\Phi:\mathcal U\times \mathbb R_{\ge 0}\to
\mathcal U\times \mathbb R_{\ge 0},
\qquad
\Phi(u,t)=\left(u,\frac{t}{r(u)}\right)
\]
on the set $\{u:r(u)>0\}$. Points with $r(u)=0$ are assigned score
$+\infty$ and hence never appear among the finite mapped arrival times.
For any nonnegative measurable function $g$ on
$\mathcal U\times \mathbb R_{\ge 0}$, we have
\[
\begin{aligned}
\int_{\mathcal U}\int_0^\infty
g\left(u,\frac{t}{r(u)}\right)\,dt\,Q(du)
&=
\int_{\{r>0\}}\int_0^\infty
g\left(u,\frac{t}{r(u)}\right)\,dt\,Q(du) \\
&=
\int_{\{r>0\}}\int_0^\infty
g(u,s)r(u)\,ds\,Q(du) \\
&=
\int_{\mathcal U}\int_0^\infty
g(u,s)\,ds\,P(du).
\end{aligned}
\]
Therefore, by the mapping theorem for Poisson point processes, the mapped
point process
\[
\Pi_P:=\left\{\left(Z_i,S_i\right):i\ge 1,\ r(Z_i)>0\right\},
\qquad
S_i:=\frac{T_i}{r(Z_i)},
\]
is a Poisson point process on $\mathcal U\times \mathbb R_{\ge 0}$ with
intensity measure $P\times \lambda$.

Let
\[
S_{i_1}<S_{i_2}<\cdots
\]
be the increasing ordering of the finite scores in $\Pi_P$. Since the
time-marginal intensity of $\Pi_P$ is $\lambda$, the ordered times
$S_{i_1},S_{i_2},\ldots$ are the arrival times of a unit-rate Poisson
process. Moreover, conditional on these arrival times, the associated marks
are independent with common distribution $P$. Hence
\[
(Z_{i_1},Z_{i_2},\ldots)\stackrel{d}{=}(U_1,U_2,\ldots),
\]
where $U_1,U_2,\ldots\stackrel{iid}{\sim}P$. In particular,
\[
(Z_{i_1},\ldots,Z_{i_B})
\]
are exactly the first $B$ points of the mapped Poisson process with respect
to $P$, and they are i.i.d. according to $P$.

It remains to show that the finite algorithm stops safely and terminates
almost surely. After observing candidates $1,\ldots,n$, let
\[
S_{n,(B)}
\]
denote the $B$-th smallest value among $S_1,\ldots,S_n$, whenever at least
$B$ finite scores have been observed. By construction, the algorithm keeps
the $B$ smallest observed scores, so the largest score retained by the
algorithm is exactly $S_{n,(B)}$.

Suppose the algorithm stops at time $n$. Then
\[
S_{n,(B)}\le T_n\underline w.
\]
For every future candidate $m>n$, we have $T_m>T_n$. If $r(Z_m)=0$, then
$S_m=+\infty$, so such a candidate cannot improve the current list. If
$r(Z_m)>0$, then by the assumption $1/r(u)\ge \underline w$ $Q$-a.s.,
\[
S_m
=
\frac{T_m}{r(Z_m)}
=
T_m\frac{1}{r(Z_m)}
\ge
T_m\underline w
>
T_n\underline w
\ge
S_{n,(B)}.
\]
Thus every future score is strictly larger than the current $B$-th smallest
observed score. Therefore no future candidate can enter the set of the
global $B$ smallest scores. Hence, when the algorithm stops, the retained
candidates are exactly the global $B$ smallest-score candidates among the
infinite sequence.

Finally, we prove almost sure finite termination. Since $\Pi_P$ is a
Poisson point process with intensity $P\times \lambda$, its first $B$ mapped
arrival times are finite almost surely. Let
\[
S_{i_B}
\]
denote the $B$-th smallest mapped score, and let
\[
N_0:=\max\{i_1,\ldots,i_B\}.
\]
Then $N_0<\infty$ almost surely. After the algorithm has processed
candidates $1,\ldots,N_0$, it has already observed the true global first
$B$ mapped candidates, so the largest retained score is $S_{i_B}$. Since
\[
T_n\to \infty
\qquad\text{almost surely},
\]
and since $\underline w>0$, there exists an almost surely finite random
index $N\ge N_0$ such that
\[
T_N\underline w\ge S_{i_B}.
\]
At time $N$, the algorithm therefore satisfies the stopping condition
\[
S_{N,(B)}=S_{i_B}\le T_N\underline w.
\]
Thus the algorithm terminates almost surely after finitely many iterations.

Combining the safe stopping argument with the mapped Poisson process
argument, the algorithm returns exactly
\[
(Z_{i_1},\ldots,Z_{i_B}),
\]
which are the first $B$ points of the mapped Poisson process with respect to
$P$. Consequently, the returned samples are i.i.d. according to $P$.
\end{proof}

\subsection{A better finite-support MPFR implementation}
\label{subsec::direct_topk_MPFR}

The proposal-scanning implementation in Appendix~\ref{subapp::direct_implementation_MPFR} can be slow in finite-vocabulary LLM decoding. 
Its stopping rule depends on a lower bound on the likelihood ratio, which can be very small when a rare token has small proposal probability but non-negligible target probability. 
In our experiments, the logits are processed by temperature and top-$k$ truncation, so the resulting distribution has finite support of size at most $K$. 
In this setting, we can simulate the relevant Poisson clocks directly and avoid proposal scanning altogether.

Let $P$ be a distribution on a finite vocabulary $\mathcal U$, and define
\[
\mathcal A_P:=\{u\in\mathcal U:P(u)>0\}.
\]
For each $u\in\mathcal A_P$, let
\[
G_{u,j}:=\sum_{m=1}^{j}E_{u,m},\qquad j\geq 1,
\]
where the $E_{u,m}$'s are independent $\mathrm{Exp}(1)$ random variables. 
We associate to $(u,j)$ the mapped-Poisson score
\[
S_{u,j}:=\frac{G_{u,j}}{P(u)}.
\]
The first $B$ MPFR samples are the labels of the $B$ smallest scores among
\[
\{S_{u,j}:u\in\mathcal A_P,\ j\geq 1\}.
\]
Since no token can contribute more than $B$ arrivals among the first $B$ total arrivals, it is enough to generate $j=1,\ldots,B$ for every $u\in\mathcal A_P$.

\begin{algorithm}[tbp]
\SetAlgoLined
\DontPrintSemicolon
\SetAlgoNoEnd
\SetKwInOut{Input}{Given}
\SetKwInOut{Output}{Return}

\Input{finite distribution $P$ on $\mathcal U$, number of samples $B$, and keyed source $\Pi$}
\Output{$B$ samples from $P$}

$\mathcal A_P\gets\{u\in\mathcal U:P(u)>0\}$\;
$H\gets\emptyset$\tcp*[r]{candidate triples $(S,u,j)$}

\For{each $u\in\mathcal A_P$}{
    $G\gets 0$\;
    \For{$j=1$ \KwTo $B$}{
        Generate $E_{u,j}\sim\mathrm{Exp}(1)$ from $\Pi$\;
        $G\gets G+E_{u,j}$\;
        $S\gets G/P(u)$\;
        Insert $(S,u,j)$ into $H$\;
    }
}

Let
\[
(S_{(1)},U_{(1)},J_{(1)}),\ldots,(S_{(B)},U_{(B)},J_{(B)})
\]
be the $B$ elements of $H$ with smallest scores, ordered by increasing score\;

\Return{$(U_{(1)},\ldots,U_{(B)})$}

\caption{Direct finite-support implementation of $B$-sample MPFR}
\label{alg:direct_finite_support_MPFR}
\end{algorithm}

\begin{theorem}[Exactness of direct finite-support MPFR]
\label{thm::direct_MPFR_exactness}
Let $P$ be a distribution on a finite set $\mathcal U$. 
The output $(U_{(1)},\ldots,U_{(B)})$ of Algorithm~\ref{alg:direct_finite_support_MPFR} consists of $B$ independent samples from $P$. 
That is, for every $(u_1,\ldots,u_B)\in\mathcal U^B$,
\[
\mathbf P\{U_{(1)}=u_1,\ldots,U_{(B)}=u_B\}
=
\prod_{b=1}^B P(u_b).
\]
\end{theorem}

\begin{proof}
For each $u\in\mathcal A_P$, the sequence $(G_{u,j})_{j\geq 1}$ is the arrival process of a unit-rate Poisson process. 
Therefore,
\[
\left(\frac{G_{u,j}}{P(u)}\right)_{j\geq 1}
\]
is a Poisson process on $\mathbb R_+$ with rate $P(u)$, since
\[
\#\left\{j:\frac{G_{u,j}}{P(u)}\leq t\right\}
=
\#\{j:G_{u,j}\leq P(u)t\}
\]
is Poisson with mean $P(u)t$. 
The processes corresponding to different tokens are independent. 
Their superposition is therefore a marked Poisson process on $\mathbb R_+\times\mathcal U$ with intensity measure
\[
dt\otimes P.
\]
Equivalently, the total arrival rate is $\sum_{u\in\mathcal A_P}P(u)=1$, and the marks of the ordered arrivals are independent with common distribution $P$.

Algorithm~\ref{alg:direct_finite_support_MPFR} constructs the first $B$ candidate arrivals of each token clock. 
This is sufficient because among the first $B$ arrivals of the superposed process, no individual token clock can contribute more than $B$ arrivals. 
Hence the $B$ smallest elements in $H$ are exactly the first $B$ arrivals of the superposed marked Poisson process, and their labels are $B$ independent samples from $P$.
\end{proof}

We now describe the batched speculative decoding implementation. 
For a context $c$, let $P_k(\cdot|c)$ and $Q_k(\cdot|c)$ denote the processed target and draft distributions after applying the same logits processing, such as temperature and top-$k$ truncation. 
The draft side builds a tree of candidate continuations using Algorithm~\ref{alg:direct_finite_support_MPFR} with $Q_k$. 
The target side then evaluates all contexts needed for verification in a batched manner and uses the same MPFR clocks with $P_k$ to follow the realized target path. 
This batching changes only the order and grouping of model evaluations, not the sampling law.

\begin{algorithm}[tbp]
\SetAlgoLined
\DontPrintSemicolon
\SetAlgoNoEnd
\SetKwIF{If}{ElseIf}{Else}{if}{:}{else if}{else}{end}
\SetKwInOut{Input}{Given}

\Input{lookahead $L$, number of drafts $B$, output length $N$, target model $P$, draft model $Q$, prompt $w_{1:n}$, key $\mathsf{k}$}
% \Output{generated continuation from the processed target model $P_k$}

\While{$n<N$}{
    $h\gets w_{1:n}$, $\ell\gets \min\{L,N-n\}$\;
    $\mathcal C_0\gets\{h\}$ and $\nu(h)\gets B$\tcp*[r]{context multiplicities}

    \tcp{Draft tree construction}
    \For{$s=1$ \KwTo $\ell$}{
        $\mathcal C_s\gets\emptyset$\;
        \For{each $c\in\mathcal C_{s-1}$}{
            Generate
            \[
            (X_1(c),\ldots,X_{\nu(c)}(c))
            \gets
            \textsc{MPFR}(Q_k(\cdot|c),\Pi_{\mathsf{k}}(c),\nu(c)).
            \]
            Let $\mathcal D(c)$ be the set of distinct tokens among these samples\;
            \For{each $u\in\mathcal D(c)$}{
                $m(c,u)\gets |\{i:X_i(c)=u\}|$\;
                Add $c\|u$ to $\mathcal C_s$\;
                $\nu(c\|u)\gets \nu(c\|u)+m(c,u)$\;
            }
        }
    }

    \tcp{Batched target evaluation}
    Compute the processed target distributions $P_k(\cdot|c)$ for all contexts
    \[
    c\in \bigcup_{s=0}^{\ell}\mathcal C_s
    \]
    using batched target-model forward passes\;
    \For{each such context $c$}{
        Generate
        \[
        Y(c)\gets
        \textsc{MPFR}(P_k(\cdot|c),\Pi_{\mathsf{k}}(c),1).
        \]
    }

    \tcp{Follow the realized target path}
    $c^\star\gets h$, $\mathrm{accepted}\gets\textbf{true}$\;
    \For{$s=1$ \KwTo $\ell$}{
        Emit $Y(c^\star)$ and append it to the output\;
        \If{$Y(c^\star)\notin\mathcal D(c^\star)$}{
            $\mathrm{accepted}\gets\textbf{false}$\;
            \textbf{break}\tcp*[r]{first rejection ends the block}
        }
        \If{$n=N$}{
            \textbf{break}\;
        }
        $c^\star\gets c^\star\|Y(c^\star)$\;
    }

    \If{$\mathrm{accepted}$ \textbf{and} $n<N$}{
        Emit the precomputed bonus token $Y(c^\star)$\;
    }
}

\caption{Batched-target direct-top-$k$ MPFR speculative sampling}
\label{alg:batched_direct_topk_MPFR_spec}
\end{algorithm}

\begin{theorem}[Exactness of batched-target direct-top-$k$ MPFR decoding]
\label{thm::batched_direct_topk_MPFR_spec_exactness}
Assume that for every context $c$, the keyed source $\Pi_{\mathsf{k}}(c)$ provides independent exponential clocks across tokens and local arrival indices, and that the sources $\{\Pi_{\mathsf{k}}(c)\}_c$ are independent across distinct contexts. 
Then Algorithm~\ref{alg:batched_direct_topk_MPFR_spec} generates exactly from the autoregressive model with transition kernel $P_k(\cdot|c)$. 
That is, for every finite sequence $u_1,\ldots,u_m$,
\[
\mathbf P\{W_{n+1}=u_1,\ldots,W_{n+m}=u_m\}
=
\prod_{t=1}^{m}
P_k(u_t|w_{1:n},u_1,\ldots,u_{t-1}).
\]
\end{theorem}

\begin{proof}
Fix a context $c$. 
By Theorem~\ref{thm::direct_MPFR_exactness},
\[
Y(c)=\textsc{MPFR}(P_k(\cdot|c),\Pi_{\mathsf{k}}(c),1)
\]
has distribution $P_k(\cdot|c)$. 
The draft samples generated from $Q_k(\cdot|c)$ may use the same keyed clocks, and hence are coupled with $Y(c)$, but this does not change the marginal distribution of $Y(c)$.

Let
\[
C_t:=(w_{1:n},W_{n+1},\ldots,W_{n+t-1})
\]
be the random context before generating $W_{n+t}$. 
The event $\{C_t=c\}$ is determined by the keyed sources associated with strict prefix contexts of $c$. 
By the assumed independence across contexts, $\Pi_{\mathsf{k}}(c)$ is independent of this event. 
Therefore, conditional on $C_t=c$,
\[
W_{n+t}=Y(c)\sim P_k(\cdot|c),
\]
and hence
\[
\mathbf P\{W_{n+t}=u|C_t=c\}=P_k(u|c).
\]
Applying the chain rule gives
\[
\begin{aligned}
&\mathbf P\{W_{n+1}=u_1,\ldots,W_{n+m}=u_m\} \\
&\qquad =
\prod_{t=1}^{m}
\mathbf P\{W_{n+t}=u_t
| W_{n+1}=u_1,\ldots,W_{n+t-1}=u_{t-1}\} \\
&\qquad =
\prod_{t=1}^{m}
P_k(u_t|w_{1:n},u_1,\ldots,u_{t-1}).
\end{aligned}
\]
Thus the generated sequence has exactly the desired autoregressive law. 
Finally, batched target evaluation only changes the computational order of deterministic model evaluations. 
It does not change the context-indexed MPFR clocks or the resulting samples $Y(c)$, so it does not change the output distribution.
\end{proof}

This implementation removes the proposal-scanning overhead by simulating the token-wise Poisson clocks directly. 
For support size at most $K$, producing $B$ MPFR samples requires $KB$ exponential variables followed by a top-$B$ selection. 
The batched target implementation further reduces the number of target-model calls: instead of verifying the realized path sequentially with roughly one target forward pass per generated token, it evaluates the contexts in a speculative block together and extracts all logits needed for verification. 
If a block emits $\mathrm{BE}$ tokens on average, the target-forward cost per emitted token is reduced from approximately $1$ to approximately $1/\mathrm{BE}$, while preserving the exact target marginal distribution.

\section{Theoretic Guarantees}
\label{app::theory}

\subsection{Proof of Theorem~\ref{thm::PML_accept}}
\label{subapp::proof_PML_accept}

\begin{proof}
The proof is based on the Poisson matching lemma~\citep{li2021unified}. 
Let $p_i := P(i)$ and $q_i := Q(i)$ for $i\in\mathcal{V}$. 
Let $Y := \tilde Z_P(1)$ and $X^{(b)} := \tilde Z_Q(b)$ for $b=1,\ldots,B$ be generated from the same mapped Poisson process, and define the token-level acceptance event
\[
\mathcal{A}_t := \Big\{Y \in \{X^{(1)},\ldots,X^{(B)}\}\Big\}, 
\]
which is the acceptance event. Apply the ~\citep[Lemma 3]{li2021unified} with $j=1$ and $k=B$.
Conditioned on $Y=\tilde Z_P(1)=i$, the event $\mathcal{A}_t^c$ is exactly the event that
the first mapped $P$-sample does not appear among the first $B$ mapped $Q$-samples.
Hence
\[
\mathbf{P}(\mathcal{A}_t^c  | Y=i)
\le
\left(1-\left(1+\frac{p_i}{q_i}\right)^{-1}\right)^B
=
\left(\frac{p_i}{p_i+q_i}\right)^B.
\]
Taking complements gives
\[
\mathbf{P}(\mathcal{A}_t  | Y=i)
\ge
1-\left(\frac{p_i}{p_i+q_i}\right)^B.
\]
Finally, average over $Y\sim P$:
\[
\mathbf{P}(\mathcal{A}_t)
=
\sum_{i\in\mathcal{V}} p_i \mathbf{P}(\mathcal{A}_t  | Y=i)
\ge
\sum_{i\in\mathcal{V}}
p_i\left[1-\left(\frac{p_i}{p_i+q_i}\right)^B\right].
\]
For $B=1$, this simplifies to
\[
1-\sum_{i\in\mathcal{V}}\frac{p_i^2}{p_i+q_i}
=
\sum_{i\in\mathcal{V}}\frac{p_iq_i}{p_i+q_i}.
\qedhere
\]
\end{proof}

\subsection{Comparison with~\citep{rowan2025list} on Acceptance Probability}
\label{subapp::compare_Rowan_acpt_prob}

We compare our token-level acceptance probability with the token-level acceptance probability from \citep{rowan2025list} in two levels, and we prove that in each level our algorithm's acceptance probability is strictly better.

\subsubsection{}

We first prove a stronger version of  Theorem~\ref{thm::PML_accept}, but not as clean as it.

\begin{theorem}
\label{thm::PML_accept_exact}
Let $P$ be the target distribution and $Q$ be the draft distribution on a finite
vocabulary $\mathcal V$. Assume first that $P(i),Q(i)>0$ for all
$i\in\mathcal V$. For each $i\in\mathcal V$, define
\[
    \alpha_i
    :=
    P(i)\sum_{j\in\mathcal V}
    \left(
        \frac{Q(j)}{Q(i)}-\frac{P(j)}{P(i)}
    \right)_+ .
\]
If the target token is generated as $\tilde U_P(1)$ and the $B$ draft tokens are
generated as $\tilde U_Q(1),\ldots,\tilde U_Q(B)$ from the same underlying
Poisson process, then the acceptance probability of at least one token satisfies
\begin{equation}
    \label{eq::PML_accept_exact}
    \mathbf P(\mathrm{accept})
    \ge
    1-\sum_{i\in\mathcal V}
    P(i)
    \left(
        \frac{\alpha_i}{1+\alpha_i}
    \right)^B .
\end{equation}
\end{theorem}

\begin{proof}
Let the target token be
\[
    Y:=\tilde U_P(1),
\]
and let the $B$ draft tokens be
\[
    X_b:=\tilde U_Q(b),\qquad b=1,\ldots,B.
\]
The acceptance event is
\[
    \mathrm{accept}
    =
    \{Y\in\{X_1,\ldots,X_B\}\}.
\]
Conditioning on $Y=i$, the failure event is
\[
    \{Y\notin\{X_1,\ldots,X_B\}\}.
\]
By the exact $j=1$ tail formula in the generalized Poisson matching lemma,
\[
    \mathbf P\!\left(
        Y\notin\{X_1,\ldots,X_B\}
        \,\middle|\, Y=i
    \right)
    \le
    \left(
        \frac{\alpha_i}{1+\alpha_i}
    \right)^B,
\]
where
\[
    \alpha_i
    =
    P(i)\sum_{j\in\mathcal V}
    \left(
        \frac{Q(j)}{Q(i)}-\frac{P(j)}{P(i)}
    \right)_+ .
\]
Therefore,
\[
    \mathbf P(\mathrm{accept}| Y=i)
    \ge
    1-
    \left(
        \frac{\alpha_i}{1+\alpha_i}
    \right)^B .
\]
Averaging over $Y\sim P$ gives
\[
\begin{aligned}
    \mathbf P(\mathrm{accept})
    &=
    \sum_{i\in\mathcal V} P(i)
    \mathbf P(\mathrm{accept}| Y=i)\\
    &\ge
    \sum_{i\in\mathcal V} P(i)
    \left[
        1-
        \left(
            \frac{\alpha_i}{1+\alpha_i}
        \right)^B
    \right]\\
    &=
    1-\sum_{i\in\mathcal V}
    P(i)
    \left(
        \frac{\alpha_i}{1+\alpha_i}
    \right)^B .
\end{aligned}
\]
This proves~\eqref{eq::PML_accept_exact}.
\end{proof}

Moreover,
\begin{equation}
    \label{eq::PML_accept_exact_implies_weak}
    1-\sum_{i\in\mathcal V}
    P(i)
    \left(
        \frac{\alpha_i}{1+\alpha_i}
    \right)^B
    \ge
    1-\sum_{i\in\mathcal V}
    P(i)
    \left(
        \frac{P(i)}{P(i)+Q(i)}
    \right)^B .
\end{equation}
Thus~\eqref{eq::PML_accept_exact} strictly improves the weaker bound
\[
    \mathbf P(\mathrm{accept}) \ge
    1-\sum_{i\in\mathcal V} P(i)
    \left(
        \frac{P(i)}{P(i)+Q(i)}
    \right)^B
\]
whenever the inequality in~\eqref{eq::PML_accept_exact_implies_weak} is strict.

We then show above Theorem~\ref{thm::PML_accept_exact} is tighter than \citep[Proposition 2]{rowan2025list}. 

\begin{theorem}
\label{thm::exact_PML_dominates_LML}
Let $P$ be the target distribution and $Q$ be the draft distribution on a finite
vocabulary $\mathcal V$. Write
\[
    p_i := P(i),\qquad q_i := Q(i),\qquad i\in\mathcal V,
\]
and assume first that $p_i,q_i>0$ for all $i\in\mathcal V$. For each
$i\in\mathcal V$, define
\[
    \alpha_i
    :=
    p_i\sum_{j\in\mathcal V}
    \left(
        \frac{q_j}{q_i}-\frac{p_j}{p_i}
    \right)_+ .
\]
Let
\[
    L_{\mathrm{PML}}
    :=
    1-\sum_{i\in\mathcal V}
    p_i
    \left(
        \frac{\alpha_i}{1+\alpha_i}
    \right)^B
\]
be the acceptance lower bound obtained from the exact $j=1$ tail formula of the
generalized Poisson matching lemma. Let
\[
    L_{\mathrm{LML}}
    :=
    \sum_{i\in\mathcal V}
    \frac{B}{
    \sum_{j\in\mathcal V}
    \left[
        \max\left\{
            \frac{p_j}{p_i},
            \frac{q_j}{q_i}
        \right\}
        +(B-1)\frac{p_j}{p_i}
    \right]}
\]
be the lower bound in Eq.~(3) of the list matching lemma, written in the same
target--draft notation. Then
\[
    L_{\mathrm{PML}}\ge L_{\mathrm{LML}}.
\]
Moreover, if $B>1$ and $P\neq Q$, then
\[
    L_{\mathrm{PML}}>L_{\mathrm{LML}}.
\]
\end{theorem}

\begin{proof}
The proof is based on the Poisson matching lemma~\citep{li2021unified}.
Let $Y:=\tilde Z_P(1)$ and $X^{(b)}:=\tilde Z_Q(b)$ for
$b=1,\ldots,B$ be generated from the same mapped Poisson process. By the
mapping theorem, $\tilde Z_Q(1),\tilde Z_Q(2),\ldots$ are i.i.d. according
to $Q$, and $\tilde Z_P(1)\sim P$. Define the token-level acceptance event
\[
    \mathcal A_t
    :=
    \Big\{
        Y\in\{X^{(1)},\ldots,X^{(B)}\}
    \Big\}.
\]
Apply the exact $j=1$ tail formula of the generalized Poisson matching lemma
with the first distribution equal to $P$ and the second distribution equal to
$Q$. Conditioned on $Y=\tilde Z_P(1)=i$, the event $\mathcal A_t^c$ is the
event that the first mapped $P$-sample does not appear among the first $B$
mapped $Q$-samples. Hence
\[
    \mathbf{P}(\mathcal A_t^c| Y=i)
    \le
    \left(
        \frac{\alpha_i}{1+\alpha_i}
    \right)^B,
\]
where
\[
    \alpha_i
    =
    p_i\sum_{j\in\mathcal V}
    \left(
        \frac{q_j}{q_i}-\frac{p_j}{p_i}
    \right)_+ .
\]
Taking complements gives
\[
    \mathbf{P}(\mathcal A_t| Y=i)
    \ge
    1-
    \left(
        \frac{\alpha_i}{1+\alpha_i}
    \right)^B .
\]
Finally, averaging over $Y\sim P$ gives
\[
\begin{aligned}
    \mathbf{P}(\mathcal A_t)
    &=
    \sum_{i\in\mathcal V}
    p_i\mathbf{P}(\mathcal A_t| Y=i)\\
    &\ge
    \sum_{i\in\mathcal V}
    p_i
    \left[
        1-
        \left(
            \frac{\alpha_i}{1+\alpha_i}
        \right)^B
    \right]\\
    &=
    1-\sum_{i\in\mathcal V}
    p_i
    \left(
        \frac{\alpha_i}{1+\alpha_i}
    \right)^B
    =
    L_{\mathrm{PML}}.
\end{aligned}
\]

We now rewrite the LML bound in terms of the same quantities $\alpha_i$.
For each fixed $i\in\mathcal V$,
\begin{align*}
    p_i
    \sum_{j\in\mathcal V}
    \max\left\{
        \frac{p_j}{p_i},
        \frac{q_j}{q_i}
    \right\}
    &=
    \sum_{j\in\mathcal V}
    \max\left\{
        p_j,
        \frac{p_i}{q_i}q_j
    \right\} \\
    &=
    \sum_{j\in\mathcal V}
    \left[
        p_j+
        \left(
            \frac{p_i}{q_i}q_j-p_j
        \right)_+
    \right] \\
    &=
    1+
    p_i\sum_{j\in\mathcal V}
    \left(
        \frac{q_j}{q_i}-\frac{p_j}{p_i}
    \right)_+ \\
    &=
    1+\alpha_i.
\end{align*}
Therefore,
\begin{align*}
    \sum_{j\in\mathcal V}
    \left[
        \max\left\{
            \frac{p_j}{p_i},
            \frac{q_j}{q_i}
        \right\}
        +(B-1)\frac{p_j}{p_i}
    \right]
    &=
    \frac{1+\alpha_i}{p_i}
    +(B-1)\frac{1}{p_i} \\
    &=
    \frac{B+\alpha_i}{p_i}.
\end{align*}
Thus Eq.~(3) of the list matching lemma can be written as
\[
    L_{\mathrm{LML}}
    =
    \sum_{i\in\mathcal V}
    p_i\frac{B}{B+\alpha_i}.
\]

It remains to compare the two expressions term by term. For every
$\alpha\ge 0$, we claim that
\[
    1-\left(\frac{\alpha}{1+\alpha}\right)^B
    \ge
    \frac{B}{B+\alpha}.
\]
If $\alpha=0$, both sides are equal to $1$. If $\alpha>0$, the inequality is
equivalent to
\[
    \left(\frac{\alpha}{1+\alpha}\right)^B
    \le
    \frac{\alpha}{B+\alpha},
\]
or equivalently
\[
    \left(1+\frac{1}{\alpha}\right)^B
    \ge
    1+\frac{B}{\alpha}.
\]
This follows immediately from the binomial theorem:
\[
    \left(1+\frac{1}{\alpha}\right)^B
    =
    1+\frac{B}{\alpha}
    +\sum_{m=2}^B \binom{B}{m}\alpha^{-m}
    \ge
    1+\frac{B}{\alpha}.
\]
Hence, for each $i\in\mathcal V$,
\[
    p_i
    \left[
        1-\left(\frac{\alpha_i}{1+\alpha_i}\right)^B
    \right]
    \ge
    p_i\frac{B}{B+\alpha_i}.
\]
Summing over $i\in\mathcal V$ yields
\[
    L_{\mathrm{PML}}
    \ge
    L_{\mathrm{LML}}.
\]

Finally, suppose $B>1$ and $P\neq Q$. We show that the inequality is strict.
If $\alpha_i=0$ for every $i\in\mathcal V$, then for every $i,j\in\mathcal V$,
\[
    \frac{q_j}{q_i}
    \le
    \frac{p_j}{p_i}.
\]
Equivalently,
\[
    \frac{q_j}{p_j}
    \le
    \frac{q_i}{p_i}.
\]
Since this holds for every pair $(i,j)$, all ratios $q_i/p_i$ must be equal to
a common constant. Because both $P$ and $Q$ are probability distributions, this
constant must be $1$, and hence $P=Q$, a contradiction. Therefore, when
$P\neq Q$, there exists at least one $i$ such that $\alpha_i>0$.

For $B>1$ and $\alpha>0$, the binomial expansion above is strict:
\[
    \left(1+\frac{1}{\alpha}\right)^B
    >
    1+\frac{B}{\alpha}.
\]
Hence
\[
    1-\left(\frac{\alpha}{1+\alpha}\right)^B
    >
    \frac{B}{B+\alpha}.
\]
Applying this to an index $i$ with $\alpha_i>0$ gives a strict term in the
sum. Therefore
\[
    L_{\mathrm{PML}}
    >
    L_{\mathrm{LML}}.
\]
This completes the proof.
\end{proof}

\subsubsection{}

We then show the clean version,  Theorem~\ref{thm::PML_accept}, is already tighter than \citep[Eq. (4)]{rowan2025list}. 
	
	\begin{theorem}
\label{thm::PML_weak_dominates_relaxed_LML}
Let $P$ be the target distribution and $Q$ be the draft distribution on a finite
vocabulary $\mathcal V$. Write
\[
    p_i:=P(i),\qquad q_i:=Q(i),\qquad i\in\mathcal V.
\]
Assume first that $p_i,q_i>0$ for all $i\in\mathcal V$. Define
\[
    L_{\mathrm{PML,weak}}
    :=
    1-\sum_{i\in\mathcal V}
    p_i
    \left(
        \frac{p_i}{p_i+q_i}
    \right)^B
\]
and
\[
    L_{\mathrm{LML,rel}}
    :=
    \sum_{i\in\mathcal V}
    p_i
    \left(
        1+\frac{p_i}{Bq_i}
    \right)^{-1}
    =
    \sum_{i\in\mathcal V}
    p_i
    \frac{Bq_i}{p_i+Bq_i}.
\]
Then
\[
    L_{\mathrm{PML,weak}}
    \ge
    L_{\mathrm{LML,rel}}.
\]
Moreover, if $B>1$, then the inequality is strict.
\end{theorem}

\begin{proof}
The proof is based on the Poisson matching lemma~\citep{li2021unified}.
Let $Y:=\tilde Z_P(1)$ and let
\[
    X^{(b)}:=\tilde Z_Q(b),\qquad b=1,\ldots,B,
\]
be generated from the same mapped Poisson process. Define the token-level
acceptance event
\[
    \mathcal A_t
    :=
    \Big\{
        Y\in\{X^{(1)},\ldots,X^{(B)}\}
    \Big\}.
\]
Applying the relaxed Poisson matching bound with $j=1$ and $k=B$, conditioned
on $Y=\tilde Z_P(1)=i$, gives
\[
    \mathbf{P}(\mathcal A_t^c| Y=i)
    \le
    \left(
        1-\left(1+\frac{p_i}{q_i}\right)^{-1}
    \right)^B
    =
    \left(
        \frac{p_i}{p_i+q_i}
    \right)^B .
\]
Taking complements yields
\[
    \mathbf{P}(\mathcal A_t| Y=i)
    \ge
    1-
    \left(
        \frac{p_i}{p_i+q_i}
    \right)^B .
\]
Averaging over $Y\sim P$ gives
\[
\begin{aligned}
    \mathbf{P}(\mathcal A_t)
    &=
    \sum_{i\in\mathcal V}
    p_i\mathbf{P}(\mathcal A_t| Y=i)\\
    &\ge
    \sum_{i\in\mathcal V}
    p_i
    \left[
        1-
        \left(
            \frac{p_i}{p_i+q_i}
        \right)^B
    \right]\\
    &=
    1-\sum_{i\in\mathcal V}
    p_i
    \left(
        \frac{p_i}{p_i+q_i}
    \right)^B
    =
    L_{\mathrm{PML,weak}}.
\end{aligned}
\]

We now compare this lower bound with the relaxed version of the list matching
lemma. In the notation of the list matching lemma paper, the proposal
distribution is denoted by $p$ and the target distribution by $q$. Since here
we use $P$ for the target and $Q$ for the draft, their relaxed conditional
bound, Eq.~(4), becomes
\[
    \mathbf{P}(\mathcal A_t| Y=i)
    \ge
    \left(
        1+\frac{p_i}{Bq_i}
    \right)^{-1}
    =
    \frac{Bq_i}{p_i+Bq_i}.
\]
Thus the corresponding unconditional relaxed LML bound is
\[
    L_{\mathrm{LML,rel}}
    =
    \sum_{i\in\mathcal V}
    p_i
    \frac{Bq_i}{p_i+Bq_i}.
\]

It remains to compare the two conditional terms. Let
\[
    r_i:=\frac{p_i}{q_i}.
\]
Then
\[
    1-
    \left(
        \frac{p_i}{p_i+q_i}
    \right)^B
    =
    1-
    \left(
        \frac{r_i}{1+r_i}
    \right)^B,
\]
whereas
\[
    \frac{Bq_i}{p_i+Bq_i}
    =
    \frac{B}{B+r_i}.
\]
Hence it suffices to prove that, for every $r>0$,
\[
    1-
    \left(
        \frac{r}{1+r}
    \right)^B
    \ge
    \frac{B}{B+r}.
\]
Equivalently,
\[
    \left(
        \frac{r}{1+r}
    \right)^B
    \le
    \frac{r}{B+r}.
\]
Since $r>0$, this is equivalent to
\[
    \left(
        1+\frac{1}{r}
    \right)^B
    \ge
    1+\frac{B}{r}.
\]
This follows from the binomial theorem:
\[
    \left(
        1+\frac{1}{r}
    \right)^B
    =
    1+\frac{B}{r}
    +
    \sum_{m=2}^B
    \binom{B}{m}r^{-m}
    \ge
    1+\frac{B}{r}.
\]
Therefore, for each $i\in\mathcal V$,
\[
    1-
    \left(
        \frac{p_i}{p_i+q_i}
    \right)^B
    \ge
    \frac{Bq_i}{p_i+Bq_i}.
\]
Multiplying by $p_i$ and summing over $i\in\mathcal V$ gives
\[
    L_{\mathrm{PML,weak}}
    \ge
    L_{\mathrm{LML,rel}}.
\]

Finally, if $B>1$, then for every $r>0$,
\[
    \sum_{m=2}^B
    \binom{B}{m}r^{-m}
    >
    0.
\]
Thus the above inequality is strict for every $i$ with $p_i,q_i>0$. Under full
support, all such terms have positive weight $p_i$, and hence
\[
    L_{\mathrm{PML,weak}}
    >
    L_{\mathrm{LML,rel}}.
\]
This completes the proof.
\end{proof}

\subsection{Generalized Theoretic Guarantee and Proof}
\label{subapp::proof_PML_accept_generalized}

\begin{theorem}
\label{thm:accept_bound_pfr_active}
Consider the exact draft-indexed active-set version of multi-draft speculative sampling
with total draft count $B$.
Fix a current decoding step $t$, and let $S_t\subseteq\{1,\ldots,B\}$ be the active set
before sampling the next token. Write $J_t := |S_t|$. 
Let the current realized context be $c_t$, and similarly define $P_t := \mathcal{M}_t(\cdot | c_t), Q_t := \mathcal{M}_d(\cdot | c_t)$. 
For $i\in\mathcal{V}$, write $p_i^{(t)} := P_t(i)$ and $q_i^{(t)} := Q_t(i)$. 
Then the token-level acceptance probability at step $t$ satisfies
\[
\mathbf{P}(\mathcal{A}_t  | c_t,S_t)
\ge
\sum_{i\in\mathcal{V}}
p_i^{(t)}
\left[
1-\left(1-\left(1+\frac{p_i^{(t)}}{q_i^{(t)}}\right)^{-1}\right)^{J_t}
\right]
=
1-\sum_{i\in\mathcal{V}}
p_i^{(t)}
\left(\frac{p_i^{(t)}}{p_i^{(t)}+q_i^{(t)}}\right)^{J_t}.
\]
In particular, at the first step of each speculative block one has $J_t=B$, so
\[
\mathbf{P}(\mathcal{A}_t  | c_t)
\ge
1-\sum_{i\in\mathcal{V}}
p_i^{(t)}
\left(\frac{p_i^{(t)}}{p_i^{(t)}+q_i^{(t)}}\right)^B.
\]
\end{theorem}

\begin{proof}
Conditioned on the realized context $c_t$ and the active set $S_t$, the exact active-set
algorithm uses only the $J_t$ active draft streams for both target-side selection and acceptance.
Inactive drafts do not participate in the current-step minimum and therefore do not affect the
current-step coupling.

Hence, after conditioning on $(c_t,S_t)$, the current step is identical to the setting of
Theorem~\ref{thm::PML_accept} with the number of drafts equal to $J_t$.
Applying Theorem~\ref{thm::PML_accept} with $B$ replaced by $J_t$ gives
\[
\mathbf{P}(\mathcal{A}_t  | c_t,S_t)
\ge
1-\sum_{i\in\mathcal{V}}
p_i^{(t)}
\left(\frac{p_i^{(t)}}{p_i^{(t)}+q_i^{(t)}}\right)^{J_t}.
\]
At the first step of a speculative block, all drafts are active, so $J_t=B$, which yields
the final claim.
\end{proof}

\subsection{Extension of Theorem \ref{thm::PML_accept_exact} to non-identically distributed proposals}
\label{app::non_id_prop}

\begin{proposition}
\label{prop::heterogeneous_PML_bound}
Let $P$ be the target distribution on a finite vocabulary $\mathcal V$, and let
$Q_1,\ldots,Q_B$ be $B$ proposal distributions. Write
\[
    p_i:=P(i),\qquad q_{b,i}:=Q_b(i),
    \qquad i\in\mathcal V,\ b\in[B].
\]
Assume first that $p_i,q_{b,i}>0$ for all $i$ and $b$. For each $i\in\mathcal V$
and $b\in[B]$, define
\[
    \alpha_i^{(b)}
    :=
    p_i\sum_{j\in\mathcal V}
    \left(
        \frac{q_{b,j}}{q_{b,i}}
        -
        \frac{p_j}{p_i}
    \right)_+ .
\]
Then there exists a common-randomness coupling such that $Y\sim P$,
$X^{(b)}\sim Q_b$ for each $b\in[B]$, the proposals
$X^{(1)},\ldots,X^{(B)}$ are mutually independent, and
\begin{equation}
\label{eq::heterogeneous_PML_conditional}
    \Pr\!\left[
        Y\in\{X^{(1)},\ldots,X^{(B)}\}
        \,\middle|\,Y=i
    \right]
    \ge
    \sum_{b=1}^B
    \frac{1}{B+\alpha_i^{(b)}} .
\end{equation}
Consequently,
\begin{equation}
\label{eq::heterogeneous_PML_unconditional}
    \Pr\!\left[
        Y\in\{X^{(1)},\ldots,X^{(B)}\}
    \right]
    \ge
    \sum_{i\in\mathcal V}
    p_i
    \sum_{b=1}^B
    \frac{1}{B+\alpha_i^{(b)}} .
\end{equation}
\end{proposition}

\begin{proof}
The proof is based on the Poisson matching lemma~\citep{li2021unified}. 
The main difference from the identical-proposal case is that there is no single
proposal ordering when the proposal laws are $Q_1,\ldots,Q_B$. We therefore
augment the sample space by a block label. Define
\[
    \widetilde{\mathcal V}:=\mathcal V\times[B],
\]
and introduce the augmented target and proposal distributions
\[
    \widetilde P(i,b):=\frac{p_i}{B},
    \qquad
    \widetilde Q_b(i,\ell):=\mathbf 1\{\ell=b\}q_{b,i}.
\]
Let
\[
    (Y,J):=\tilde Z_{\widetilde P}(1),
    \qquad
    (X^{(b)},b):=\tilde Z_{\widetilde Q_b}(1),
    \quad b=1,\ldots,B,
\]
be generated from the same underlying Poisson process on
$\widetilde{\mathcal V}\times\mathbb R_{\ge 0}$. By the mapping theorem,
$(Y,J)\sim\widetilde P$, and hence $Y\sim P$ and $J$ is uniform on $[B]$ and
independent of $Y$. Also, for each $b$, $(X^{(b)},b)\sim\widetilde Q_b$, so
$X^{(b)}\sim Q_b$. Since the laws $\widetilde Q_b$ are supported on disjoint
blocks, the corresponding proposal samples are mutually independent.

Fix $i\in\mathcal V$. Define the acceptance event
\[
    \mathcal A
    :=
    \left\{
        Y\in\{X^{(1)},\ldots,X^{(B)}\}
    \right\}.
\]
For each $b\in[B]$, let
\[
    \mathcal A_b
    :=
    \left\{
        (Y,J)=(X^{(b)},b)
    \right\}.
\]
Then $\mathcal A\supseteq \bigcup_{b=1}^B\mathcal A_b$, and the events
$\mathcal A_1,\ldots,\mathcal A_B$ are pairwise disjoint because they occur in
different blocks. Therefore,
\begin{align}
    \mathbf{P}(\mathcal A| Y=i)
    &\ge
    \sum_{b=1}^B
    \mathbf{P}(\mathcal A_b| Y=i) \notag\\
    &=
    \sum_{b=1}^B
    \mathbf{P}(J=b| Y=i)
    \mathbf{P}(\mathcal A_b| Y=i,J=b) \notag\\
    &=
    \frac{1}{B}
    \sum_{b=1}^B
    \Pr\!\left[
        \tilde Z_{\widetilde P}(1)=\tilde Z_{\widetilde Q_b}(1)
        \,\middle|\,
        \tilde Z_{\widetilde P}(1)=(i,b)
    \right].
\label{eq::heterogeneous_reduction}
\end{align}

We now apply the exact $j=1$ tail formula in the generalized Poisson matching
lemma to the pair $(\widetilde P,\widetilde Q_b)$. Conditioned on
$\tilde Z_{\widetilde P}(1)=(i,b)$, the local PML parameter is
\[
\begin{aligned}
    \widetilde\alpha_i^{(b)}
    &:=
    \widetilde P(i,b)
    \sum_{\ell=1}^B
    \sum_{j\in\mathcal V}
    \left(
        \frac{\widetilde Q_b(j,\ell)}{\widetilde Q_b(i,b)}
        -
        \frac{\widetilde P(j,\ell)}{\widetilde P(i,b)}
    \right)_+  \\
    &=
    \frac{p_i}{B}
    \sum_{j\in\mathcal V}
    \left(
        \frac{q_{b,j}}{q_{b,i}}
        -
        \frac{p_j}{p_i}
    \right)_+  \\
    &=
    \frac{\alpha_i^{(b)}}{B}.
\end{aligned}
\]
Hence
\[
    \Pr\!\left[
        \tilde Z_{\widetilde P}(1)=\tilde Z_{\widetilde Q_b}(1)
        \,\middle|\,
        \tilde Z_{\widetilde P}(1)=(i,b)
    \right]
    \ge
    1-\frac{\widetilde\alpha_i^{(b)}}{1+\widetilde\alpha_i^{(b)}}
    =
    \frac{B}{B+\alpha_i^{(b)}}.
\]
Substituting this into~\eqref{eq::heterogeneous_reduction} gives
\[
    \mathbf{P}(\mathcal A| Y=i)
    \ge
    \frac1B
    \sum_{b=1}^B
    \frac{B}{B+\alpha_i^{(b)}}
    =
    \sum_{b=1}^B
    \frac{1}{B+\alpha_i^{(b)}}.
\]
This proves~\eqref{eq::heterogeneous_PML_conditional}. Averaging over
$Y\sim P$ gives~\eqref{eq::heterogeneous_PML_unconditional}.
\end{proof}

\section{Drafter Invariance}

\subsection{Definitions}

For a speculative sampling algorithm $\mathcal{A}$ with target model $\mathcal{M}_t$, draft model $\mathcal{M}_d$, shared randomness $\mathcal{R}$ and context $x_{:t}$, let $Y_{1:\tau}$ denote the emitted output sequence of $\mathcal{A}$ in the current speculative block.

\begin{definition}[\citep{daliri2025coupling}]
\label{def:strong_drafter_invariance_single}
The algorithm $\mathcal{A}$ is said to be \emph{strongly drafter invariant} if for every $1\le j\le \tau$,
\begin{equation*}
\mathbf{P}\big\{
Y_{1:j}=y_{1:j}\,\big|\,\mathcal{R},\,x_{:t}
\big\} =
\mathbf{P}\big\{
\widetilde{Y}_{1:j}=y_{1:j}\,\Big|\,\mathcal{R},\,x_{:t}
\big\}. 
\end{equation*}
\end{definition}

\begin{definition}[\citep{rowan2025list}]
\label{def:conditional_drafter_invariance_single}
The algorithm $\mathcal{A}$ is said to be \emph{conditional drafter invariant} if for every $1\le j\le \tau$, letting $X_{1:L} := X_{1:L}(\mathcal{M}_d)$ denote the length-$L$ draft sequence by $\mathcal{M}_d$, we have
\begin{equation*}
\mathbf{P}\big\{
Y_{1:j}=y_{1:j}\,\big|\,\mathcal{R},\,x_{:t},\,X_{1:L}=x_{1:L}
\big\}  =
\mathbf{P}\big\{
\widetilde{Y}_{1:j}=y_{1:j}\,\big|\,\mathcal{R},\,x_{:t},\,\widetilde{X}_{1:L}=x_{1:L}
\big\}. 
\end{equation*}
\end{definition}

\subsubsection{Proof of Proposition~\ref{prop::single_dft_st_tm_inv}}
\label{subsubapp::proof_single_dft_st_tm_inv}

We rephrase Proposition~\ref{prop::single_dft_st_tm_inv} to a more detailed version as follows.

\begin{proposition}[Stopping-time drafter invariance]
\label{prop:single_pfr_stopping_time_invariance}
Let $\tau$ denote the number of tokens emitted in one speculative block of
Algorithm~\ref{alg:single_PFR_Spec_Water_simplified}. 
Fix the initial prompt $h:=w_{1:n}$ and the shared randomness
\[
\mathcal{R}:=\{\,G(\mathsf{k},c): c \text{ is any context}\,\}.
\]
Define the \emph{target-side PFR chain} recursively by
\[
Z_1:=\textsc{PFR}\bigl(P(\cdot | h),\,G(\mathsf{k},h)\bigr),
\]
and for $s\ge 2$,
\[
Z_s:=\textsc{PFR}\bigl(P(\cdot | h\|Z_{1:s-1}),\,G(\mathsf{k},\,h\|Z_{1:s-1})\bigr).
\]
Then, for any stopping-time value $\tau_0$ and every $1\le j\le \tau_0$,
\[
Y_{1:j}=Z_{1:j}
\qquad\text{almost surely on the event }\{\tau=\tau_0\}.
\]
Equivalently, for every output prefix $y_{1:j}$,
\[
\mathbf{P}\bigl\{
Y_{1:j}=y_{1:j}\,\big|\,\mathcal{R},\,h,\,\tau=\tau_0
\bigr\}
=
\mathds{1}\{Z_{1:j}=y_{1:j}\}.
\]
In particular, for any two draft models $Q$ and $\widetilde Q$,
\[
\mathbf{P}\bigl\{
Y_{1:j}=y_{1:j}\,\big|\,\mathcal{R},\,h,\,\tau=\tau_0
\bigr\}
=
\mathbf{P}\bigl\{
\widetilde Y_{1:j}=y_{1:j}\,\big|\,\mathcal{R},\,h,\,\widetilde\tau=\tau_0
\bigr\},
\]
so Algorithm~\ref{alg:single_PFR_Spec_Water_simplified} is stopping-time drafter invariant in the sense of Definition~\ref{def:st_tm_drafter_invariance_single}.
\end{proposition}

\begin{proof}
Fix $(\mathcal{R},h)$ and a stopping-time value $\tau_0$.
We first show by induction on $s$ that, on the event $\{\tau\ge s\}$,
\[
Y_s=Z_s.
\]

For $s=1$, the first context is
\[
c_1=h.
\]
Hence
\[
Y_1
=
\textsc{PFR}\bigl(P(\cdot | c_1),\,G(\mathsf{k},c_1)\bigr)
=
\textsc{PFR}\bigl(P(\cdot | h),\,G(\mathsf{k},h)\bigr)
=
Z_1.
\]

Now suppose $s\ge 2$ and that $Y_r=Z_r$ holds on $\{\tau\ge r\}$ for all
$r=1,\ldots,s-1$. On the event $\{\tau\ge s\}$, the block has reached step $s$,
so the first $s-1$ draft tokens were accepted. Therefore
\[
\tilde w_{1:s-1}=Y_{1:s-1}.
\]
By the induction hypothesis, on $\{\tau\ge s\}$ we also have
\[
Y_{1:s-1}=Z_{1:s-1},
\]
and hence
\[
c_s
=
h\|\tilde w_{1:s-1}
=
h\|Y_{1:s-1}
=
h\|Z_{1:s-1}.
\]
Thus,
\[
Y_s
=
\textsc{PFR}\bigl(P(\cdot | c_s),\,G(\mathsf{k},c_s)\bigr)
=
\textsc{PFR}\bigl(P(\cdot | h\|Z_{1:s-1}),\,G(\mathsf{k},\,h\|Z_{1:s-1})\bigr)
=
Z_s.
\]
This completes the induction.

Now let $1\le j\le \tau_0$. On the event $\{\tau=\tau_0\}$ we have
$\tau\ge j$, so the above induction gives
\[
Y_{1:j}=Z_{1:j}
\qquad\text{almost surely on }\{\tau=\tau_0\}.
\]
Therefore,
\[
\mathbf{P}\bigl\{
Y_{1:j}=y_{1:j}\,\big|\,\mathcal{R},\,h,\,\tau=\tau_0
\bigr\}
=
\mathds{1}\{Z_{1:j}=y_{1:j}\}.
\]
Since the right-hand side depends only on $(\mathcal{R},h,\tau_0)$ and not on
the draft model, the same conditional law holds for any other draft model
$\widetilde Q$ conditioned on $\widetilde\tau=\tau_0$. This is exactly
stopping-time drafter invariance.
\end{proof}

\subsection{Multi-Draft Drafter invariance}
\label{subapp::multi_draft_invariance}

For multi-draft speculative sampling scheme $\mathcal{A}$ with target model $\mathcal{M}_t$, draft models $\mathcal{M}_d^{(1)},\ldots,\mathcal{M}_d^{(B)}$, context $x_{:t}$, shared randomness $\mathcal{R}$ and $b\in[B]$, let $X_{1:L}^{(b)} := X_{1:L}(\mathcal{M}_d^{(b)})$ denote the length-$L$ draft sequence by $\mathcal{M}_d^{(b)}$ and $Y_{1:\tau}$ denote the emitted output sequence of $\mathcal{A}$ in the current speculative block.

\begin{definition}[Strong drafter invariance]
\label{def:drafter_invariance}
The algorithm $\mathcal{A}$ is \emph{drafter invariant} if for every $1\le j\le \tau$,
every initial context $x_{:t}$, every realization of the shared randomness $\mathcal{R}$,
and every output prefix $y_{1:j}$,
\begin{align*}
&\mathbf{P}\Bigl\{
Y_{1:j}=y_{1:j}\,\Big|\,\mathcal{R},\,x_{:t}
\Bigr\} =
\mathbf{P}\Bigl\{
\widetilde{Y}_{1:j}=y_{1:j}\,\Big|\,\mathcal{R},\,x_{:t}
\Bigr\},
\end{align*}
for any two choices of draft models
\[
\mathcal{M}_d^{(1)},\ldots,\mathcal{M}_d^{(B)}
\quad\text{and}\quad
\widetilde{\mathcal{M}}_d^{(1)},\ldots,\widetilde{\mathcal{M}}_d^{(B)}. 
\]
\end{definition}

\begin{definition}[Conditional drafter invariance]
\label{def:conditional_drafter_invariance}
We say that the algorithm is \emph{conditionally drafter invariant} if, for every $1\le j\le \tau$,
every initial context $x_{:t}$, every realization of the shared randomness $\mathcal{R}$,
every draft sequences $x_{1:L}^{(1)},\ldots,x_{1:L}^{(B)}$, and every output prefix $y_{1:j}$,
\begin{align*}
&\mathbf{P}\Bigl\{
Y_{1:j}=y_{1:j}\,\Big|\,
\mathcal{R},\,x_{:t},\,
X_{1:L}^{(1)}=x_{1:L}^{(1)},\ldots,
X_{1:L}^{(B)}=x_{1:L}^{(B)}
\Bigr\} \\
&=
\mathbf{P}\Bigl\{
\widetilde{Y}_{1:j}=y_{1:j}\,\Big|\,
\mathcal{R},\,x_{:t},\,
\widetilde{X}_{1:L}^{(1)}=x_{1:L}^{(1)},\ldots,
\widetilde{X}_{1:L}^{(B)}=x_{1:L}^{(B)}
\Bigr\},
\end{align*}
for any two choices of draft models
\[
\mathcal{M}_d^{(1)},\ldots,\mathcal{M}_d^{(B)}
\quad\text{and}\quad
\widetilde{\mathcal{M}}_d^{(1)},\ldots,\widetilde{\mathcal{M}}_d^{(B)}.
\]
That is, conditioned on the shared randomness, the initial context, and the realized draft sequences, the conditional law of the emitted output prefix does not depend on which draft models generated those sequences.
\end{definition}

\begin{definition}[stopping time drafter invariance]
\label{def:block_length_drafter_invariance}
Consider a speculative sampling algorithm with target model $\mathcal{M}_t$,
draft models $\mathcal{M}_d^{(1)},\ldots,\mathcal{M}_d^{(B)}$, initial context $x_{:t}$,
shared randomness $\mathcal{R}$, and block output $Y_{1:\tau}$, where $\tau$ is the number
of emitted tokens in the current speculative block.

We say that the algorithm is \emph{stopping time drafter invariant} if, for every $1\le j\le \tau$,
every initial context $x_{:t}$, every realization of the shared randomness $\mathcal{R}$,
every block length value $\tau_0$, and every output prefix $y_{1:j}$,
\begin{align*}
&\mathbf{P}\Bigl\{
Y_{1:j}=y_{1:j}\,\Big|\,\mathcal{R},\,x_{:t},\,\tau=\tau_0
\Bigr\} \\
&=
\mathbf{P}\Bigl\{
\widetilde{Y}_{1:j}=y_{1:j}\,\Big|\,\mathcal{R},\,x_{:t},\,\widetilde{\tau}=\tau_0
\Bigr\},
\end{align*}
for any two choices of draft models
\[
\mathcal{M}_d^{(1)},\ldots,\mathcal{M}_d^{(B)}
\quad\text{and}\quad
\widetilde{\mathcal{M}}_d^{(1)},\ldots,\widetilde{\mathcal{M}}_d^{(B)}.
\]
That is, conditioned on the shared randomness, the initial context, and the realized block length,
the conditional law of the emitted output prefix does not depend on the draft models.
\end{definition}

\section{Experiment Details and Full Results}
\label{app::experiments}

This appendix expands the main-text evaluation along the two axes that mirror
the contributions of Section~\ref{sec::intro}:
(i) the \emph{single-draft regime} ($B{=}1$), where we certify that our keyed
Poisson sampler preserves acceptance and quality while carrying watermark
signal; and
(ii) the \emph{multi-draft regime} ($B{>}1$), where we extend the
unbiasedness/strength guarantees to larger $B$ and benchmark sampling
efficiency against the list-coupling scheme of~\citet{rowan2025list}.

\subsection{Setup}
\label{app:setup}

\paragraph{Models.}
We use two target/drafter pairs that span the two model families used in the
prior watermark and list-coupling literature.
The headline configuration follows~\citet{rowan2025list}: target
\texttt{Qwen2.5-7B-Instruct} with drafter \texttt{Qwen2.5-0.5B-Instruct}.
For instruction-following diagnostics we additionally use
\texttt{lmsys/vicuna-7b-v1.5} paired with the small same-family drafter
\texttt{double7/vicuna-68m}; the FastChat v1.1 chat template is injected
explicitly because the released checkpoint does not ship a
\texttt{chat\_template} attribute.
All models run in \texttt{float16} on a single NVIDIA H100 GPU.

\paragraph{Datasets.}
We evaluate on two tasks chosen to match the coverage of prior work:
\textsc{CNN/DailyMail} (summarization, prompted with the
\texttt{System:/INPUT:/OUTPUT:} template of~\citet{hu2024inevitable}) and
\textsc{ELI5} (open-ended explanation,~\citet{fan2019eli5}).
Each cell uses $n{=}1000$ prompts.

\paragraph{Sampling parameters.}
Unless stated otherwise we use $\mathrm{top}_k{=}50$, $\mathrm{top}_p{=}1.0$,
temperature $T{=}1.0$, and \texttt{max\_new\_tokens}${=}128$.
The watermark key is fixed across all runs so that detection is deterministic.

\paragraph{Sweeps.}
For single-draft we sweep lookahead $L\in\{1,2,3,4\}$, matching
\citet{hu2024inevitable}.
For multi-draft we fix $L{=}4$ and sweep draft counts
$B\in\{2,4,6,8\}$, matching~\citet{rowan2025list}.
Each $(\textsc{model},\textsc{dataset},\textsc{decoder},L,B)$ configuration is
run with multiple independent seeds; tables in
Appendices~\ref{app:single_draft_full}--\ref{app:multi_draft_full} report
$\mathrm{mean}\pm\mathrm{std}$ across the $1000$ prompts of the cell.

\paragraph{Baselines.}
We adopt the four-family taxonomy used in the main text
(Section~\ref{sec::experiments}). For convenience we restate the role of each
baseline:
\begin{itemize}
  \item \textbf{Quality anchor} --- \textsc{Basic-UWM}: autoregressive
    Gumbel-style unbiased watermarking with no speculation. Provides the
    LPPL reference and the watermark-strength ceiling.
  \item \textbf{Speculation-only references} --- \textsc{VSpS}
    \citep{leviathan2023fast} at $B{=}1$, and the
    \textsc{Invariant} (\emph{InvariantMultiDraftStrategy}) of
    \citet{rowan2025list} at $B{>}1$. These set the efficiency ceiling.
  \item \textbf{Trade-off frontier} --- \textsc{MWS} and \textsc{MSE}
    \citep{hu2024inevitable}, plus \textsc{MSE-Pseudo}, the dual-key
    pseudorandom-$r$ variant of~\citet{he2026improving}. These mark the
    Pareto frontier we aim to circumvent.
  \item \textbf{Ours} --- \textsc{PFR} (single-draft), \textsc{MPFR}
    (multi-draft), and the unwatermarked variant \textsc{PFR-NoWM}.
\end{itemize}

\paragraph{Metrics.}
We restate the three metric families used in the main text.
\emph{Efficiency.} \textbf{AATPS} (average accepted tokens per step) and
\textbf{Token rate} in tokens/s. The latter is hardware-dependent and is
reported only as a diagnostic.
\emph{Detectability.} \textbf{ANLPPT} --- the average negative log $p$-value
per token --- under three score variants:
\textbf{ANLPPT-U}~\citep{aaronson2023watermarking, hu2024inevitable},
\textbf{ANLPPT-Li}~\citep{li2025statistical}, and
\textbf{ANLPPT-PL}~\citep{lattimore2026refined}.
We additionally report \textbf{TPR@1\%FPR} based on the Aaronson Gamma-tail
test (Appendix~\ref{app:ablation}, ``Detection statistic'').
\emph{Quality.} \textbf{LPPL} (per-token log-perplexity under the target),
reported as an audit column. We require LPPL parity with \textsc{Basic-UWM}
to certify unbiasedness, and we do not optimise against this metric.

\paragraph{Reproducibility.}
The exact $(\mathrm{seed},\textsc{model},\textsc{dataset},\textsc{decoder},L,B)$
configurations, software versions
(\texttt{torch}, \texttt{transformers}, attention kernel), GPU model, and
per-seed result files are released with the paper.

\subsection{Prompt templates}
\label{app:prompts}

For all instruction-tuned targets we apply the model's native chat template
via the HuggingFace \texttt{apply\_chat\_template} interface. The
Vicuna-v1.5 release ships without a \texttt{chat\_template} attribute, so we
explicitly inject the FastChat v1.1 template before tokenization.

\paragraph{CNN/DailyMail (summarization).}
\begin{verbatim}
[
  {"role": "system",
   "content": "You are a helpful assistant."},
  {"role": "user",
   "content": "Summarize the following article in 3-5 sentences:\n\n"
              + article[:1500]}
]
\end{verbatim}
The article is truncated to 1500 characters before tokenization.

\paragraph{ELI5 (open-ended QA).}
\begin{verbatim}
[
  {"role": "system",
   "content": "You are a helpful assistant."},
  {"role": "user",
   "content": "Please explain like I'm five:\n\n" + question}
]
\end{verbatim}

For each prompt the model generates up to $128$ new tokens under the sampling
parameters of Appendix~\ref{app:setup}. The watermark key is fixed across all
runs.

\subsection{Single-draft experiments: full results}
\label{app:single_draft_full}

\paragraph{Goal.}
\citet{hu2024inevitable} prove that, under speculative sampling, any unbiased
reweighting scheme must give up either watermark strength (\textsc{MSE}-style)
or acceptance rate (\textsc{MWS}-style); their Theorem~3.1 forbids schemes
that strictly improve both axes simultaneously.
We claim that our Poisson-process construction sidesteps the trade-off. 

\paragraph{Protocol.}
We follow~\citet{hu2024inevitable}: the same metric battery (AATPS, ANLPPT,
LPPL) and the same lookahead range $L\in\{1,2,3,4\}$. We evaluate on
\texttt{Qwen2.5-7B-Instruct}/\texttt{0.5B} and on
\texttt{vicuna-7b-v1.5}/\texttt{vicuna-68m}, on both \textsc{CNN/DailyMail}
and \textsc{ELI5}. On top of the Hu--Huang decoder set we add (i) our
\textsc{PFR} sampler, (ii) the unwatermarked variant \textsc{PFR-NoWM},
(iii) \textsc{Basic-UWM} as the unbiasedness anchor, and (iv) the
pseudorandom-$r$ variant \textsc{MSE-Pseudo} of~\citet{he2026improving}.

\paragraph{Headline figure.}
Figure~\ref{fig::single_draft_qwen_cnn_new} (in
Section~\ref{sec:SingleVersion}) plots every method as a curve in the
$(\mathrm{AATPS},\mathrm{ANLPPT}\text{-}U)$ plane sweeping $L$ on the headline
cell. Figures~\ref{fig:single_draft_qwen_eli5}--\ref{fig:single_draft_vicuna_eli5}
report the same plot on the remaining three cells. \textsc{PFR} dominates
the \textsc{MWS}/\textsc{MSE} Pareto frontier on every cell: at every $L$ it
sits both higher (greater detection signal) and to the right (greater
acceptance) than the trade-off baselines.

\begin{figure}[htpb]
    \centering
    \includegraphics[scale=0.35]{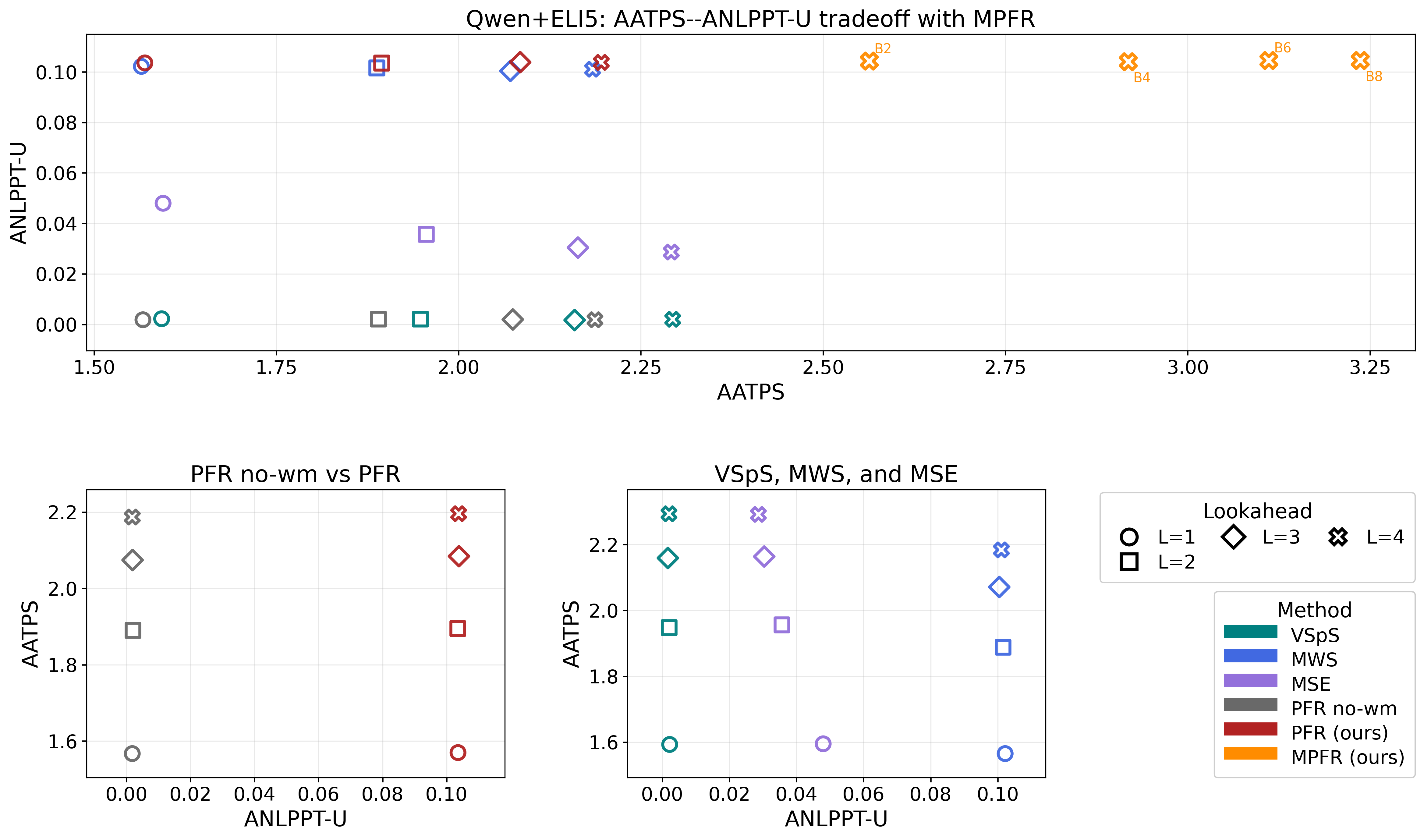}
    \caption{Single-draft (AATPS, ANLPPT-U) frontier on
    Qwen2.5-7B-Instruct / Qwen2.5-0.5B-Instruct $\times$ \textsc{ELI5}.
    \textsc{PFR} dominates \textsc{MWS}/\textsc{MSE} simultaneously on both
    axes; attaching the watermark does not decrease AATPS.}
    \label{fig:single_draft_qwen_eli5}
\end{figure}

\begin{figure}[htpb]
    \centering
    \includegraphics[scale=0.35]{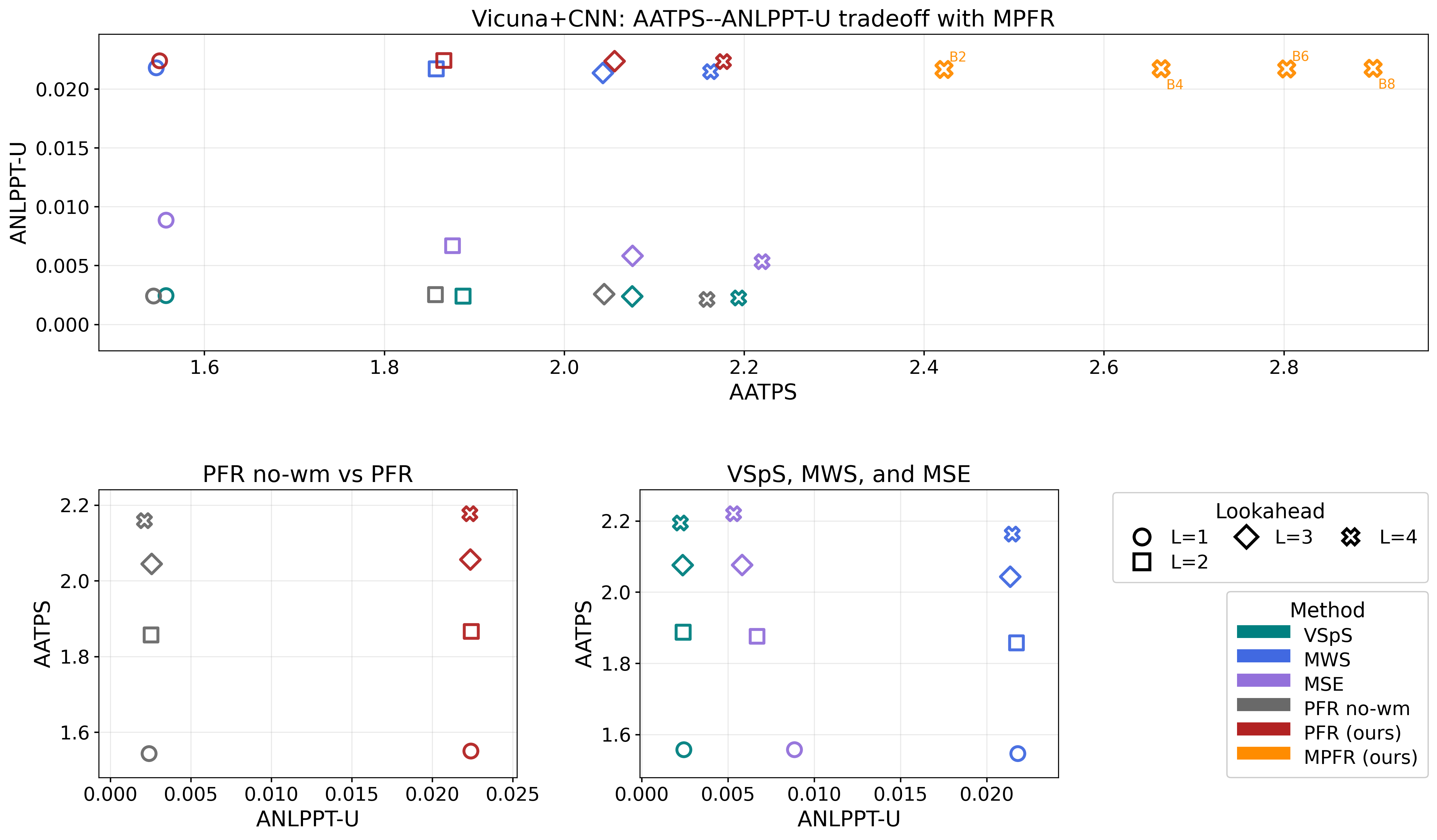}
    \caption{Single-draft (AATPS, ANLPPT-U) frontier on
    Vicuna-7B-v1.5 / Vicuna-68m $\times$ \textsc{CNN/DailyMail}.
    \textsc{PFR} dominates \textsc{MWS}/\textsc{MSE} simultaneously on both
    axes; attaching the watermark does not decrease AATPS.}
    \label{fig:single_draft_vicuna_cnn}
\end{figure}

\begin{figure}[htpb]
    \centering
    \includegraphics[scale=0.35]{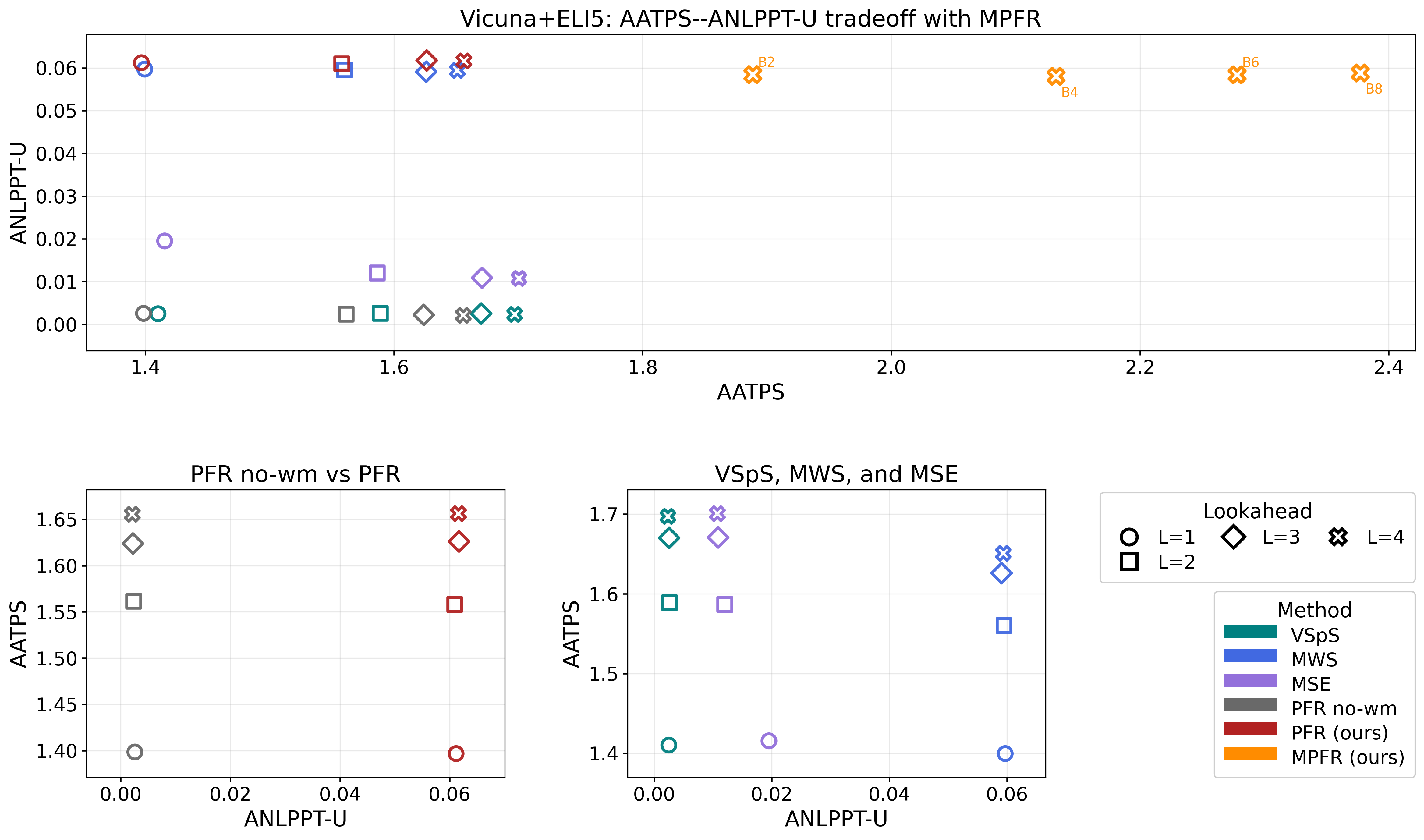}
    \caption{Single-draft (AATPS, ANLPPT-U) frontier on
    Vicuna-7B-v1.5 / Vicuna-68m $\times$ \textsc{ELI5}.
    \textsc{PFR} dominates \textsc{MWS}/\textsc{MSE} simultaneously on both
    axes; attaching the watermark does not decrease AATPS.}
    \label{fig:single_draft_vicuna_eli5}
\end{figure}

\paragraph{Unbiasedness audit.}
We report LPPL under the target model as an empirical audit of whether the
keyed Poisson coupling introduces a measurable likelihood shift. 
\textsc{PFR} matches \textsc{Basic-UWM} in LPPL within statistical noise, as
shown in the final column of
Tables~\ref{tab:single_draft_qwen_cnn}--\ref{tab:single_draft_vicuna_eli5}.
Together with the AATPS and ANLPPT results, this supports our empirical claim
that coupling-level watermarking preserves both acceptance behavior and
detector-visible signal; the exact marginal guarantee follows from the PFR
construction.

\paragraph{Per-cell tables.}
Tables~\ref{tab:single_draft_qwen_cnn}--\ref{tab:single_draft_vicuna_eli5}
report all metrics per $(L, \textsc{decoder})$ for the four
$(\textsc{target},\textsc{dataset})$ cells. Each cell reports
$\mathrm{mean}\pm\mathrm{std}$ across the $1000$ prompts of the cell.

\begin{table*}[!htbp]
\centering
\caption{Single-draft results on Qwen2.5-7B-Instruct / Qwen2.5-0.5B-Instruct
$\times$ \textsc{CNN/DailyMail}. Each cell: $\mathrm{mean}\pm\mathrm{std}$
across $1000$ prompts.}
\label{tab:single_draft_qwen_cnn}
\resizebox{\textwidth}{!}{
\begin{tabular}{llcccccc}
\toprule
$L$ & Method & AATPS & Token rate & ANLPPT-U & ANLPPT-Li & ANLPPT-PL & LPPL \\
\midrule
1 & \textsc{Basic-UWM} & $1.000 \pm 0.000$ & $33.353 \pm 2.815$ & $0.048 \pm 0.035$ & $0.035 \pm 0.029$ & $0.051 \pm 0.036$ & $0.473 \pm 0.134$ \\
 & \textsc{VSpS} & $1.641 \pm 0.070$ & $33.229 \pm 1.944$ & $0.003 \pm 0.005$ & $0.003 \pm 0.006$ & $0.003 \pm 0.006$ & $0.477 \pm 0.136$ \\
 & \textsc{PFR-NoWM} & $1.623 \pm 0.075$ & $32.147 \pm 2.161$ & $0.002 \pm 0.006$ & $0.003 \pm 0.006$ & $0.002 \pm 0.006$ & $0.474 \pm 0.140$ \\
 & \textsc{PFR} & $1.625 \pm 0.072$ & $33.321 \pm 1.706$ & $0.051 \pm 0.036$ & $0.036 \pm 0.030$ & $0.053 \pm 0.038$ & $0.470 \pm 0.143$ \\
 & \textsc{MSE} & $1.642 \pm 0.070$ & $20.578 \pm 2.645$ & $0.023 \pm 0.021$ & $0.016 \pm 0.017$ & $0.025 \pm 0.023$ & $0.472 \pm 0.136$ \\
 & \textsc{MWS} & $1.621 \pm 0.072$ & $20.136 \pm 2.952$ & $0.049 \pm 0.035$ & $0.034 \pm 0.029$ & $0.052 \pm 0.036$ & $0.475 \pm 0.133$ \\
 & \textsc{MSE-Pseudo} & $1.641 \pm 0.072$ & $20.703 \pm 2.420$ & $0.033 \pm 0.028$ & $0.023 \pm 0.023$ & $0.035 \pm 0.029$ & $0.475 \pm 0.138$ \\
\midrule
2 & \textsc{Basic-UWM} & $1.000 \pm 0.000$ & $33.353 \pm 2.815$ & $0.048 \pm 0.035$ & $0.035 \pm 0.029$ & $0.051 \pm 0.036$ & $0.473 \pm 0.134$ \\
 & \textsc{VSpS} & $2.055 \pm 0.147$ & $29.399 \pm 2.693$ & $0.002 \pm 0.006$ & $0.002 \pm 0.006$ & $0.002 \pm 0.006$ & $0.482 \pm 0.139$ \\
 & \textsc{PFR-NoWM} & $2.014 \pm 0.151$ & $30.119 \pm 2.438$ & $0.002 \pm 0.005$ & $0.003 \pm 0.005$ & $0.002 \pm 0.005$ & $0.474 \pm 0.140$ \\
 & \textsc{PFR} & $2.016 \pm 0.154$ & $29.746 \pm 2.374$ & $0.051 \pm 0.036$ & $0.036 \pm 0.030$ & $0.053 \pm 0.038$ & $0.471 \pm 0.143$ \\
 & \textsc{MSE} & $2.056 \pm 0.149$ & $16.954 \pm 2.182$ & $0.016 \pm 0.017$ & $0.012 \pm 0.014$ & $0.018 \pm 0.019$ & $0.475 \pm 0.139$ \\
 & \textsc{MWS} & $2.012 \pm 0.155$ & $17.915 \pm 2.679$ & $0.048 \pm 0.035$ & $0.034 \pm 0.028$ & $0.052 \pm 0.036$ & $0.473 \pm 0.134$ \\
 & \textsc{MSE-Pseudo} & $2.053 \pm 0.150$ & $17.768 \pm 2.415$ & $0.030 \pm 0.026$ & $0.022 \pm 0.021$ & $0.033 \pm 0.028$ & $0.481 \pm 0.143$ \\
\midrule
3 & \textsc{Basic-UWM} & $1.000 \pm 0.000$ & $33.353 \pm 2.815$ & $0.048 \pm 0.035$ & $0.035 \pm 0.029$ & $0.051 \pm 0.036$ & $0.473 \pm 0.134$ \\
 & \textsc{VSpS} & $2.327 \pm 0.228$ & $26.272 \pm 2.968$ & $0.002 \pm 0.006$ & $0.002 \pm 0.006$ & $0.003 \pm 0.006$ & $0.481 \pm 0.141$ \\
 & \textsc{PFR-NoWM} & $2.270 \pm 0.228$ & $26.709 \pm 2.682$ & $0.002 \pm 0.006$ & $0.002 \pm 0.006$ & $0.002 \pm 0.005$ & $0.474 \pm 0.140$ \\
 & \textsc{PFR} & $2.272 \pm 0.225$ & $25.999 \pm 2.742$ & $0.051 \pm 0.036$ & $0.037 \pm 0.030$ & $0.053 \pm 0.038$ & $0.472 \pm 0.144$ \\
 & \textsc{MSE} & $2.330 \pm 0.228$ & $13.453 \pm 1.603$ & $0.015 \pm 0.017$ & $0.011 \pm 0.014$ & $0.017 \pm 0.018$ & $0.476 \pm 0.140$ \\
 & \textsc{MWS} & $2.262 \pm 0.222$ & $16.187 \pm 2.509$ & $0.048 \pm 0.035$ & $0.034 \pm 0.028$ & $0.051 \pm 0.036$ & $0.474 \pm 0.134$ \\
 & \textsc{MSE-Pseudo} & $2.337 \pm 0.232$ & $14.933 \pm 2.118$ & $0.029 \pm 0.026$ & $0.021 \pm 0.022$ & $0.032 \pm 0.027$ & $0.480 \pm 0.138$ \\
\midrule
4 & \textsc{Basic-UWM} & $1.000 \pm 0.000$ & $33.353 \pm 2.815$ & $0.048 \pm 0.035$ & $0.035 \pm 0.029$ & $0.051 \pm 0.036$ & $0.473 \pm 0.134$ \\
 & \textsc{VSpS} & $2.513 \pm 0.288$ & $23.229 \pm 3.176$ & $0.002 \pm 0.006$ & $0.002 \pm 0.005$ & $0.002 \pm 0.006$ & $0.475 \pm 0.134$ \\
 & \textsc{PFR-NoWM} & $2.439 \pm 0.292$ & $23.285 \pm 2.840$ & $0.002 \pm 0.005$ & $0.002 \pm 0.005$ & $0.002 \pm 0.006$ & $0.474 \pm 0.142$ \\
 & \textsc{PFR} & $2.446 \pm 0.286$ & $23.075 \pm 2.755$ & $0.050 \pm 0.036$ & $0.036 \pm 0.030$ & $0.053 \pm 0.038$ & $0.471 \pm 0.144$ \\
 & \textsc{MSE} & $2.525 \pm 0.298$ & $12.623 \pm 1.893$ & $0.014 \pm 0.016$ & $0.010 \pm 0.013$ & $0.016 \pm 0.017$ & $0.479 \pm 0.145$ \\
 & \textsc{MWS} & $2.432 \pm 0.287$ & $13.592 \pm 2.213$ & $0.048 \pm 0.035$ & $0.034 \pm 0.028$ & $0.052 \pm 0.036$ & $0.475 \pm 0.132$ \\
 & \textsc{MSE-Pseudo} & $2.499 \pm 0.287$ & $13.494 \pm 1.934$ & $0.029 \pm 0.026$ & $0.021 \pm 0.021$ & $0.031 \pm 0.027$ & $0.481 \pm 0.141$ \\
\bottomrule
\end{tabular}
}
\end{table*}

\begin{table*}[!htbp]
\centering
\caption{Single-draft results on Qwen2.5-7B-Instruct / Qwen2.5-0.5B-Instruct
$\times$ \textsc{ELI5}. Each cell: $\mathrm{mean}\pm\mathrm{std}$ across
$1000$ prompts.}
\label{tab:single_draft_qwen_eli5}
\resizebox{\textwidth}{!}{
\begin{tabular}{llcccccc}
\toprule
$L$ & Method & AATPS & Token rate & ANLPPT-U & ANLPPT-Li & ANLPPT-PL & LPPL \\
\midrule
1 & \textsc{Basic-UWM} & $1.000 \pm 0.000$ & $26.132 \pm 1.924$ & $0.100 \pm 0.045$ & $0.072 \pm 0.039$ & $0.104 \pm 0.046$ & $0.696 \pm 0.141$ \\
 & \textsc{VSpS} & $1.593 \pm 0.060$ & $32.090 \pm 1.889$ & $0.002 \pm 0.005$ & $0.002 \pm 0.005$ & $0.002 \pm 0.005$ & $0.701 \pm 0.133$ \\
 & \textsc{PFR-NoWM} & $1.567 \pm 0.061$ & $33.046 \pm 1.405$ & $0.002 \pm 0.004$ & $0.002 \pm 0.004$ & $0.002 \pm 0.004$ & $0.703 \pm 0.139$ \\
 & \textsc{PFR} & $1.570 \pm 0.062$ & $33.213 \pm 1.487$ & $0.104 \pm 0.048$ & $0.075 \pm 0.039$ & $0.107 \pm 0.049$ & $0.699 \pm 0.134$ \\
 & \textsc{MSE} & $1.595 \pm 0.060$ & $18.218 \pm 2.022$ & $0.048 \pm 0.030$ & $0.032 \pm 0.024$ & $0.053 \pm 0.032$ & $0.704 \pm 0.135$ \\
 & \textsc{MWS} & $1.565 \pm 0.063$ & $20.177 \pm 2.688$ & $0.102 \pm 0.047$ & $0.075 \pm 0.040$ & $0.106 \pm 0.048$ & $0.700 \pm 0.137$ \\
 & \textsc{MSE-Pseudo} & $1.593 \pm 0.062$ & $19.553 \pm 2.107$ & $0.068 \pm 0.038$ & $0.049 \pm 0.032$ & $0.071 \pm 0.039$ & $0.703 \pm 0.145$ \\
\midrule
2 & \textsc{Basic-UWM} & $1.000 \pm 0.000$ & $26.132 \pm 1.924$ & $0.100 \pm 0.045$ & $0.072 \pm 0.039$ & $0.104 \pm 0.046$ & $0.696 \pm 0.141$ \\
 & \textsc{VSpS} & $1.948 \pm 0.116$ & $27.350 \pm 2.434$ & $0.002 \pm 0.005$ & $0.002 \pm 0.005$ & $0.002 \pm 0.004$ & $0.701 \pm 0.136$ \\
 & \textsc{PFR-NoWM} & $1.890 \pm 0.118$ & $28.719 \pm 1.782$ & $0.002 \pm 0.005$ & $0.002 \pm 0.004$ & $0.002 \pm 0.004$ & $0.698 \pm 0.137$ \\
 & \textsc{PFR} & $1.895 \pm 0.125$ & $28.851 \pm 1.937$ & $0.104 \pm 0.048$ & $0.075 \pm 0.039$ & $0.107 \pm 0.049$ & $0.699 \pm 0.134$ \\
 & \textsc{MSE} & $1.956 \pm 0.126$ & $14.867 \pm 1.179$ & $0.036 \pm 0.024$ & $0.024 \pm 0.020$ & $0.040 \pm 0.026$ & $0.704 \pm 0.140$ \\
 & \textsc{MWS} & $1.888 \pm 0.121$ & $16.985 \pm 2.208$ & $0.102 \pm 0.046$ & $0.074 \pm 0.039$ & $0.105 \pm 0.047$ & $0.700 \pm 0.142$ \\
 & \textsc{MSE-Pseudo} & $1.951 \pm 0.118$ & $24.767 \pm 2.016$ & $0.062 \pm 0.034$ & $0.043 \pm 0.028$ & $0.066 \pm 0.036$ & $0.706 \pm 0.136$ \\
\midrule
3 & \textsc{Basic-UWM} & $1.000 \pm 0.000$ & $26.132 \pm 1.924$ & $0.100 \pm 0.045$ & $0.072 \pm 0.039$ & $0.104 \pm 0.046$ & $0.696 \pm 0.141$ \\
 & \textsc{VSpS} & $2.159 \pm 0.175$ & $23.619 \pm 2.707$ & $0.002 \pm 0.004$ & $0.002 \pm 0.004$ & $0.002 \pm 0.003$ & $0.706 \pm 0.139$ \\
 & \textsc{PFR-NoWM} & $2.074 \pm 0.167$ & $24.315 \pm 2.065$ & $0.002 \pm 0.004$ & $0.002 \pm 0.005$ & $0.002 \pm 0.004$ & $0.710 \pm 0.139$ \\
 & \textsc{PFR} & $2.084 \pm 0.174$ & $24.675 \pm 2.059$ & $0.104 \pm 0.049$ & $0.075 \pm 0.041$ & $0.107 \pm 0.049$ & $0.697 \pm 0.135$ \\
 & \textsc{MSE} & $2.164 \pm 0.175$ & $13.359 \pm 2.005$ & $0.030 \pm 0.023$ & $0.020 \pm 0.019$ & $0.035 \pm 0.025$ & $0.705 \pm 0.135$ \\
 & \textsc{MWS} & $2.071 \pm 0.169$ & $14.243 \pm 1.924$ & $0.100 \pm 0.046$ & $0.073 \pm 0.040$ & $0.104 \pm 0.047$ & $0.698 \pm 0.138$ \\
 & \textsc{MSE-Pseudo} & $2.169 \pm 0.172$ & $21.339 \pm 2.086$ & $0.061 \pm 0.034$ & $0.043 \pm 0.028$ & $0.064 \pm 0.035$ & $0.709 \pm 0.133$ \\
\midrule
4 & \textsc{Basic-UWM} & $1.000 \pm 0.000$ & $26.132 \pm 1.924$ & $0.100 \pm 0.045$ & $0.072 \pm 0.039$ & $0.104 \pm 0.046$ & $0.696 \pm 0.141$ \\
 & \textsc{VSpS} & $2.294 \pm 0.220$ & $17.142 \pm 3.485$ & $0.002 \pm 0.005$ & $0.002 \pm 0.005$ & $0.002 \pm 0.004$ & $0.705 \pm 0.137$ \\
 & \textsc{PFR-NoWM} & $2.187 \pm 0.204$ & $21.512 \pm 2.145$ & $0.002 \pm 0.004$ & $0.002 \pm 0.005$ & $0.002 \pm 0.004$ & $0.703 \pm 0.138$ \\
 & \textsc{PFR} & $2.196 \pm 0.214$ & $21.319 \pm 2.212$ & $0.104 \pm 0.048$ & $0.075 \pm 0.041$ & $0.107 \pm 0.049$ & $0.698 \pm 0.134$ \\
 & \textsc{MSE} & $2.292 \pm 0.220$ & $13.200 \pm 2.021$ & $0.029 \pm 0.021$ & $0.019 \pm 0.017$ & $0.033 \pm 0.024$ & $0.701 \pm 0.136$ \\
 & \textsc{MWS} & $2.184 \pm 0.208$ & $11.570 \pm 1.496$ & $0.101 \pm 0.046$ & $0.073 \pm 0.039$ & $0.105 \pm 0.047$ & $0.696 \pm 0.141$ \\
 & \textsc{MSE-Pseudo} & $2.302 \pm 0.214$ & $18.298 \pm 1.942$ & $0.060 \pm 0.033$ & $0.042 \pm 0.028$ & $0.063 \pm 0.035$ & $0.705 \pm 0.132$ \\
\bottomrule
\end{tabular}
}
\end{table*}

\begin{table*}[!htbp]
\centering
\caption{Single-draft results on Vicuna-7B-v1.5 / Vicuna-68m $\times$
\textsc{CNN/DailyMail}. Each cell: $\mathrm{mean}\pm\mathrm{std}$ across
$1000$ prompts.}
\label{tab:single_draft_vicuna_cnn}
\resizebox{\textwidth}{!}{
\begin{tabular}{llcccccc}
\toprule
$L$ & Method & AATPS & Token rate & ANLPPT-U & ANLPPT-Li & ANLPPT-PL & LPPL \\
\midrule
1 & \textsc{Basic-UWM} & $1.000 \pm 0.000$ & $37.058 \pm 0.848$ & $0.021 \pm 0.024$ & $0.016 \pm 0.019$ & $0.022 \pm 0.025$ & $0.274 \pm 0.190$ \\
 & \textsc{VSpS} & $1.557 \pm 0.086$ & $48.890 \pm 3.439$ & $0.002 \pm 0.006$ & $0.003 \pm 0.006$ & $0.002 \pm 0.006$ & $0.275 \pm 0.167$ \\
 & \textsc{PFR-NoWM} & $1.544 \pm 0.088$ & $37.027 \pm 5.064$ & $0.002 \pm 0.006$ & $0.003 \pm 0.006$ & $0.002 \pm 0.005$ & $0.269 \pm 0.121$ \\
 & \textsc{PFR} & $1.550 \pm 0.090$ & $38.256 \pm 6.092$ & $0.022 \pm 0.024$ & $0.017 \pm 0.019$ & $0.023 \pm 0.025$ & $0.272 \pm 0.128$ \\
 & \textsc{MSE} & $1.557 \pm 0.088$ & $45.513 \pm 4.490$ & $0.009 \pm 0.013$ & $0.007 \pm 0.011$ & $0.010 \pm 0.015$ & $0.269 \pm 0.123$ \\
 & \textsc{MWS} & $1.547 \pm 0.091$ & $30.814 \pm 4.078$ & $0.022 \pm 0.024$ & $0.017 \pm 0.020$ & $0.023 \pm 0.026$ & $0.276 \pm 0.192$ \\
 & \textsc{MSE-Pseudo} & $1.555 \pm 0.088$ & $27.597 \pm 2.405$ & $0.015 \pm 0.018$ & $0.012 \pm 0.016$ & $0.016 \pm 0.018$ & $0.279 \pm 0.118$ \\
\midrule
2 & \textsc{Basic-UWM} & $1.000 \pm 0.000$ & $37.058 \pm 0.848$ & $0.021 \pm 0.024$ & $0.016 \pm 0.019$ & $0.022 \pm 0.025$ & $0.274 \pm 0.190$ \\
 & \textsc{VSpS} & $1.888 \pm 0.179$ & $53.865 \pm 5.770$ & $0.002 \pm 0.006$ & $0.002 \pm 0.006$ & $0.002 \pm 0.006$ & $0.268 \pm 0.129$ \\
 & \textsc{PFR-NoWM} & $1.857 \pm 0.178$ & $40.154 \pm 5.953$ & $0.003 \pm 0.005$ & $0.003 \pm 0.006$ & $0.003 \pm 0.005$ & $0.273 \pm 0.123$ \\
 & \textsc{PFR} & $1.866 \pm 0.183$ & $42.639 \pm 7.272$ & $0.022 \pm 0.024$ & $0.017 \pm 0.019$ & $0.023 \pm 0.025$ & $0.272 \pm 0.127$ \\
 & \textsc{MSE} & $1.876 \pm 0.186$ & $46.276 \pm 7.794$ & $0.007 \pm 0.010$ & $0.006 \pm 0.009$ & $0.007 \pm 0.011$ & $0.268 \pm 0.131$ \\
 & \textsc{MWS} & $1.858 \pm 0.181$ & $31.588 \pm 5.453$ & $0.022 \pm 0.024$ & $0.016 \pm 0.020$ & $0.023 \pm 0.025$ & $0.278 \pm 0.193$ \\
 & \textsc{MSE-Pseudo} & $1.877 \pm 0.182$ & $29.023 \pm 4.118$ & $0.015 \pm 0.018$ & $0.012 \pm 0.015$ & $0.015 \pm 0.020$ & $0.273 \pm 0.127$ \\
\midrule
3 & \textsc{Basic-UWM} & $1.000 \pm 0.000$ & $37.058 \pm 0.848$ & $0.021 \pm 0.024$ & $0.016 \pm 0.019$ & $0.022 \pm 0.025$ & $0.274 \pm 0.190$ \\
 & \textsc{VSpS} & $2.076 \pm 0.254$ & $53.835 \pm 7.463$ & $0.002 \pm 0.006$ & $0.003 \pm 0.006$ & $0.002 \pm 0.005$ & $0.276 \pm 0.212$ \\
 & \textsc{PFR-NoWM} & $2.044 \pm 0.259$ & $42.513 \pm 7.326$ & $0.003 \pm 0.006$ & $0.003 \pm 0.006$ & $0.002 \pm 0.006$ & $0.274 \pm 0.128$ \\
 & \textsc{PFR} & $2.056 \pm 0.263$ & $41.813 \pm 7.586$ & $0.022 \pm 0.024$ & $0.017 \pm 0.019$ & $0.023 \pm 0.025$ & $0.272 \pm 0.127$ \\
 & \textsc{MSE} & $2.076 \pm 0.260$ & $38.713 \pm 8.338$ & $0.006 \pm 0.009$ & $0.005 \pm 0.008$ & $0.006 \pm 0.010$ & $0.266 \pm 0.118$ \\
 & \textsc{MWS} & $2.043 \pm 0.258$ & $28.960 \pm 4.759$ & $0.021 \pm 0.024$ & $0.016 \pm 0.020$ & $0.022 \pm 0.026$ & $0.277 \pm 0.191$ \\
 & \textsc{MSE-Pseudo} & $2.062 \pm 0.257$ & $32.256 \pm 6.945$ & $0.015 \pm 0.019$ & $0.012 \pm 0.016$ & $0.015 \pm 0.020$ & $0.277 \pm 0.125$ \\
\midrule
4 & \textsc{Basic-UWM} & $1.000 \pm 0.000$ & $37.058 \pm 0.848$ & $0.021 \pm 0.024$ & $0.016 \pm 0.019$ & $0.022 \pm 0.025$ & $0.274 \pm 0.190$ \\
 & \textsc{VSpS} & $2.194 \pm 0.324$ & $53.940 \pm 8.237$ & $0.002 \pm 0.005$ & $0.002 \pm 0.005$ & $0.002 \pm 0.005$ & $0.271 \pm 0.120$ \\
 & \textsc{PFR-NoWM} & $2.159 \pm 0.326$ & $40.083 \pm 7.877$ & $0.002 \pm 0.005$ & $0.002 \pm 0.005$ & $0.002 \pm 0.005$ & $0.268 \pm 0.117$ \\
 & \textsc{PFR} & $2.177 \pm 0.323$ & $41.800 \pm 9.355$ & $0.022 \pm 0.023$ & $0.017 \pm 0.019$ & $0.023 \pm 0.025$ & $0.272 \pm 0.127$ \\
 & \textsc{MSE} & $2.220 \pm 0.330$ & $32.730 \pm 6.784$ & $0.005 \pm 0.008$ & $0.005 \pm 0.008$ & $0.006 \pm 0.009$ & $0.264 \pm 0.119$ \\
 & \textsc{MWS} & $2.163 \pm 0.319$ & $27.253 \pm 4.410$ & $0.021 \pm 0.024$ & $0.016 \pm 0.019$ & $0.023 \pm 0.025$ & $0.277 \pm 0.191$ \\
 & \textsc{MSE-Pseudo} & $2.183 \pm 0.317$ & $33.007 \pm 6.806$ & $0.014 \pm 0.018$ & $0.011 \pm 0.015$ & $0.014 \pm 0.018$ & $0.277 \pm 0.121$ \\
\bottomrule
\end{tabular}
}
\end{table*}

\begin{table*}[!htbp]
\centering
\caption{Single-draft results on Vicuna-7B-v1.5 / Vicuna-68m $\times$
\textsc{ELI5}. Each cell: $\mathrm{mean}\pm\mathrm{std}$ across $1000$
prompts.}
\label{tab:single_draft_vicuna_eli5}
\resizebox{\textwidth}{!}{
\begin{tabular}{llcccccc}
\toprule
$L$ & Method & AATPS & Token rate & ANLPPT-U & ANLPPT-Li & ANLPPT-PL & LPPL \\
\midrule
1 & \textsc{Basic-UWM} & $1.000 \pm 0.000$ & $38.536 \pm 1.896$ & $0.060 \pm 0.040$ & $0.044 \pm 0.033$ & $0.062 \pm 0.042$ & $0.507 \pm 0.152$ \\
 & \textsc{VSpS} & $1.410 \pm 0.063$ & $46.306 \pm 3.178$ & $0.002 \pm 0.005$ & $0.003 \pm 0.005$ & $0.002 \pm 0.005$ & $0.511 \pm 0.150$ \\
 & \textsc{PFR-NoWM} & $1.399 \pm 0.063$ & $33.636 \pm 4.757$ & $0.003 \pm 0.006$ & $0.003 \pm 0.007$ & $0.002 \pm 0.005$ & $0.509 \pm 0.152$ \\
 & \textsc{PFR} & $1.397 \pm 0.065$ & $35.231 \pm 5.397$ & $0.061 \pm 0.043$ & $0.045 \pm 0.038$ & $0.063 \pm 0.042$ & $0.513 \pm 0.162$ \\
 & \textsc{MSE} & $1.415 \pm 0.062$ & $44.719 \pm 3.192$ & $0.020 \pm 0.020$ & $0.014 \pm 0.017$ & $0.022 \pm 0.022$ & $0.509 \pm 0.147$ \\
 & \textsc{MWS} & $1.399 \pm 0.064$ & $26.309 \pm 3.816$ & $0.060 \pm 0.041$ & $0.044 \pm 0.034$ & $0.062 \pm 0.042$ & $0.507 \pm 0.151$ \\
 & \textsc{MSE-Pseudo} & $1.416 \pm 0.063$ & $28.987 \pm 4.863$ & $0.040 \pm 0.030$ & $0.029 \pm 0.025$ & $0.042 \pm 0.031$ & $0.504 \pm 0.152$ \\
\midrule
2 & \textsc{Basic-UWM} & $1.000 \pm 0.000$ & $38.536 \pm 1.896$ & $0.060 \pm 0.040$ & $0.044 \pm 0.033$ & $0.062 \pm 0.042$ & $0.507 \pm 0.152$ \\
 & \textsc{VSpS} & $1.589 \pm 0.104$ & $47.569 \pm 3.837$ & $0.003 \pm 0.006$ & $0.003 \pm 0.006$ & $0.003 \pm 0.006$ & $0.509 \pm 0.146$ \\
 & \textsc{PFR-NoWM} & $1.562 \pm 0.103$ & $34.678 \pm 4.452$ & $0.002 \pm 0.005$ & $0.002 \pm 0.005$ & $0.002 \pm 0.006$ & $0.505 \pm 0.147$ \\
 & \textsc{PFR} & $1.558 \pm 0.106$ & $34.425 \pm 4.996$ & $0.061 \pm 0.043$ & $0.044 \pm 0.038$ & $0.063 \pm 0.042$ & $0.513 \pm 0.162$ \\
 & \textsc{MSE} & $1.587 \pm 0.105$ & $33.648 \pm 4.994$ & $0.012 \pm 0.014$ & $0.009 \pm 0.012$ & $0.014 \pm 0.016$ & $0.509 \pm 0.147$ \\
 & \textsc{MWS} & $1.560 \pm 0.109$ & $23.481 \pm 2.014$ & $0.060 \pm 0.040$ & $0.043 \pm 0.034$ & $0.062 \pm 0.042$ & $0.510 \pm 0.157$ \\
 & \textsc{MSE-Pseudo} & $1.593 \pm 0.110$ & $28.595 \pm 4.733$ & $0.040 \pm 0.032$ & $0.029 \pm 0.027$ & $0.041 \pm 0.032$ & $0.509 \pm 0.152$ \\
\midrule
3 & \textsc{Basic-UWM} & $1.000 \pm 0.000$ & $38.536 \pm 1.896$ & $0.060 \pm 0.040$ & $0.044 \pm 0.033$ & $0.062 \pm 0.042$ & $0.507 \pm 0.152$ \\
 & \textsc{VSpS} & $1.670 \pm 0.134$ & $45.003 \pm 5.451$ & $0.003 \pm 0.006$ & $0.002 \pm 0.005$ & $0.002 \pm 0.006$ & $0.513 \pm 0.157$ \\
 & \textsc{PFR-NoWM} & $1.624 \pm 0.132$ & $31.214 \pm 4.662$ & $0.002 \pm 0.005$ & $0.002 \pm 0.005$ & $0.002 \pm 0.005$ & $0.506 \pm 0.149$ \\
 & \textsc{PFR} & $1.626 \pm 0.130$ & $35.694 \pm 5.835$ & $0.062 \pm 0.043$ & $0.045 \pm 0.038$ & $0.063 \pm 0.042$ & $0.515 \pm 0.162$ \\
 & \textsc{MSE} & $1.671 \pm 0.136$ & $26.134 \pm 3.956$ & $0.011 \pm 0.014$ & $0.008 \pm 0.012$ & $0.012 \pm 0.015$ & $0.499 \pm 0.145$ \\
 & \textsc{MWS} & $1.626 \pm 0.136$ & $22.536 \pm 2.127$ & $0.059 \pm 0.040$ & $0.043 \pm 0.033$ & $0.061 \pm 0.041$ & $0.508 \pm 0.152$ \\
 & \textsc{MSE-Pseudo} & $1.666 \pm 0.134$ & $28.611 \pm 4.748$ & $0.037 \pm 0.028$ & $0.027 \pm 0.023$ & $0.039 \pm 0.029$ & $0.505 \pm 0.148$ \\
\midrule
4 & \textsc{Basic-UWM} & $1.000 \pm 0.000$ & $38.536 \pm 1.896$ & $0.060 \pm 0.040$ & $0.044 \pm 0.033$ & $0.062 \pm 0.042$ & $0.507 \pm 0.152$ \\
 & \textsc{VSpS} & $1.697 \pm 0.152$ & $43.287 \pm 5.333$ & $0.002 \pm 0.005$ & $0.002 \pm 0.006$ & $0.002 \pm 0.005$ & $0.510 \pm 0.152$ \\
 & \textsc{PFR-NoWM} & $1.656 \pm 0.145$ & $28.914 \pm 4.147$ & $0.002 \pm 0.005$ & $0.002 \pm 0.005$ & $0.002 \pm 0.005$ & $0.503 \pm 0.149$ \\
 & \textsc{PFR} & $1.656 \pm 0.146$ & $32.316 \pm 5.921$ & $0.062 \pm 0.043$ & $0.045 \pm 0.038$ & $0.063 \pm 0.042$ & $0.515 \pm 0.162$ \\
 & \textsc{MSE} & $1.700 \pm 0.154$ & $23.993 \pm 3.469$ & $0.011 \pm 0.014$ & $0.008 \pm 0.012$ & $0.012 \pm 0.015$ & $0.514 \pm 0.150$ \\
 & \textsc{MWS} & $1.651 \pm 0.149$ & $22.198 \pm 3.679$ & $0.059 \pm 0.039$ & $0.043 \pm 0.033$ & $0.062 \pm 0.041$ & $0.509 \pm 0.153$ \\
 & \textsc{MSE-Pseudo} & $1.703 \pm 0.155$ & $26.004 \pm 4.603$ & $0.037 \pm 0.030$ & $0.027 \pm 0.024$ & $0.039 \pm 0.031$ & $0.508 \pm 0.147$ \\
\bottomrule
\end{tabular}
}
\end{table*}

\subsection{Multi-draft experiments: full per-cell results}
\label{app:multi_draft_full}

\paragraph{Setup.}
The main-text multi-draft paragraph (Section~\ref{sec::experiments},
``Multi-draft scaling'') summarises a single representative cell
(Qwen2.5-7B-Instruct $\times$ \textsc{CNN/DailyMail}). Here we extend the
evaluation to the $2{\times}2$ matrix
\{Qwen2.5-7B-Instruct, Vicuna-7B-v1.5\} $\times$
\{\textsc{CNN/DailyMail}, \textsc{ELI5}\}, using the target/drafter pairs
\texttt{Qwen2.5-7B-Instruct}/\texttt{Qwen2.5-0.5B-Instruct} and
\texttt{lmsys/vicuna-7b-v1.5}/\texttt{double7/vicuna-68m}.
Each cell uses $n{=}1000$ prompts, lookahead $L{=}4$, and draft counts
$B\in\{2,4,6,8\}$, with the sampling parameters of Appendix~\ref{app:setup}.

\paragraph{Decoders.}
\textsc{MPFR} is our keyed multi-draft Poisson sampler;
\textsc{Invariant} is the unwatermarked
\emph{InvariantMultiDraftStrategy} of~\citet{rowan2025list} run in the same
harness on the same prompt set. The two decoders share the draft/verify
mechanism and differ only in whether the per-context watermark key is applied.

\paragraph{Per-cell tables.}
Tables~\ref{tab:full-qwen-cnn}--\ref{tab:full-vicuna-eli5} report all quality
and watermark-strength metrics per (decoder, $B$). Columns: AATPS, token rate
in tokens/s (TR), ANLPPT-\{U, Li, PL\}, and LPPL. The corresponding
TPR-vs-$T_{\mathrm{eval}}$ curves are reported in
Figure~\ref{fig:full_detection_4cells} of Appendix~\ref{app:ablation}.

\begin{table}[h]
\centering
\small
\color{black}
\caption{Multi-draft results on Qwen2.5-7B-Instruct $\times$
\textsc{CNN/DailyMail}, $L{=}4$.}
\label{tab:full-qwen-cnn}
\resizebox{\textwidth}{!}{
\begin{tabular}{l c c c c c c c}
\toprule
Decoder & $B$ & AATPS & TR & ANLPPT-U & ANLPPT-Li & ANLPPT-PL & LPPL \\
\midrule
\textsc{MPFR}      & 2 & $2.797 \pm 0.001$ & $25.11 \pm 0.22$ & $0.051 \pm 0.001$ & $0.036 \pm 0.001$ & $0.053 \pm 0.001$ & $0.474 \pm 0.002$ \\
\textsc{MPFR}      & 4 & $3.123 \pm 0.002$ & $27.50 \pm 0.31$ & $0.051 \pm 0.001$ & $0.036 \pm 0.001$ & $0.053 \pm 0.001$ & $0.474 \pm 0.002$ \\
\textsc{MPFR}      & 6 & $3.294 \pm 0.003$ & $28.50 \pm 0.49$ & $0.050 \pm 0.001$ & $0.036 \pm 0.001$ & $0.053 \pm 0.001$ & $0.474 \pm 0.001$ \\
\textsc{MPFR}      & 8 & $3.409 \pm 0.002$ & $29.13 \pm 0.59$ & $0.051 \pm 0.002$ & $0.036 \pm 0.001$ & $0.053 \pm 0.002$ & $0.474 \pm 0.002$ \\
\textsc{Invariant} & 2 & $2.803 \pm 0.007$ & $24.78 \pm 0.14$ & $0.002 \pm 0.000$ & $0.003 \pm 0.000$ & $0.002 \pm 0.000$ & $0.483 \pm 0.005$ \\
\textsc{Invariant} & 4 & $3.109 \pm 0.008$ & $27.12 \pm 0.20$ & $0.002 \pm 0.000$ & $0.002 \pm 0.000$ & $0.002 \pm 0.000$ & $0.484 \pm 0.003$ \\
\textsc{Invariant} & 6 & $3.272 \pm 0.010$ & $27.84 \pm 0.60$ & $0.002 \pm 0.000$ & $0.002 \pm 0.000$ & $0.002 \pm 0.000$ & $0.481 \pm 0.000$ \\
\textsc{Invariant} & 8 & $3.384 \pm 0.012$ & $28.62 \pm 0.54$ & $0.002 \pm 0.000$ & $0.003 \pm 0.000$ & $0.002 \pm 0.000$ & $0.476 \pm 0.004$ \\
\bottomrule
\end{tabular}
}
\end{table}

\begin{table}[h]
\centering
\small
\color{black}
\caption{Multi-draft results on Qwen2.5-7B-Instruct $\times$
\textsc{ELI5}, $L{=}4$.}
\label{tab:full-qwen-eli5}
\resizebox{\textwidth}{!}{
\begin{tabular}{l c c c c c c c}
\toprule
Decoder & $B$ & AATPS & TR & ANLPPT-U & ANLPPT-Li & ANLPPT-PL & LPPL \\
\midrule
\textsc{MPFR}      & 2 & $2.563 \pm 0.007$ & $23.57 \pm 0.15$ & $0.104 \pm 0.001$ & $0.075 \pm 0.001$ & $0.108 \pm 0.001$ & $0.705 \pm 0.001$ \\
\textsc{MPFR}      & 4 & $2.918 \pm 0.008$ & $26.46 \pm 0.25$ & $0.104 \pm 0.001$ & $0.076 \pm 0.001$ & $0.108 \pm 0.001$ & $0.706 \pm 0.000$ \\
\textsc{MPFR}      & 6 & $3.109 \pm 0.005$ & $27.81 \pm 0.35$ & $0.103 \pm 0.001$ & $0.075 \pm 0.001$ & $0.107 \pm 0.001$ & $0.705 \pm 0.000$ \\
\textsc{MPFR}      & 8 & $3.235 \pm 0.010$ & $28.73 \pm 0.49$ & $0.104 \pm 0.002$ & $0.075 \pm 0.002$ & $0.108 \pm 0.001$ & $0.706 \pm 0.001$ \\
\textsc{Invariant} & 2 & $2.523 \pm 0.004$ & $22.73 \pm 0.22$ & $0.002 \pm 0.000$ & $0.002 \pm 0.000$ & $0.002 \pm 0.000$ & $0.747 \pm 0.005$ \\
\textsc{Invariant} & 4 & $2.850 \pm 0.010$ & $25.60 \pm 0.26$ & $0.002 \pm 0.000$ & $0.002 \pm 0.000$ & $0.002 \pm 0.000$ & $0.736 \pm 0.002$ \\
\textsc{Invariant} & 6 & $3.039 \pm 0.004$ & $27.01 \pm 0.32$ & $0.002 \pm 0.000$ & $0.002 \pm 0.000$ & $0.002 \pm 0.000$ & $0.728 \pm 0.002$ \\
\textsc{Invariant} & 8 & $3.156 \pm 0.009$ & $28.09 \pm 0.36$ & $0.002 \pm 0.000$ & $0.002 \pm 0.000$ & $0.002 \pm 0.000$ & $0.728 \pm 0.002$ \\
\bottomrule
\end{tabular}
}
\end{table}

\begin{table}[t]
\centering
\small
\color{black}
\caption{Multi-draft results on Llama-3.1-8B-Instruct $\times$ CNN/\textsc{DailyMail}, $L=4$.
Mean $\pm$ std over 3 random seeds, $1000$ prompts each. The drafter is
Llama-3.2-1B-Instruct, the official distillation of the target sharing the same tokenizer.}
\label{tab:multi_llama_cnn}
\small
\resizebox{\textwidth}{!}{
\begin{tabular}{l r r r r r r r}
\toprule
Decoder & $B$ & AATPS & TR & ANLPPT-U & ANLPPT-Li & ANLPPT-PL & LPPL \\
\midrule
\textsc{MPFR} & 2 & $3.459 \pm 0.015$ & $37.44 \pm 0.57$ & $0.058 \pm 0.001$ & $0.041 \pm 0.000$ & $0.061 \pm 0.001$ & $0.514 \pm 0.002$ \\
\textsc{MPFR} & 4 & $3.784 \pm 0.013$ & $40.20 \pm 0.28$ & $0.058 \pm 0.000$ & $0.041 \pm 0.000$ & $0.062 \pm 0.001$ & $0.514 \pm 0.003$ \\
\textsc{MPFR} & 6 & $3.935 \pm 0.005$ & $41.71 \pm 0.30$ & $0.058 \pm 0.000$ & $0.041 \pm 0.000$ & $0.061 \pm 0.001$ & $0.514 \pm 0.003$ \\
\textsc{MPFR} & 8 & $4.034 \pm 0.005$ & $42.04 \pm 0.81$ & $0.058 \pm 0.001$ & $0.041 \pm 0.000$ & $0.061 \pm 0.001$ & $0.514 \pm 0.003$ \\
\midrule
\textsc{Invariant} & 2 & $3.431 \pm 0.005$ & $36.79 \pm 0.37$ & $0.002 \pm 0.000$ & $0.002 \pm 0.000$ & $0.002 \pm 0.000$ & $0.514 \pm 0.003$ \\
\textsc{Invariant} & 4 & $3.741 \pm 0.007$ & $39.31 \pm 0.29$ & $0.002 \pm 0.000$ & $0.002 \pm 0.000$ & $0.002 \pm 0.000$ & $0.508 \pm 0.001$ \\
\textsc{Invariant} & 6 & $3.894 \pm 0.007$ & $40.51 \pm 0.25$ & $0.002 \pm 0.000$ & $0.002 \pm 0.000$ & $0.002 \pm 0.000$ & $0.507 \pm 0.001$ \\
\textsc{Invariant} & 8 & $3.987 \pm 0.010$ & $39.87 \pm 0.66$ & $0.002 \pm 0.000$ & $0.002 \pm 0.000$ & $0.002 \pm 0.000$ & $0.505 \pm 0.003$ \\
\bottomrule
\end{tabular}
}
\end{table}

\begin{table}[t]
\centering
\small
\color{black}
\caption{Multi-draft results on Llama-3.1-8B-Instruct $\times$ ELI5, $L=4$.
Mean $\pm$ std over 3 random seeds, $1000$ prompts each.}
\label{tab:multi_llama_eli5}
\small
\resizebox{\textwidth}{!}{
\begin{tabular}{l r r r r r r r}
\toprule
Decoder & $B$ & AATPS & TR & ANLPPT-U & ANLPPT-Li & ANLPPT-PL & LPPL \\
\midrule
\textsc{MPFR} & 2 & $3.103 \pm 0.006$ & $34.29 \pm 0.35$ & $0.135 \pm 0.001$ & $0.099 \pm 0.001$ & $0.139 \pm 0.001$ & $0.811 \pm 0.004$ \\
\textsc{MPFR} & 4 & $3.455 \pm 0.010$ & $37.79 \pm 0.38$ & $0.135 \pm 0.002$ & $0.099 \pm 0.001$ & $0.139 \pm 0.002$ & $0.812 \pm 0.003$ \\
\textsc{MPFR} & 6 & $3.635 \pm 0.001$ & $39.59 \pm 0.43$ & $0.135 \pm 0.002$ & $0.099 \pm 0.001$ & $0.139 \pm 0.002$ & $0.812 \pm 0.005$ \\
\textsc{MPFR} & 8 & $3.749 \pm 0.002$ & $39.76 \pm 0.80$ & $0.135 \pm 0.001$ & $0.098 \pm 0.000$ & $0.139 \pm 0.001$ & $0.811 \pm 0.004$ \\
\midrule
\textsc{Invariant} & 2 & $3.080 \pm 0.011$ & $33.72 \pm 0.60$ & $0.002 \pm 0.000$ & $0.002 \pm 0.000$ & $0.002 \pm 0.000$ & $0.803 \pm 0.003$ \\
\textsc{Invariant} & 4 & $3.408 \pm 0.006$ & $37.11 \pm 0.12$ & $0.002 \pm 0.000$ & $0.002 \pm 0.000$ & $0.002 \pm 0.000$ & $0.794 \pm 0.004$ \\
\textsc{Invariant} & 6 & $3.588 \pm 0.008$ & $39.12 \pm 0.31$ & $0.002 \pm 0.000$ & $0.002 \pm 0.000$ & $0.002 \pm 0.000$ & $0.795 \pm 0.008$ \\
\textsc{Invariant} & 8 & $3.700 \pm 0.008$ & $39.37 \pm 0.58$ & $0.002 \pm 0.000$ & $0.002 \pm 0.000$ & $0.002 \pm 0.000$ & $0.790 \pm 0.003$ \\
\bottomrule
\end{tabular}
}
\end{table}

\begin{table}[h]
\centering
\small
\caption{Multi-draft results on Qwen2.5-7B-Instruct $\times$
\textsc{CNN/DailyMail}, $L{=}4$.}
\label{tab:full-qwen-cnn}
\begin{tabular}{l c c c c c c c}
\toprule
Decoder & $B$ & AATPS & TR & ANLPPT-U & ANLPPT-Li & ANLPPT-PL & LPPL \\
\midrule
\textsc{MPFR}      & 2 & 2.799 & 23.18 & 0.052 & 0.037 & 0.055 & 0.475 \\
\textsc{MPFR}      & 4 & 3.124 & 24.89 & 0.052 & 0.037 & 0.055 & 0.474 \\
\textsc{MPFR}      & 6 & 3.298 & 25.39 & 0.052 & 0.037 & 0.055 & 0.476 \\
\textsc{MPFR}      & 8 & 3.410 & 25.17 & 0.053 & 0.037 & 0.055 & 0.476 \\
\textsc{Invariant} & 2 & 2.804 & 25.03 & 0.003 & 0.003 & 0.002 & 0.487 \\
\textsc{Invariant} & 4 & 3.114 & 26.63 & 0.002 & 0.002 & 0.002 & 0.481 \\
\textsc{Invariant} & 6 & 3.275 & 26.79 & 0.002 & 0.003 & 0.002 & 0.482 \\
\textsc{Invariant} & 8 & 3.387 & 26.38 & 0.002 & 0.003 & 0.002 & 0.479 \\
\bottomrule
\end{tabular}
\end{table}

\begin{table}[h]
\centering
\small
\caption{Multi-draft results on Qwen2.5-7B-Instruct $\times$ \textsc{ELI5},
$L{=}4$.}
\label{tab:full-qwen-eli5}
\begin{tabular}{l c c c c c c c}
\toprule
Decoder & $B$ & AATPS & TR & ANLPPT-U & ANLPPT-Li & ANLPPT-PL & LPPL \\
\midrule
\textsc{MPFR}      & 2 & 2.563 & 22.38 & 0.104 & 0.076 & 0.108 & 0.704 \\
\textsc{MPFR}      & 4 & 2.919 & 24.48 & 0.104 & 0.076 & 0.107 & 0.703 \\
\textsc{MPFR}      & 6 & 3.111 & 25.54 & 0.105 & 0.076 & 0.108 & 0.706 \\
\textsc{MPFR}      & 8 & 3.237 & 25.83 & 0.105 & 0.076 & 0.108 & 0.704 \\
\textsc{Invariant} & 2 & 2.519 & 23.18 & 0.002 & 0.002 & 0.002 & 0.750 \\
\textsc{Invariant} & 4 & 2.852 & 25.75 & 0.002 & 0.002 & 0.002 & 0.733 \\
\textsc{Invariant} & 6 & 3.043 & 27.11 & 0.002 & 0.002 & 0.002 & 0.730 \\
\textsc{Invariant} & 8 & 3.158 & 27.96 & 0.002 & 0.002 & 0.002 & 0.725 \\
\bottomrule
\end{tabular}
\end{table}

\begin{table}[h]
\centering
\small
\caption{Multi-draft results on Vicuna-7B-v1.5 $\times$
\textsc{CNN/DailyMail}, $L{=}4$.}
\label{tab:full-vicuna-cnn}
\begin{tabular}{l c c c c c c c}
\toprule
Decoder & $B$ & AATPS & TR & ANLPPT-U & ANLPPT-Li & ANLPPT-PL & LPPL \\
\midrule
\textsc{MPFR}      & 2 & 2.422 & 52.10 & 0.022 & 0.017 & 0.022 & 0.270 \\
\textsc{MPFR}      & 4 & 2.664 & 48.28 & 0.022 & 0.017 & 0.023 & 0.271 \\
\textsc{MPFR}      & 6 & 2.803 & 44.44 & 0.022 & 0.017 & 0.023 & 0.270 \\
\textsc{MPFR}      & 8 & 2.899 & 40.76 & 0.022 & 0.017 & 0.023 & 0.270 \\
\textsc{Invariant} & 2 & 2.478 & 58.15 & 0.002 & 0.002 & 0.002 & 0.287 \\
\textsc{Invariant} & 4 & 2.690 & 53.12 & 0.002 & 0.002 & 0.002 & 0.286 \\
\textsc{Invariant} & 6 & 2.831 & 48.41 & 0.002 & 0.003 & 0.002 & 0.283 \\
\textsc{Invariant} & 8 & 2.925 & 44.01 & 0.002 & 0.002 & 0.002 & 0.288 \\
\bottomrule
\end{tabular}
\end{table}

\begin{table}[!htbp]
\centering
\small
\caption{Multi-draft results on Vicuna-7B-v1.5 $\times$ \textsc{ELI5},
$L{=}4$.}
\label{tab:full-vicuna-eli5}
\begin{tabular}{l c c c c c c c}
\toprule
Decoder & $B$ & AATPS & TR & ANLPPT-U & ANLPPT-Li & ANLPPT-PL & LPPL \\
\midrule
\textsc{MPFR}      & 2 & 1.889 & 45.89 & 0.058 & 0.043 & 0.061 & 0.504 \\
\textsc{MPFR}      & 4 & 2.133 & 47.43 & 0.058 & 0.042 & 0.061 & 0.504 \\
\textsc{MPFR}      & 6 & 2.278 & 47.42 & 0.058 & 0.042 & 0.062 & 0.503 \\
\textsc{MPFR}      & 8 & 2.377 & 46.42 & 0.059 & 0.043 & 0.062 & 0.504 \\
\textsc{Invariant} & 2 & 1.884 & 50.29 & 0.002 & 0.002 & 0.002 & 0.555 \\
\textsc{Invariant} & 4 & 2.108 & 53.42 & 0.002 & 0.002 & 0.002 & 0.552 \\
\textsc{Invariant} & 6 & 2.248 & 55.08 & 0.002 & 0.003 & 0.002 & 0.546 \\
\textsc{Invariant} & 8 & 2.361 & 56.00 & 0.003 & 0.003 & 0.002 & 0.539 \\
\bottomrule
\end{tabular}
\end{table}

\paragraph{Speculative-acceptance parity.}
Within each cell, MPFR remains close to INVARIANT in AATPS across $B \in \{2,4,6,8\}$, with small deviations in both directions depending on the model--dataset cell. This indicates that the target-side keyed-race construction does not introduce a systematic per-step acceptance penalty. Token rate is reported as a hardware- and implementation-dependent diagnostic, since it also reflects batching, hashing, and watermark-bookkeeping overhead beyond acceptance behavior.

\paragraph{Stability of the watermark signal across $B$.}
ANLPPT-U, ANLPPT-Li, and ANLPPT-PL for \textsc{MPFR} remain within $\pm 0.001$
across $B\in\{2,4,6,8\}$ on every cell of
Tables~\ref{tab:full-qwen-cnn}--\ref{tab:full-vicuna-eli5}. The per-token
detection signal therefore does not dilute with the number of drafts, in
contrast to naive multi-draft watermarking, which would re-introduce the
efficiency--detectability trade-off of~\citet{hu2024inevitable} at the
multi-draft level.

\paragraph{Distortion-free audit.}
LPPL of \textsc{MPFR} outputs falls within $\pm 0.013$ nats/token of
\textsc{Invariant} on every cell. This empirical audit is consistent with the
exact marginal-preservation guarantee of the \textsc{MPFR} construction and
indicates that the watermark does not introduce a measurable shift in
target-model log-perplexity in the multi-draft regime.

\subsection{Drafter-substitution ablation: full results}
\label{app:ablation}

\paragraph{Drafter conditions.}
$D_0$: Qwen2.5-0.5B-Instruct at $T{=}1.0$ (default).
$D_1$: Qwen2.5-1.5B-Instruct at $T{=}1.0$ ($3{\times}$ scale, model swap).
$D_2$: Qwen2.5-0.5B-Instruct at $T{=}0.5$ (sharper drafter).
$D_3$: Qwen2.5-0.5B-Instruct at $T{=}1.5$ (more diffuse drafter).
The target temperature, watermark key, and all other parameters are held fixed
across $D_0$--$D_3$.

\paragraph{Detection statistic.}
For the Aaronson Gamma-tail test, the per-token score is
$s_t=-\log(1-U_t)$, where $U_t$ is the per-token uniform recovered from the
watermark code. Under the no-watermark null, the cumulative score
$S_T=\sum_{t\le T} s_t$ is $\mathrm{Gamma}(T,1)$; we declare detection at
FPR$=\alpha$ when $\Pr_{H_0}\!\left[\mathrm{Gamma}(T,1)>S_T\right]\le\alpha$,
computed via the regularised upper incomplete gamma. We report
$\alpha{=}1\%$ throughout.

\paragraph{Full detection across models and datasets.}
Figure~\ref{fig:full_detection_4cells} extends the default-drafter detection
comparison of Figure~\ref{fig:drafter_invariance}(a) to all four
$(\textsc{target},\textsc{dataset})$ cells of the main $2{\times}2$ result.
In every cell, \textsc{PFR}/\textsc{MPFR} (ours) coincide with the
strong-watermark ceiling defined by autoregressive \textsc{Basic-UWM} and
\textsc{MWS}, while \textsc{MSE} and \textsc{MSE-Pseudo}~\citep{he2026improving}
are uniformly weaker. The size of the gap is dataset-dependent --- largest on
\textsc{CNN/DailyMail}, smallest on the longer-form \textsc{ELI5} where every
scheme approaches saturation by $T_{\mathrm{eval}}{=}128$ --- but the ranking is
invariant. This confirms that the detection-strength conclusion drawn from
the anchor cell in Figure~\ref{fig:drafter_invariance}(a) is not specific to
a single $(\textsc{target},\textsc{dataset})$ pair.

\begin{figure}[h]
\centering
\includegraphics[width=\linewidth]{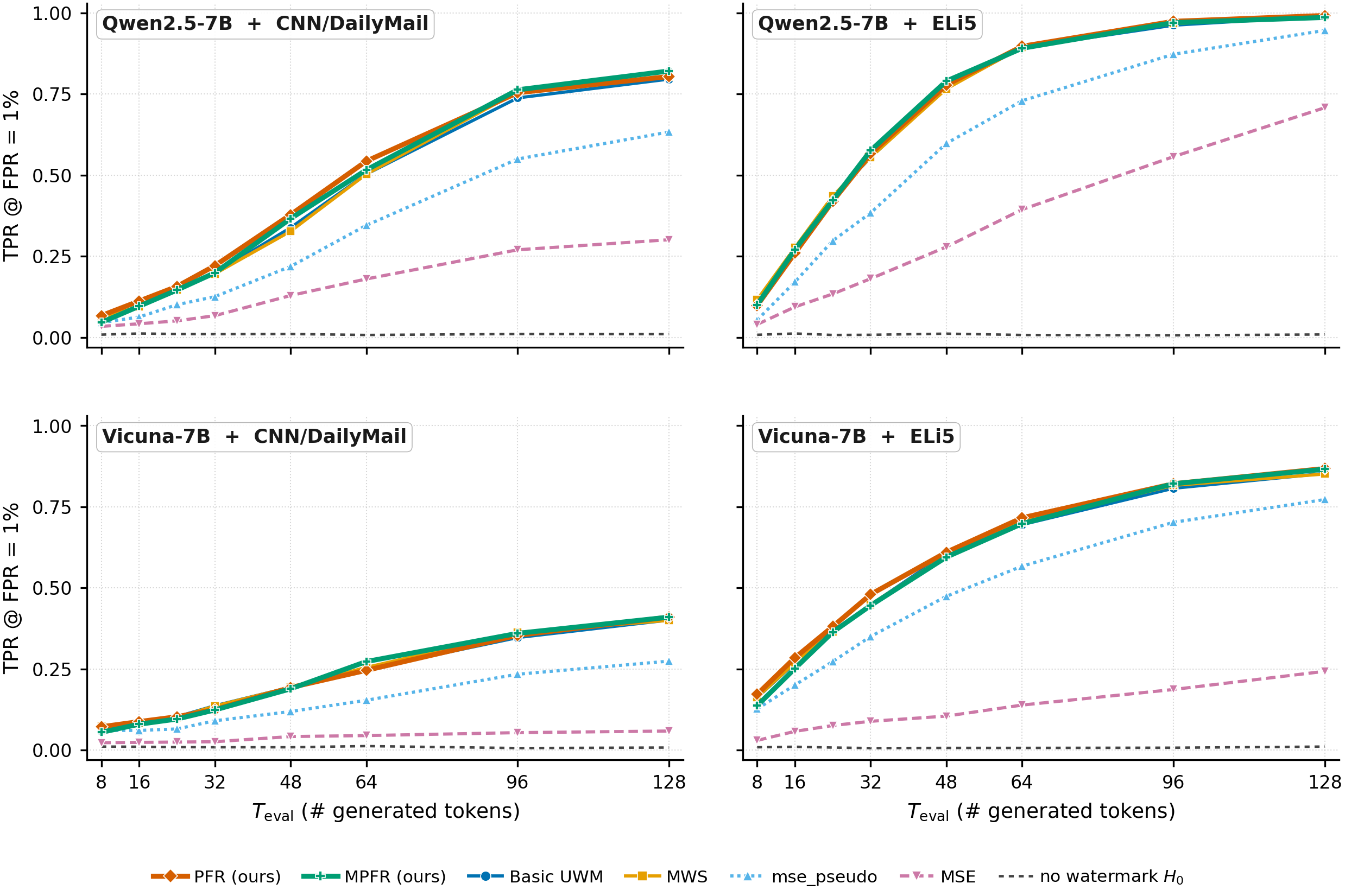}
\caption{\textbf{Full detection across $(\textsc{target},\textsc{dataset})$
cells.}
TPR@1\%FPR (Aaronson Gamma-tail) versus detection token budget
$T_{\mathrm{eval}}$ at the default drafter $D_0$, for the seven schemes
plotted in Figure~\ref{fig:drafter_invariance}(a). Each panel is one cell of
the main $2{\times}2$ result: target~$\in$
\{Qwen2.5-7B-Instruct, Vicuna-7B-v1.5\} $\times$ dataset~$\in$
\{\textsc{CNN/DailyMail}, \textsc{ELI5}\}. \textsc{PFR}/\textsc{MPFR} (ours)
match the strong-watermark ceiling (\textsc{Basic-UWM}, \textsc{MWS}) in
every cell, while \textsc{MSE} and \textsc{MSE-Pseudo} are uniformly weaker.
The empirical $H_0$ (no-watermark) floor sits near the nominal $1\%$ FPR,
calibrating the test.}
\label{fig:full_detection_4cells}
\end{figure}

\paragraph{Pairwise ROUGE-L under drafter substitution.}
For each prompt $p$ and decoder $d$ we have one realised output token sequence
per drafter condition. We compute ROUGE-L F1 over each pair $(D_i,D_j)$ of
drafter conditions on these token sequences, then report the per-prompt
minimum over the six pairs (worst pair) averaged across prompts
(Table~\ref{tab:exp2_pairwise_rouge}). An unbiased verify rule that depends
only on the target logits and the watermark key produces identical sequences
across all drafters, so this quantity should approach $1$.

\begin{table}[h]
\centering
\caption{\textbf{Output quality.}
LPPL$\downarrow$ under target.
ROUGE-L$\uparrow$ vs.\ gold \texttt{highlights}.
AATPS$\uparrow$ accepted target tokens per step.}
\label{tab:exp1_quality}
\small
\setlength{\tabcolsep}{4pt}
\resizebox{\textwidth}{!}{%
\begin{tabular}{l ccccccc | cc}
\toprule
 & \multicolumn{7}{c|}{\emph{Single-draft baselines}}
 & \multicolumn{2}{c}{\emph{Multi-draft ($B{=}4$)}} \\
\cmidrule(lr){2-8} \cmidrule(lr){9-10}
Metric
 & \textsc{VSpS}
 & \textsc{PFR-NoWM}
 & \textsc{Basic-UWM}
 & \textsc{MSE}
 & \textsc{MWS}
 & \textsc{MSE-Pseudo}
 & \textbf{\textsc{PFR}}
 & \textsc{Invariant}
 & \textbf{\textsc{MPFR}} \\
\midrule
LPPL$\downarrow$
 & 0.474 & 0.469 & 0.472 & 0.470 & 0.475 & 0.477 & \textbf{0.465}
 & 0.484 & \textbf{0.474} \\
ROUGE-L$\uparrow$
 & 0.219 & 0.222 & 0.223 & 0.221 & 0.223 & 0.219 & \textbf{0.223}
 & 0.222 & \textbf{0.222} \\
AATPS$\uparrow$
 & 2.51  & 2.46  & 1.00  & 2.53  & 2.43  & 2.52  & \textbf{2.46}
 & 3.12  & \textbf{3.13} \\
\bottomrule
\end{tabular}%
}
\end{table}

\begin{table}[h]
\centering
\caption{\textbf{Output stability under drafter substitution.}
Pairwise ROUGE-L F1 between the realised token sequences across each
$(D_i,D_j)$ pair, per-prompt minimum averaged over prompts. An unbiased
verify rule produces identical sequences across all drafters, so
ROUGE-L $\to 1$.}
\label{tab:exp2_pairwise_rouge}
\small
\setlength{\tabcolsep}{8pt}
\begin{tabular}{lc}
\toprule
Decoder & min-pair mean ROUGE-L \\
\midrule
\textsc{Basic-UWM} (no drafter, vacuous bound) & 1.000 \\
\textbf{\textsc{PFR}} (ours)                    & \textbf{0.990} \\
\textsc{MWS}                                    & 0.922 \\
\midrule
\textsc{PFR-NoWM} (no wm)                       & 0.436 \\
\textsc{MSE-Pseudo}                             & 0.426 \\
\textsc{MSE}                                    & 0.406 \\
\textsc{VSpS} (no wm)                           & 0.394 \\
\bottomrule
\end{tabular}
\end{table}

\section{More on Watermark}
\label{app::water}
% \textcolor{red}{Note in implementation, to compare to~\citep{hu2024inevitable}, we also need to only include positions whose seed/context code hasn't appeared before, i.e. ``skipping repeats'' just like them so that $\{U_t\}$ are (approximately) independent across $t$. }

In this section, we discuss different watermarks that are related to our scheme.
\subsection{Gumbel Watermark}
\label{subapp::gumbel_water}

We use a common notation for the scored positions throughout this subsection. 
Let $\mathcal{T}:=\{m+1,\ldots,n\}$ and $M:=|\mathcal{T}|=n-m$. 
For each $t\in\mathcal{T}$, let $r_t:\mathcal{V}\to(0,1)$ be the keyed random function, and define
\[
    U_t:=r_t(w_t),
    \qquad
    A_t:=-\log(1-U_t).
\]
Aaronson's Gumbel watermark score~\citep{aaronson2023watermarking} is
\[
    S_{\mathrm{A}}
    :=
    \sum_{t\in\mathcal{T}} A_t
    =
    \sum_{t\in\mathcal{T}} -\log(1-r_t(w_t)).
\]
Under the null hypothesis, i.e., when $y_{1:n}$ is not generated from the watermarked model, the variables $(U_t)_{t\in\mathcal{T}}$ are independent $\mathrm{Unif}(0,1)$ random variables. Hence $(A_t)_{t\in\mathcal{T}}$ are independent $\mathrm{Exp}(1)$ random variables, and
\[
    \mathbf{E}\left[S_{\mathrm{A}}\right] = M.
\]
If $y_{1:n}$ is generated from the Gumbel-watermarked model, then
\[
    \mathbf{E}\left[S_{\mathrm{A}}\right]
    \geq
    M
    +
    \left(\frac{\pi^2}{6}-1\right)
    \sum_{t\in\mathcal{T}}
    \mathbf{E}\left[
        \mathrm{Entropy}\big[p_t(\cdot)\big]
    \right],
\]
where $p_t=\mathrm{Softmax}(u_t/T)$ and the entropy is measured in nats.

Since $S_{\mathrm{A}}\sim\mathrm{Gamma}(M,1)$ under the null hypothesis, the corresponding $p$-value is
\[
    p_{\mathrm{A}}
    =
    \mathbf{P}\big(\mathrm{Gamma}(M,1)\ge S_{\mathrm{A}}\big)
    =
    \frac{\Gamma(M,S_{\mathrm{A}})}{\Gamma(M)},
\]
where $\Gamma(M)$ is the gamma function and $\Gamma(M,S_{\mathrm{A}})$ is the upper incomplete gamma function. 
The Average Negative Log P-value Per Token (ANLPPT) for Aaronson's score is
\[
    \mathrm{ANLPPT}_{\mathrm{A}}
    :=
    -\frac{1}{M}\log p_{\mathrm{A}}.
\]

Following the $U$-score framework~\citep{hu2024inevitable}, one may instead aggregate the keyed uniforms directly:
\[
    S_{\mathrm{U}}
    :=
    \sum_{t\in\mathcal{T}} U_t.
\]
Under the null hypothesis, $U_t\sim\mathrm{Unif}(0,1)$. Since for $U\sim\mathrm{Unif}(0,1)$,
\[
    \mathbf{E}\big[e^{\lambda U}\big]
    =
    \frac{e^\lambda-1}{\lambda},
    \qquad \lambda\ge 0,
\]
we have
\[
    \mathbf{E}\big[e^{\lambda S_{\mathrm{U}}}\big]
    =
    \left(\frac{e^\lambda-1}{\lambda}\right)^M.
\]
Chernoff's bound gives
\[
    p_{\mathrm{U}}
    \le
    \inf_{\lambda\ge 0}
    \exp\left(
        M\log\frac{e^\lambda-1}{\lambda}
        -
        \lambda S_{\mathrm{U}}
    \right),
\]
and the corresponding ANLPPT is
\[
    \mathrm{ANLPPT}_{\mathrm{U}}
    :=
    -\frac{1}{M}\log p_{\mathrm{U}}.
\]

\subsection{Connection to DeltaGumbel Reweighting}
\label{subapp::connection_delta_gumbel}

We now relate the above notation to the DeltaGumbel reweighting rule used in~\citep{hu2024inevitable}. 
For each scored position $t\in\mathcal{T}$, let
\[
    P_t:=P(\cdot\mid w_{:t-1})
\]
be the target next-token distribution. 
Define
\[
    g_t(v):=-\log\bigl(-\log r_t(v)\bigr),
    \qquad v\in\mathcal{V}.
\]
Then $(g_t(v))_{v\in\mathcal{V}}$ are independent standard Gumbel random variables. 
The DeltaGumbel rule selects
\[
    a_t^\star
    :=
    \arg\max_{v\in\mathcal{V}}
    \left\{\log P_t(v)+g_t(v)\right\}.
\]
Equivalently, for a fixed key, the reweighted distribution at position $t$ is the point mass
\[
    Q_t:=R_{g_t}(P_t)=\delta_{a_t^\star}.
\]

The likelihood-agnostic $U$-score of~\citep{hu2024inevitable} is
\[
    U_t
    :=
    \exp\bigl(-\exp(-g_t(w_t))\bigr).
\]
By the definition of $g_t$, this is exactly
\[
    U_t=r_t(w_t).
\]
Therefore, Aaronson's per-token score is a deterministic transform of the DeltaGumbel $U$-score:
\[
    A_t
    =
    -\log(1-r_t(w_t))
    =
    -\log(1-U_t).
\]
Thus
\[
    S_{\mathrm{A}}
    =
    \sum_{t\in\mathcal{T}} A_t
    =
    \sum_{t\in\mathcal{T}} -\log(1-U_t).
\]
In particular, if the full vector $(U_t)_{t\in\mathcal{T}}$ is retained, then Aaronson's score can be computed exactly from the $U$-scores, and conversely
\[
    U_t=1-\exp(-A_t).
\]
However, the aggregate statistics
\[
    \sum_{t\in\mathcal{T}} U_t
    \qquad\text{and}\qquad
    \sum_{t\in\mathcal{T}} -\log(1-U_t)
\]
are not equivalent in general, since they apply different nonlinear transformations before aggregation.

\subsection{Watermark Score by~\citep{li2025statistical}}

The score of~\citet{li2025statistical} is also computed from the same pivot
$r_t(w_t)$. For a parameter $\Delta$, define
\[
k_\Delta:=\left\lfloor\frac{1}{1-\Delta}\right\rfloor,
\qquad
q_\Delta:=1-k_\Delta(1-\Delta).
\]
Their per-token score is
\[
h^\star_{\mathrm{gum},\Delta}(r)
=
\log\left(
k_\Delta r^{\frac{\Delta}{1-\Delta}}
+
\mathds{1}\{q_\Delta>0\}
r^{\frac{1-q_\Delta}{q_\Delta}}
\right),
\]
and the corresponding statistic is
\[
S_{\mathrm{Li},\Delta}
:=
\sum_{t\in\mathcal{T}}
h^\star_{\mathrm{gum},\Delta}(r_t(w_t)).
\]

We now discuss another watermark score for Aaronson's Gumbel Watermark, proposed by~\citep{li2025statistical}, which is a statistically optimal score function for the Gumbel watermark under a class-dependent efficiency criterion. 
We introduce their result using our notation as follows. 
Their Gumbel pivot is exactly
\[
Y_t^{\mathrm{gum}}
:=
r_t(w_t).
\]
Under the null hypothesis,
\[
Y_t^{\mathrm{gum}}\sim\mathrm{Unif}[0,1].
\]
Under the Gumbel-watermarked alternative, conditional on
\[
P_t:=P(\cdot|w_{:t-1}),
\]
the pivot satisfies
\[
\mathbf{P}_{H_1}\!\left(Y_t^{\mathrm{gum}}\leq r| P_t\right)
=
\sum_{v\in\mathcal{V}}P_t(v)r^{1/P_t(v)},
\qquad r\in[0,1],
\]
where terms with $P_t(v)=0$ are omitted.

Aaronson's original score corresponds to the per-token score function
\[
h_{\mathrm{A}}(r):=-\log(1-r),
\]
so that
\[
S_{\mathrm{A}}
=
\sum_{t\in\mathcal{T}}h_{\mathrm{A}}(r_t(w_t)).
\]
In contrast, Li et al.~\citep{li2025statistical} consider the distribution class
\[
\mathcal{P}_{\Delta}
:=
\left\{
P:\max_{v\in\mathcal{V}}P(v)\leq 1-\Delta
\right\}
\]
and derive the optimal Gumbel score
\[
h^\star_{\mathrm{gum},\Delta}(r)
=
\log\left(
k_\Delta r^{\frac{\Delta}{1-\Delta}}
+
\mathds{1}\{q_\Delta>0\}
r^{\frac{1-q_\Delta}{q_\Delta}}
\right),
\]
where
\[
k_\Delta
:=
\left\lfloor\frac{1}{1-\Delta}\right\rfloor,
\qquad
q_\Delta
:=
1-k_\Delta(1-\Delta).
\]
The corresponding statistic in our notation is
\[
S_{\mathrm{Li},\Delta}
:=
\sum_{t\in\mathcal{T}}
h^\star_{\mathrm{gum},\Delta}(r_t(w_t)).
\]
This score is the log-likelihood ratio associated with the least-favorable distribution
\[
P_\Delta^\star
=
(\underbrace{1-\Delta,\ldots,1-\Delta}_{k_\Delta\text{ times}},
q_\Delta,0,\ldots,0),
\]
with the $q_\Delta$ term omitted when $q_\Delta=0$.
Therefore, $h^\star_{\mathrm{gum},\Delta}$ is not the same as Aaronson's original score $h_{\mathrm{A}}$.
Indeed,
\[
h_{\mathrm{A}}(r)\to\infty
\qquad\text{as }r\uparrow 1,
\]
whereas
\[
h^\star_{\mathrm{gum},\Delta}(r)
\to
\log\left(k_\Delta+\mathds{1}\{q_\Delta>0\}\right)
\qquad\text{as }r\uparrow 1.
\]
Thus, Li et al.'s score and Aaronson's original score are based on the same Gumbel pivot
$r_t(w_t)$, but they lead to different cumulative detection statistics.

\subsection{Watermark Score by~\citep{lattimore2026refined}}

\citet{lattimore2026refined}
proposed a refined detector for the Gumbel watermark based on a truncated power-law score.
In our notation, the same Gumbel pivot is
\[
r_t(w_t),
\qquad t\in\mathcal{T}.
\]
Instead of Aaronson's original per-token score
\[
h_{\mathrm{A}}(r):=-\log(1-r),
\]
Lattimore~\citep{lattimore2026refined} considers
\[
h_{\mathrm{PL},\epsilon}(r)
=
\min\left\{
\epsilon^{-1/2},
(1-r)^{-1/2}
\right\}
-(2-\sqrt{\epsilon}),
\]
and the corresponding statistic is
\[
S_{\mathrm{PL},\epsilon}
:=
\sum_{t\in\mathcal{T}}
h_{\mathrm{PL},\epsilon}(r_t(w_t)).
\]
Here the centering term $2-\sqrt{\epsilon}$ is chosen so that, under the null hypothesis,
\[
\mathbf{E}\!\left[h_{\mathrm{PL},\epsilon}(U)\right]=0,
\qquad U\sim\mathrm{Unif}[0,1].
\]
Thus, under the null and assuming independence across the scored positions,
\[
\mathbf{E}[S_{\mathrm{PL},\epsilon}]=0.
\]
Both $S_{\mathrm{PL},\epsilon}$ and Aaronson's statistic
\[
S_{\mathrm{A}}
=
\sum_{t\in\mathcal{T}}h_{\mathrm{A}}(r_t(w_t))
\]
are model-agnostic statistics based on the same Gumbel pivot $r_t(w_t)$.
However, they are not equivalent. In particular,
\[
h_{\mathrm{A}}(r)\to\infty
\qquad\text{as }r\uparrow 1,
\]
whereas
\[
h_{\mathrm{PL},\epsilon}(r)
\leq
\epsilon^{-1/2}-(2-\sqrt{\epsilon})
\]
is bounded because of the truncation. Therefore, Lattimore's truncated power-law score and Aaronson's original score use the same keyed randomness, but lead to different cumulative detection statistics.

\section{Algorithm Implementation}
\label{app::alg_imple}

% \begin{algorithm}[tbp]
\begin{algorithm}[H]
\SetAlgoLined
\DontPrintSemicolon
\SetKwInOut{Input}{Input}\SetKwInOut{Output}{Output}

\Input{output sequence $w_{1:N}$, prompt length $n_0$, labeler $\Lambda$, secret key $\mathsf{k}$, reference measure $\mu$}
\Output{Aaronson score $S_{\mathrm A}$, $p$-value $p_{\mathrm A}$, ANLPPT $\mathrm{ANLPPT}_{\mathrm A}$ and $M = |\mathcal{T}|$}

$S_{\mathrm A} \gets 0$\;
$M \gets 0$\;

\For{$t=n_0+1$ \KwTo $N$}{
    $c_t \gets w_{:t-1}$\;
    $\ell_t \gets \Lambda(c_t)$\;
    $\Pi_t \gets G(\mathsf{k}, \ell_t)$\;

    \tcp{recover keyed uniform value for the observed token $w_t$}
    $\tau_t(w_t) \gets \textsc{FirstArrival}(w_t,\Pi_t,\mu)$\;
    $r_t \gets \exp(-\mu(w_t)\tau_t(w_t))$\;

    $S_{\mathrm A} \gets S_{\mathrm A} - \log(1-r_t)$\;
    $M \gets M + 1$\;
}

\tcp{Under $H_0$, $S_{\mathrm A}\sim \mathrm{Gamma}(M,1)$}
$p_{\mathrm A} \gets \Gamma(M,S_{\mathrm A})/\Gamma(M)$\;

$\mathrm{ANLPPT}_{\mathrm A} \gets -\frac{1}{M}\log p_{\mathrm A}$\;

\caption{Compute Aaronson score, exact null $p$-value, and ANLPPT for keyed PFR watermark}
\label{alg:pfr_aaronson_anlppt}
\end{algorithm}

% \begin{algorithm}[tbp]
\begin{algorithm}[H]
\SetAlgoLined
\DontPrintSemicolon
\SetKwIF{If}{ElseIf}{Else}{if}{:}{else if}{else}{end}
\SetKwInOut{Input}{Input}\SetKwInOut{Output}{Output}

\Input{lookahead $K$, output length $N$, target model family $P$, draft model family $Q$,
initial prefix $w_{1:n}$, common reference measure $\mu$ on $\mathcal{Z}$ such that
$P(\cdot | c),Q(\cdot | c)\ll\mu$ for all contexts $c$,
labeler $\Lambda$, secret key $\mathsf{k}$}

\While{$n < N$}{
    $h \gets w_{1:n}$ \tcp*[r]{freeze block-start prefix}
    $L \gets \min\{K,\,N-n\}$\;

    \For{$s=1$ \KwTo $L$ \tcp*[r]{draft $L$ tokens from shared PFR sources}}{
        $c_s \gets h \,\|\, \tilde w_{1:s-1}$\;
        $\ell_s \gets \Lambda(c_s)$ \tcp*[r]{includes repeated-context masking if used}
        $\Pi_s \gets G(\mathsf{k}, \ell_s)$\;
        $\tilde w_s \gets \textsc{PFR}(Q(\cdot | c_s), \Pi_s, \mu)$\;
    }

    \If{$n+L < N$}{
        $c_{L+1} \gets h \,\|\, \tilde w_{1:L}$\;
    }

    \tcp{Prepare target-side evaluations for $c_1,\ldots,c_L$ and, if needed, $c_{L+1}$ in parallel}
    $accepted \gets \textbf{true}$\;

    \tcp{Verify each draft token using the same shared Poisson source}
    \For{$s=1$ \KwTo $L$}{
        $y_s \gets \textsc{PFR}(P(\cdot | c_s), \Pi_s, \mu)$\;
        $w_{n+1} \gets y_s$; $n \gets n+1$ \tcp*[r]{emit the target PFR winner}
        \If{$y_s \neq \tilde w_s$}{
            $accepted \gets \textbf{false}$\;
            \textbf{break}\; \tcp*[r]{correction: later draft contexts are invalid}
        }
    }

    \If{$accepted$ \textbf{and} $n < N$}{
        $\ell_{L+1} \gets \Lambda(c_{L+1})$\;
        $\Pi_{L+1} \gets G(\mathsf{k}, \ell_{L+1})$\;
        $w_{n+1} \gets \textsc{PFR}(P(\cdot | c_{L+1}), \Pi_{L+1}, \mu)$\;
        $n \gets n+1$ \tcp*[r]{bonus token}
    }
}
\caption{Watermarkable Speculative Sampling via Poisson Processes}
\label{alg:fwss_pfr}
\end{algorithm}

% \begin{algorithm}[tbp]
\begin{algorithm}[H]
\SetAlgoLined
\DontPrintSemicolon
\SetAlgoNoEnd
\SetKwIF{If}{ElseIf}{Else}{if}{:}{else if}{else}{end}
\SetKwInOut{Input}{Given}

\Input{lookahead $L$, number of drafts $B$, output length $N$, target model $P$,
draft model $Q$, initial prompt $w_{1:n}$, keyed Poisson source generator $G$,
and secret key $\mathsf{k}$}

\While{$n < N$}{
    $h \gets w_{1:n}$; $\ell \gets \min\{L,\,N-n\}$\;
    $\mathcal{C}_0 \gets \{h\}$; $\nu(h) \gets B$
    \tcp*[r]{root context; all $B$ homogeneous drafts start here}

    \For{$s=1$ \KwTo $\ell$}{
        $\mathcal{C}_s \gets \emptyset$\;
        \For{$c \in \mathcal{C}_{s-1}$}{
            $\Pi(c) \gets G(\mathsf{k}, c)$
            \tcp*[r]{context-indexed keyed source shared by draft and target at context $c$}
            \For{$i=1$ \KwTo $\nu(c)$}{
                $\tilde w(c,i) \gets \textsc{MPFR}(Q(\cdot | c), \Pi(c), i)$\;
            }
            $\mathcal{D}(c) \gets \{\tilde w(c,1),\ldots,\tilde w(c,\nu(c))\}$
            \tcp*[r]{distinct next-token proposals from context $c$}
            \For{$u \in \mathcal{D}(c)$}{
                $\nu(c \,\|\, u) \gets
                \big|\{\,i\in\{1,\ldots,\nu(c)\}: \tilde w(c,i)=u\,\}\big|$\;
                \tcp*[r]{how many drafts move to child context $c\|u$}
                $\mathcal{C}_s \gets \mathcal{C}_s \cup \{\,c \,\|\, u\,\}$\;
            }
        }
    }

    \For(\tcp*[f]{compute target-side winner at every internal context}){$c \in \mathcal{C}_0 \cup \cdots \cup \mathcal{C}_{\ell-1}$}{
        $y(c) \gets \textsc{MPFR}(P(\cdot | c), \Pi(c), 1)$\;
        \tcp*[r]{target-side mapped-Poisson winner at context $c$}
    }

    $c^\star \gets h$; $accepted \gets \textbf{true}$\;
    \tcp*[r]{start following the unique realized target path from the root}

    \For{$s=1$ \KwTo $\ell$}{
        $a_s \gets \mathds{1}\!\left\{y(c^\star)\in \mathcal{D}(c^\star)\right\}$\;
        \tcp*[r]{accept if the target winner is among the drafts at the current context}
        $w_{n+1} \gets y(c^\star)$; $n \gets n+1$\;
        \tcp*[r]{emit the target-side token, not a draft token}
        \If{$a_s = 0$}{
            $accepted \gets \textbf{false}$ and \textbf{break}
            \tcp*[r]{first rejection ends the speculative block}
        }
        \If{$n = N$}{
            \textbf{break}
        }
        $c^\star \gets c^\star \,\|\, y(c^\star)$
        \tcp*[r]{move to the next realized context}
    }

    \If(\tcp*[f]{if all $\ell$ steps were accepted, emit one extra target token}){$accepted$ \textbf{and} $n < N$}{
        $\Pi^\star \gets G(\mathsf{k}, c^\star)$\;
        $w_{n+1} \gets \textsc{MPFR}(P(\cdot | c^\star), \Pi^\star, 1)$\;
        $n \gets n+1$ \tcp*[r]{bonus token}
    }
}
\caption{Watermarkable Multi-Draft Speculative Sampling (Formal)}
\label{alg:multi_fwss_pfr_formal}
\end{algorithm}

% {\color{orange}
% \section{Questions to be considered} 

% \begin{itemize}
%     \item Is it possible, and if possible why, that an imporved version of Gumbel watermarking improves our speculative sampling? Related to~\citep{huangwatermarking}? 

%     \item Research question: if advance to multi-draft speculative sampling, what is the watermark rule, detection rule and detectability? Any theoretic promise, generalizing Aaronson's single-draft version? 

%     \item \citep[Theorem 3.2]{he2026improving} gives $\mathsf{WS}(P_\zeta) = H(P)$ if and only if $H(P_\zeta) = 0$, is this \emph{golden formula}? 

%     \item \textcolor{green}{How do we justify that we break the inevitable boundary????}

%     \item \textcolor{green}{How do use theory/experiment to discuss draft-invariance and watermark strength????}

%     \item \textcolor{green}{New paper on Aaronson's score?????}

%     \item \textcolor{green}{Emphasize more on the idea comes from Gumbel Watermark?????}

%     \item \textcolor{green}{Robustness via PML?????}

%     \item Why our sampling and watermarking do NOT conflict? Either \citep{hu2024inevitable}'s proof is wrong in our setup or we use the KL score in \citep{he2026improving}. 

%     \item If we use the KL score, we need to plot our fundamental trade-off???? 
% \end{itemize}
% }

\end{document}